\documentclass[a4paper,twocolumn,11pt, accepted=2024-12-21]{quantumarticle}
\pdfoutput=1
\usepackage[utf8]{inputenc}
\usepackage[english]{babel}
\usepackage{url}
\usepackage[normalem]{ulem}
\usepackage{graphicx}% Include figure files
\usepackage{dcolumn}% Align Table columns on decimal point
\usepackage{bm}% bold math
\usepackage[boldmath]{numprint}
\usepackage{makecell}
\usepackage{multirow}
\usepackage[T1]{fontenc}
\usepackage{amsmath}
\usepackage{amsthm}
\usepackage{amsfonts}
\usepackage{amssymb}
\usepackage{hyperref}

\usepackage{soul}
\usepackage{subfig}
\usepackage{graphicx}
\usepackage{dsfont}

\usepackage{tikz}
\usepackage{lipsum}
\usepackage{boldline}
\usepackage{physics}
\usepackage{comment}
\usepackage{titlesec}
\usepackage{mathtools}
\usepackage{color,soul}
\DeclarePairedDelimiter\ceil{\lceil}{\rceil}

\newcommand{\etal}{\textit{et al.~}}
\usepackage{graphicx}
\usepackage{comment}
\usepackage{algorithm}
\usepackage{algpseudocode}
\usepackage{multicol}
\usepackage{tabularx}
\usepackage[makeroom]{cancel}
\newtheorem{theorem}{Theorem}[section]
\newtheorem{lemma}{Lemma}[section]

\begin{document}

\title{Des-q: a quantum algorithm to provably speedup retraining of decision trees}

\author{Niraj Kumar$^{\dagger*}$}
\affiliation{Global Technology Applied Research, JPMorgan Chase, New York, NY 10017 USA}
\email{$\dagger$niraj.x7.kumar@jpmchase.com.}
\thanks{$*$: NK and RY contributed equally in this work.}
\author{Romina Yalovetzky$^{*}$}
\affiliation{Global Technology Applied Research, JPMorgan Chase, New York, NY 10017 USA}
\author{Changhao Li}
\affiliation{Global Technology Applied Research, JPMorgan Chase, New York, NY 10017 USA}
\author{Pierre Minssen}
\affiliation{Global Technology Applied Research, JPMorgan Chase, New York, NY 10017 USA}
\author{Marco Pistoia}
\affiliation{Global Technology Applied Research, JPMorgan Chase, New York, NY 10017 USA}
\maketitle

\begin{abstract}
Decision trees are widely adopted machine learning models due to their simplicity and explainability. However, as training data size grows, standard methods become increasingly slow, scaling polynomially with the number of training examples. In this work, we introduce Des-q, a novel quantum algorithm to construct and retrain decision trees for regression and binary classification tasks. Assuming the data stream produces small, periodic increments of new training examples, Des-q significantly reduces the tree retraining time. Des-q achieves a logarithmic complexity in the combined total number of old and new examples, even accounting for the time needed to load the new samples into quantum-accessible memory. Our approach to grow the tree from any given node involves performing piecewise linear splits to generate multiple hyperplanes, thus partitioning the input feature space into distinct regions. To determine the suitable anchor points for these splits, we develop an efficient quantum-supervised clustering method, building upon the q-means algorithm introduced by Kerenidis \etal We benchmark the simulated version of Des-q against the state-of-the-art classical methods on multiple data sets and observe that our algorithm exhibits similar performance to the state-of-the-art decision trees while significantly speeding up the periodic tree retraining.
\end{abstract}

\section{Introduction}
We are currently in the era of \emph{big-data} where organizations collect vast amounts of daily user data \cite{sagiroglu2013big}, be it with retail chains logging transactions, telecommunication companies connecting millions of calls and messages, or large banks processing millions of ATM and credit card transactions daily \cite{zikopoulos2011understanding, smith2012big}. To tackle the data surge, entities aim to build efficient, accurate machine learning models and swiftly retrain them to avoid performance degradation. As the Internet expands, data volume grows, challenging current methods.

Despite the success of deep neural network models \cite{samek2021explaining} for these tasks, decision tree-based models for regression and classification remain competitive and widely adopted \cite{zhou2017deep, chen2019robust, xu2005decision}. A decision tree is a logical model comprising a hierarchical set of rules for predicting target labels based on training examples. They are favored for their model explainability, low parameter requirements, noise robustness, and efficiency with large-scale datasets, offering a cost-effective alternative to other supervised learning models \cite{sharma2016survey}. Despite its practical success, the construction of an optimal decision tree from a given training dataset is considered to be a hard problem \cite{hancock1996lower, hyafil1976hyafil}. Consequently, heuristics such as greedy methods have long been used to construct trees. These top-down or bottom-up recursive inducers seek to find the best local feature and threshold value to partition the data. Popular methods following this approach include CART \cite{breiman2017classification,Lewis2000AnIT,buntine1992learning}, ID3 \cite{quinlan1986induction}, and C4.5 \cite{quinlan2014c4}. Building a decision tree model using these methods for a training dataset with $N$ examples and $d$ features incurs a time complexity that scales polynomially in both $N$ and $d$.

In practical applications, decision trees require periodic updates with batches of new labeled data to prevent model deterioration over time. These updates fall into two categories: \emph{incremental decision tree} methods, which update the tree using only new data \cite{utgoff1989incremental, domingos2000mining, manapragada2018extremely}, and \emph{tree retraining} methods, which rebuild the tree using both old and new data. Incremental methods offer quick updates but struggle with concept drift, where model performance degrades as data distributions change \cite{tsymbal2004problem}. Therefore, retraining the entire tree is often more effective, though existing methods make it time-consuming and resource-intensive, posing challenges in big-data environments.

Parallel to classical development, efforts have aimed to build quantum decision trees using techniques like Grover's search algorithm for potential speedups \cite{lu2014quantum, Khadiev2019arxiv, Beigi2020Quantum, Heese2022Quantum}. Despite these attempts, efficient tree construction and retraining remain challenging, achieving at best a quadratic speedup in the number of features $d$ over classical methods. However, these methods do not offer speedups in the number of training examples $N$. Since in most practical cases the number of training examples greatly exceeds the number of features, a significant speedup is only meaningful if achieved over $N$.
This limitation becomes more pronounced with periodic data accumulation, where timely model retraining is essential to prevent performance degradation. Therefore, a critical question arises: \emph{once the decision tree model is built and put online, is it possible to achieve a poly-logarithmic in $N$ run-time dependence in the decision-tree retraining?}\footnote{This implies an exponential speedup in tree retraining over classical greedy methods, although further improvements could occur in classical methods from quantum-inspired classical techniques \cite{tang2019quantum, doriguello2023you}.} Achieving this goal would ensure the model remains consistently online, with rapid retraining capabilities and sustained performance quality. 

In this work, we answer this question affirmatively by proposing a novel quantum algorithm called Des-q for decision tree construction and retraining for regression and binary classification. The initial tree construction takes time polynomial in $N$ and $d$ solely due to the time to load the initial data into quantum-accessible-data structure\footnote{By quantum-accessible data structure, we mean a classical data structure that allows accessing quantum states of the input data in poly-logarithmic time relative to the input size. Constructing this data structure requires polynomial time relative to the input size.}. In this work, we utilize the KP-tree data structure \cite{kerenidis2016quantum}, which, once constructed, allows for querying data as quantum superposition states in polylogarithmic time relative to $N$ and $d$. Moreover, this permits the rest of the algorithmic steps of Des-q to also run in polylogarithmic time. Although tree construction with Des-q is initially slow due to KP-tree creation for data with $N$ examples, subsequent retrains with small batches of new data become extremely fast. This process simply involves updating the KP-tree with the new small batch of data while the remaining algorithmic steps are executed in poly-logarithmic time relative to $N.$

The rest of the paper is structured as follows. We review the related works in Sec~\ref{sec:related_work}, while Sec~\ref{sec:motivation} provides a key motivation for Des-q including a high-level sketch of its methodology. Our key results are summarized in Sec~\ref{sec:sum_res}. Sec~\ref{sec:qds} provides the detailed algorithmic steps along with key theorems for tree construction with Des-q, while Sec~\ref{sec:periodic_update} provides the steps for tree retraining. We provide the benchmark for simulated Des-q in Sec~\ref{sec:numerics} and finally conclude in Sec~\ref{sec:conclusion}. We also highlight some ingredients we utilize from existing quantum algorithms in Appendix~\ref{sec:ingre}.

\subsection{Terminology used in the paper} \label{sec:term}

\begin{enumerate}
    \item \emph{Root node}: no incoming edge, zero or more outgoing edges; \emph{internal node}: one incoming edge, two (or more) outgoing edges; \emph{leaf node}: each leaf node is assigned a class label; \emph{parent and child node}: if a node is split, we refer to that given node as the parent node, and the resulting nodes are called child nodes, respectively
    \item Throughout this work, we refer to
    \begin{itemize}
        \item $N$: The number of training examples in the dataset
        \item $d$: The number of attributes/features in the dataset 
        \item $D$: Maximum Tree height/depth
        \item $k$: Number of clusters generated at each clustering step
        \item $K$: Maximum allowed number of iterations in each clustering step
        \item $N_{\text{new}}$: Number of new training examples for periodic model retrain
    \end{itemize}
    \item The training dataset is referred as $\textsf{Data} = \{X, Y\}$, where $X = [x_1,\cdots,x_N]^T \in \mathbb{R}^{N \times d}$ and $Y = [y_1,\cdots, y_N]^{T}$. $X$ contains the $N$ training examples $x_i \in \mathbb{R}^d$, i.e., consisting of $d$ features with only numerical values, and $Y$ is the target label with each $y_i \in \mathbb{R}$ for the task of regression, while $y_i \in \mathcal{M} = \{0, 1\}$ for the task of binary classification. 
\end{enumerate}

\section{Related work} \label{sec:related_work} 

Beginning with classical decision trees, the state-of-the-art decision tree methods to handle both classification and regression tasks include CART \cite{breiman2017classification,Lewis2000AnIT}, ID3 \cite{quinlan1986induction} and C4.5 \cite{quinlan2014c4}. It is worth noting that, decision trees built using these methods for a dataset with $N$ instances and $d$ features typically involve a computational complexity of $O(\emph{poly} (Nd))$. Beyond heuristic methods, approaches to train optimal decision trees using mixed-integer programs (MIP) have been proposed \cite{zhu2020scalable, d2024margin}. Alternative techniques for constructing decision trees involve utilizing clustering mechanisms have also been proposed \cite{liu2020novel, berzal2004building}. 

On the quantum approaches to construct tree, the work by Lu {\em et~al.}\cite{Lu2013QIP} encodes the training data into quantum states and leverages Grover's search algorithm to build the tree. In the work by Khadiev \etal\cite{Khadiev2019arxiv}, the classical decision tree induction algorithm $C5.0$ is extended to a quantum version using Grover search-based algorithm. The quantum algorithm has a run-time complexity that is nearly quadratic better than the classical algorithm in terms of feature number $d$. Heese \etal \cite{Heese2022Quantum} introduce a quantum representation of binary classification trees called \textit{Q-tree}. The tree is constructed through probabilistic traversal via quantum circuit measurements.

We also remark on the two recent papers on improving the existing $q$-means clustering algorithm \cite{kerenidis2019q} whose modified version is a component of our Des-q algorithm. The first paper by Jaiswal \cite{jaiswal2023quantum} where the author provides a quantum algorithm for $k$-means problem with provable convergence guarantees, and Doriguello \etal \cite{doriguello2023you} who improve upon the existing $q$-means algorithm while still being heuristic. 

\section{Motivation of our method} \label{sec:motivation}

As highlighted by Ho \etal \cite{ho1998random}, numerous approaches exist for building decision trees from training examples. Most methods use linear splits at each internal node to create branches, partitioning examples with one or more hyperplanes into distinct regions of the feature space. These linear splits fall into three categories: \emph{axis-parallel linear split, oblique linear split, and piecewise linear split}.

Axis-parallel linear splitting divides data based on a specific feature $l$ and a threshold $\theta_l$, determined by minimizing an impurity metric like Gini impurity, entropy (for classification), or label variance (for regression). Subsequently, each example $x_i \in \mathbb{R}^d$ , $i \in [N]$ is assigned to a branch based on whether its feature value $x_{il}$ exceeds $\theta_l$. If $x_{il} \leq \theta_l$, the example goes to one branch; otherwise, it goes to the other. Popular decision tree methods like ID3, C4.5, and CART use this splitting mechanism.

The oblique linear split method divides data by finding a hyperplane in the feature space, not necessarily parallel to any axis. Each example is assigned based on which side of the hyperplane it lies on, determined by the condition
$$ \sum_{l=1}^d a_l x_{il} \leq \theta,$$
where $a_l$ is the coefficient of the $l$-th feature and $\theta$ is the threshold. This method often produces lower-depth trees with better generalization but finding the optimal hyperplane is challenging. Proposed solutions include using simulated annealing or hill-climbing to approximate near-optimal hyperplanes \cite{buntine1992learning, bucy1993decision}.  

The piecewise linear split method extends the oblique linear split by using multiple hyperplanes to create $k(\geq 2)$ branches or regions. Each hyperplane is defined by anchor points, and examples are assigned based on their proximity to these points. Although selecting these hyperplanes is more complex than the oblique method, it can significantly reduce tree depth and improve generalization, especially when multiple features are strongly correlated.

In particular, Liu \etal \cite{liu2020novel} proposed using the piecewise linear split technique to create a Multi-Split Decision Tree (MSDT) for multi-class classification. In MSDT, each branch acts as a cluster with an anchor point as its centroid. They first calculate feature weights using the RELIEF-F algorithm \cite{kononenko1994estimating} to determine feature relevance. Then, they use weighted $k$-means \cite{lloyd1982least} to cluster examples, forming $k$ clusters. This weighted approach ensures nodes generated by the supervised $k$-means minimize node class impurity by giving more weight to features strongly correlated with the target label, aligning with the decision tree's objective to minimize class impurity or variance. 

Our approach to constructing decision trees builds on a similar concept, involving piecewise linear splits using a supervised clustering quantum algorithm.

\begin{figure*}[t]%
    \centering
    \includegraphics[width=16cm]{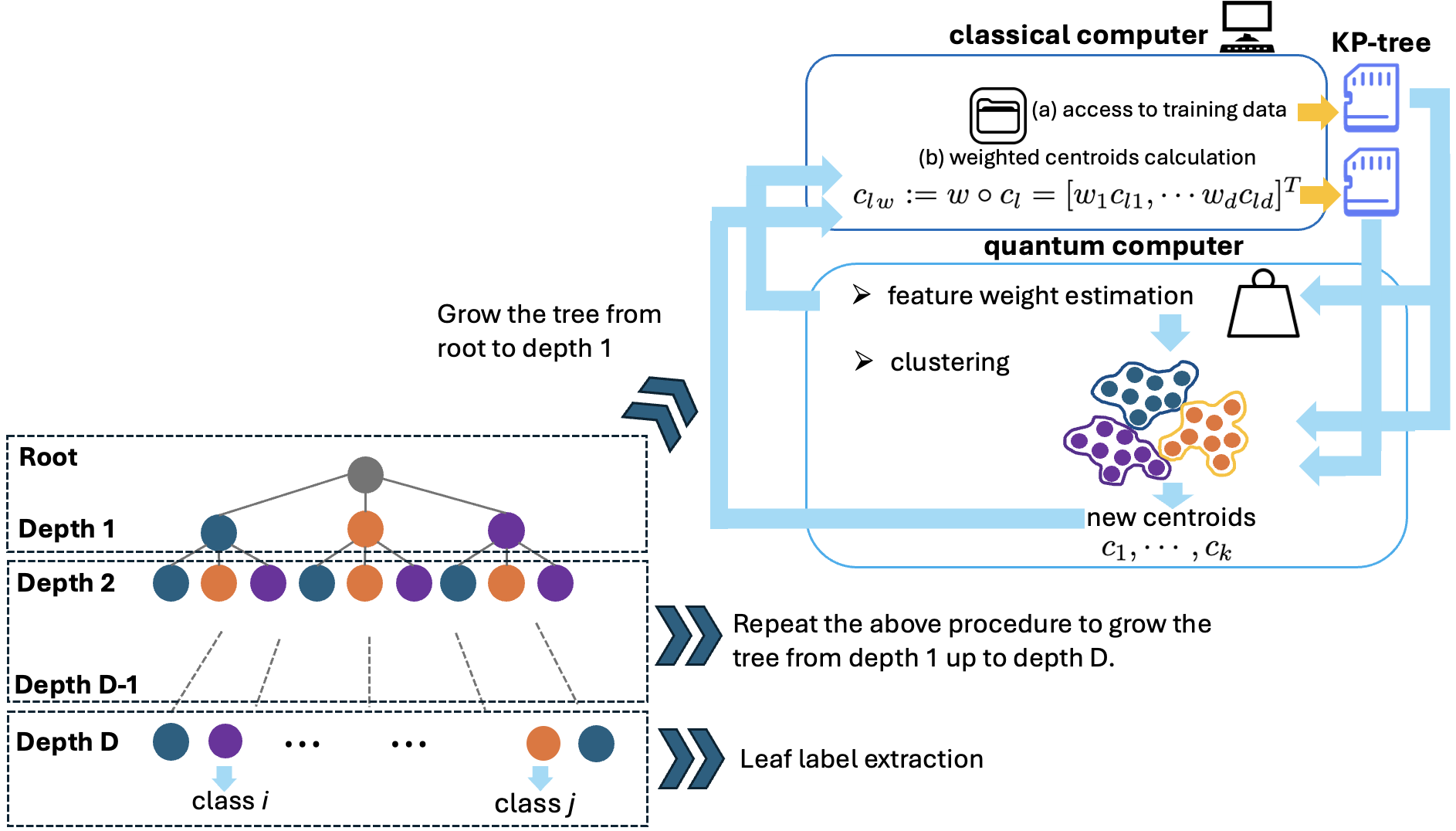} 
    \caption{Diagram of Des-Q. We highlight the procedure to construct a decision tree from root node to depth $D$. The yellow arrows indicate communication between classical components whereas the light blue arrows indicate between classical and quantum components. To grow the tree from root to depth 1, we load the data into the KP-tree data structure from which the samples are queried in superposition to the quantum computer, and the feature weights are estimated. Subsequently, one performs weighted (supervised) clustering to generate $k$ clusters corresponding to depth 1. The above procedure is repeated to grow the tree up to depth $D$. Finally, we perform the leaf label extraction where we assign classes to the leaf nodes. } 
    \label{fig:diagram}%
\end{figure*}

\section{Summary of results} \label{sec:sum_res}

We introduce the Des-q algorithm for tree construction in Algorithm~\ref{alg:desq-cons} along with its illustration in Figure~\ref{fig:diagram}. Given the dataset $\textsf{Data} = \{X, Y\}$ introduced in Sec~\ref{sec:term}, we use a novel supervised quantum clustering algorithm that leverages the unsupervised quantum clustering algorithn named $q$-means introduced by Kerenidis \etal \cite{kerenidis2019q}. With this supervised quantum clustering algorithm, we generate $k$ nodes, or clusters, from the root node. For this, we pick the $k$ initial centroids $\{c_1^0, \cdots c_k^0\} \in \mathbb{R}^d$. Here the superscript denotes the centroids at the $0$-th iteration. At each $t$-th iteration, for each training sample $\{x_i\}_{i \in [N]}$, its cluster label is assigned. The criterion utilized leverages the idea proposed by Liu \etal \cite{liu2020novel} mentioned before where the cluster assignment is based on the weighted Euclidean distance that incorporates the feature weights $w_j \in [0, 1]$ of each feature vector $j \in [d]$ that measure the statistical significance of each feature vector of $X$ in predicting the target label vector $Y$. In Des-q we utilize the normalized absolute value of Pearson correlation coefficient \cite{pearson1895vii} for regression and the point-biserial correlation \cite{glass1970statistical}, which is a special case of Pearson correlation for binary label datasets, for classification. These feature weights are estimated with the proposed quantum methods. By doing this, the cluster labels are assigned as

%%either randomly or using heuristics like $k$-means++ \cite{arthur2007k}. 

\begin{equation}
    \emph{cluster-label}_i^t = \text{arg} \underset{1\leq l \leq k}{\text{min}} \|x_i - c_l^t\|_w, 
\label{ec_distance_algo}
\end{equation}
where $\|x_i - c_l^t\|_w = \sqrt{\sum_{j=1}^d w_j \cdot (x_{ij} - c_{lj}^t)^2}$ is the weighted distance\footnote{Note that this is different to unsupervised $k$-means where $\|x_i - c_l\| = \sqrt{\sum_{j=1}^d (x_{ij} - c_{lj})^2}$ is used for cluster label assignment.}. 

By utilizing this weighted distance, we are making the clustering algorithm supervised as it incorporates some information from the labels. We will show in Section 7 how incorporating these weights improved  the accuracy performance. After all the training examples are partitioned in the $k$ disjoint clusters $\mathcal{C}_l^t$, $l \in [k]$, each cluster centroid $c_l^t$ is updated as the mean of the examples in the cluster. The above two steps are repeated until some stopping criteria. At this point, the first depth of the decision tree has been constructed, where the $k$ nodes correspond to $k$ clusters and these nodes are represented by their centroids. The tree keeps growing using the mentioned supervised quantum clustering technique until it reaches a maximum depth $D$. Next, we need to extract the label information from each of the leaf nodes $l \in [k^D]$, which are characterized by the clusters $\{\mathcal{C}_{l}\}$. Our proposal is to use Des-q to train (and retrain) the decision tree and then perform inference on the classical computer. To perform classical inference on an unseen sample, the algorithm traverses the tree by selecting nodes with the smallest weighted distance. Once it reaches a leaf node, it uses the label information extracted from that leaf node to assign a label.

%either we have converged, i.e., there is no noticeable change in the centroid values from iteration $t$ to $t+1$, or until a maximum number of iteration $K$ is reached.

 %For the regression tree, the label is simply the mean of the label values of the training examples in the cluster. For binary classification, the assigned label corresponds to the value of the majority of the labels in the cluster. This can be obtained from the mean label calculation where if the mean is less $0.5$, then it is assigned the label $0$, and $1$ otherwise. 

%After the tree construction, the inference with a test sample, $x' \in \mathbb{R}^d$ proceeds in a top-down manner by checking its closeness with the $k$ centroids at depth 1, and the sample is assigned to the node whose centroid is closest to. This is repeated until one reaches the leaf node. Finally, the label assigned to the test sample is the label corresponding to the leaf node that it is assigned at the final depth.  

We provide a broad overview of our Des-q algorithm for tree construction and retraining and its steps in the subsections below. Subsequently, in Sec~\ref{sec:qds}, we delve into each of these components, offering detailed insights into their algorithmic steps and discussion on their algorithmic complexities.

\begin{figure*}[ht!]
\begin{minipage}{\linewidth}
\begin{algorithm}[H]
 \caption{Tree construction with Des-q} \label{alg:desq-cons}
%\SetAlgoLined
\begin{algorithmic}
\Require $\textsf{Data} = \{X, Y\}$, where $X = [x_1,\cdots,x_N]^T \in \mathbb{R}^{N \times d}$ and $Y = [y_1,\cdots, y_N]^{T}$. Each $y_i \in \mathbb{R}$ for regression task, and $y_i \in \mathcal{M} = \{0, 1\}$ for binary classification. $X$ can alternatively be written as $X = [x^{(1)},\cdots, x^{(d)}]$ where $x^{(j)}$ are feature vectors of length $N$; $D$: maximum tree depth;  Error parameters $\epsilon_1$ for feature weight estimation, and $\epsilon_2$ for errors in clustering; Precision parameter $\delta$ for quantum clustering; $k$: number of clusters; $K$: maximum number of clustering iterations. 
\Ensure The structure of the decision tree, the centroids $\{c_{\text{node}} \in \mathbb{R}^d\}$ corresponding to each internal node in the tree, and the leaf label prediction.
\State \textbf{Step 1: Load $X$ and $Y$ using quantum-accessible-data structure \cite{kerenidis2016quantum}.}
\begin{flushleft}
$\text{$X$ row access: } \tilde{U}_x: \ket{i}\ket{0} \rightarrow  \ket{i} \otimes  \frac{1}{\|x_i\|} \sum_{j=1}^d x_{ij} \ket{j} = \ket{i}\ket{x_i}$ \\
$\text{$Y$ column access: } \tilde{U}_y: \ket{i}\ket{0} \rightarrow \ket{i} \otimes \frac{1}{\|y\|} \sum_{i=1}^N y_{i} \ket{i}$ = $\ket{i}\otimes\ket{Y}$\\
$\text{$X$ column access: } \tilde{U}_{x^{(j)}}: \ket{0}\ket{j} \rightarrow   \frac{1}{\|x^{(j)}\|} \sum_{i=1}^N x_{ij} \ket{i} \otimes \ket{j} = |x^{(j)}\rangle\otimes \ket{j}$ 
\end{flushleft}
\State \textbf{Step 2: SWAP test for feature weight estimation.}
Use SWAP test followed by amplitude amplification between $|x^{(j)}\rangle$ and $|Y\rangle$ states to estimate the Pearson correlation coefficient between each feature $x^{(j)}$ and label vector $Y$
\begin{center}
    $w_j =  \frac{\sum_{i=1}^N(x_{ij} - \mu_j)(y_i - \mu_y)}{\sqrt{\sum_{i=1}^N(x_{ij} - \mu_j)^2}\sqrt{\sum_{i=1}^N(y_{i} - \mu_y)^2}},$
\end{center}
where $\mu_j = \frac{1}{N} \sum_{i=1}^N x_{ij}$, $\mu_y = \frac{1}{N} \sum_{i=1}^N y_{i}$ and $w := \{w_1,\cdots, w_d\}$ is referred as feature weight vector.
\State \textbf{Step 3: Quantum supervised clustering to generate depth-1 tree.} \\
\textbf{while $t \leq K$:}\\

3.1: Given $k$ centroids $c_1, \cdots, c_k$ at iteration $t$. Perform Hadamard multiplication between $w$ and each $c_j$ and store the resulting weighted centroid in quantum-accessible-data structure.  \\

3.2: \textbf{Perform weighted centroid distance estimation.} Query $X$ and weighted centroid as quantum states to perform the mapping
\begin{center}
    $\frac{1}{\sqrt{N}} \sum_{i=1}^N \ket{i} \otimes_{j \in [k]}\ket{j}\ket{0} \rightarrow \frac{1}{\sqrt{N}} \sum_{i=1}^N \ket{i} \otimes_{j \in [k]}\ket{j}|I_w(x_i, c_j)\rangle,$
\end{center}
where $I_w(x_i, c_j)$ is the estimated weighted inner product between $x_i$ and the centroid $c_j$. \\

3.3: \textbf{Assign the $X$ examples to clusters}. Find the maximum inner product value $I_w(x_i, c_j)_{j \in [k]}$ for each example $x_i$ and create the superposition of all points and their new assigned \emph{centroid labels} (based on maximum inner product) as
\begin{center}
    $\frac{1}{\sqrt{N}} \sum_{i=1}^N \ket{i} \otimes_{j \in [k]}\ket{j}|I_w(x_i, c_j)\rangle \rightarrow \frac{1}{\sqrt{N}} \sum_{i=1}^N \ket{i} \ket{\emph{centroid-label}_i}$, \\
    
\end{center}
where $\emph{centroid-label}_i \in [k]$ is the centroid label assigned to $\ket{x_i}$.\\

3.4: \textbf{Update centroids}. Measure the second register of the previous state to obtain the characteristic vector state as
\begin{center}
    $\ket{\xi_j} = \frac{1}{\sqrt{|\mathcal{C}_j|}}\sum_{i \in \mathcal{C}_j} \ket{i} \hspace{2mm} \forall j \in [k].$
\end{center}
The updated centroid vector $c_j$ is simply the average of the values assigned in the cluster $j$. In order to do this, estimate the inner product of the feature vector $|x^{(l)}\rangle$ and the characteristic vector $|\xi_j\rangle$ to estimate the $l^{th}$ component of the $j^{th}$ updated centroid vector. This is repeated for all $j \in [k], l \in [d]$.\\
\textbf{end while}

\State \textbf{Step 4: Grow the tree to depth D}. The clusters generated in the previous step correspond to the depth 1 of the tree. To further grow the tree, for each current node, generate a superposition of examples in the node cluster. Perform supervised clustering on the cluster examples to grow the tree until it reaches a maximum set depth $D$. 
\State \textbf{Step 5: Leaf label assignment.} For each leaf node/cluster $j \in [k^D]$, perform the fidelity estimation between the state $|Y\rangle$ and $|\xi_j\rangle$ to obtain the mean label value $\emph{label}_j = \frac{1}{\sqrt{|\mathcal{C_j}|}}\sum_{i \in \mathcal{C}_j}y_i$. For regression, the mean label value is the leaf label. For the binary classification, if the mean is less than 0.5, the leaf is assigned a label 0, and 1 otherwise. \\

\end{algorithmic}
\end{algorithm}
\end{minipage}
\end{figure*}

\subsection{Des-q for tree construction}
We now explain each step of Des-q and their time complexities.\\
\noindent \textbf{Step 1} [\emph{Loading classical data}]: We consider the dataset $\textsf{Data}$ as introduced in Section~\ref{sec:term}. This dataset is stored within a specialized tree-like data structure known as the \emph{KP-tree} \cite{kerenidis2016quantum}. Leveraging quantum superposition access to this data structure enables us to prepare \emph{amplitude-encoded states} for the training examples, feature vectors, and label vector\footnote{We query the KP-trees via quantum superposition of the indices and get the quantum superposition state corresponding to the index and the amplitude encoded state corresponding to the example stored in that indexed KP-tree.}. It is worth noting that constructing the KP-tree initially requires a time of $T_{kp} = \tilde{\mathcal{O}}(Nd)$. Once constructed, querying the quantum states associated with the example vectors (in superposition) can be done in $\mathcal{O}(\emph{poly}\log (Nd))$ time, while the query time for the superposition state of target labels is $\mathcal{O}(\emph{poly}\log (N))$.

\noindent \textbf{Step 2} [\emph{Feature weight estimation}]: We quantumly estimate the strength of each feature vector in predicting the label value using the Pearson correlation coefficient and refer to this as feature weight. The Pearson correlation can be applied to assess the strength of features not only when the labels are numerical, as in regression, but also in the context of binary classification by applying the special case of Pearson correlation known as point-biseral. For each feature, we perform this estimation efficiently with amplitude-encoded quantum states using the Frobenius Test circuit \cite{kerenidis2020classification} in combination with the amplitude amplification step \cite{brassard2002quantum}, where the circuit runtime scales as $\mathcal{O}(\log N/\epsilon)$, where $\epsilon$ is the estimation error.

\noindent \textbf{Step 3} [\emph{Supervised Clustering}]: We adapt the efficient (in $N$) unsupervised quantum clustering technique of Kerenidis \etal \cite{kerenidis2019q} called \emph{q-means} for the purpose of constructing decision trees. This adaptation involves performing clustering on the input dataset using weighted distance estimation. The central idea is to pick $k$ initial centroids, perform the Hadamard multiplication between the centroids and the weight vector to create weighted centroids, and finally store them in the KP-tree to access them as amplitude-encoded states\footnote{Once the data is loaded in KP-tree, the clustering process itself is efficient in $N$, meaning its runtime depends poly$-$logarithmically in the number of examples ($N$).}. 

Subsequently, we perform coherent quantum operations to estimate the inner product of all the data points with each of the $k$ weighted centroids. We call this procedure weighted inner product estimation. 
This is followed by quantumly finding the centroid having maximum inner product value with the example vector in order to assign the \emph{centroid-label} to that example. This corresponds to the closest centroid for that example. This creates a superposition state of the index of the example and their associated centroid-label $\in \{1,\cdots,k\}$ which indicates which cluster have they been assigned to. Next, we proceed to update the $k$ centroid states.  \\ 
\\
\\
\\

Measurement operations on the second register of the superposition state result in obtaining the $k$ \emph{characteristic states}, where the $j^{th}$ characteristic state is a superposition of the indices corresponding to the examples assigned to the $j^{th}$ cluster. The updated centroid is the average of the values of the elements in the cluster. This averaging is done by performing quantum multiplication of the transposed data matrix with the characteristic vectors. In contrast to the approach in \cite{kerenidis2019q}, where direct matrix multiplication is followed by quantum state tomography to obtain the updated centroid quantum state, our method involves estimating centroid vectors directly through inner product estimation between the weighted data column quantum states and the characteristic vector states. The clustering procedure is then repeated either a maximum of $K$ times or until convergence is achieved. This process results in the creation of the depth 1 of the decision tree with $k$ child nodes. The computational time for this operation is $\mathcal{O}(\emph{poly} \log (Nd))$, with other terms omitted for simplicity.

\noindent \textbf{Step 4 and 5} [\emph{Tree Growth and Majority Label}]: The growth of the tree continues through iterative clustering of the elements within the nodes at the previous depth, persisting until the desired tree depth $D$ is achieved. The overall runtime for constructing the decision tree up to this stage is $\mathcal{O}(\emph{poly} \log (Nd))$.

Subsequently, labels are assigned to each leaf node. For regression tasks, the label corresponds to the mean of the label values of the examples within the cluster. In binary classification, the label is set to 0 if the mean value is less than 0.5 and 1 otherwise. This is performed by quantumly estimating the mean value where the mean value can be obtained from the inner product of the characteristic vector state corresponding to the leaf node and the label vector state $|Y\rangle$. We show that the estimation of leaf label can be done efficiently with complexity $\mathcal{O}(\emph{poly} \log (Nd))$. 

The overall time complexity for constructing the initial decision tree is the sum of two components: the initial time required for loading the examples into the KP-tree data structure and the time spent on the actual tree construction. It becomes evident that while the tree-building process itself takes $\mathcal{O}(\emph{poly} \log(Nd))$ time, it is the initial data loading into the KP-tree structure that causes the initial tree construction to scale $\tilde{\mathcal{O}}(Nd)$. However, as elaborated in the subsequent section, tree retraining, a central task in supervised learning models, proves to be efficient in terms of $N$ since we retrain with new batches containing significantly fewer training samples.

\subsection{Des-q for tree retraining}

Once the tree has been initially constructed and deployed, it needs to be periodically retrained to account for new batch of training data of size $N_{\text{new}}$. Frequently, the quantity of new training examples is substantially smaller than the initial training set size, i.e., $N_{\text{new}} \ll N$. This periodic step is of crucial importance for the model, serving to prevent performance deterioration and maintain consistent model performance.

The retraining process begins by querying quantum superposition states from the KP-tree, encompassing the entire set of $N_{\text{tot}} = N+N_{\text{new}}$ training examples. Since the original $N$ examples were already preloaded into the KP-tree during the tree construction part, one is only required to load the new examples into the KP-tree, a step that consumes time $\tilde{\mathcal{O}}(N_{\text{new}}d)$.  Subsequently, as highlighted in Algorithm~\ref{alg:desq-retrain}, we replicate all the stages of the Algorithm \ref{alg:desq-cons} designed for constructing the tree, but where we instead query superposition states over the set of $N_{\text{tot}}$ examples. Upon loading the new training data into the data structure, the Des-q algorithm for retraining has a runtime of $\mathcal{O}(\emph{poly} \log(N_{\text{tot}}d)) \approx \mathcal{O}(\emph{poly} \log(Nd))$. Thus we see that the total time to build the tree is the time to load the new examples into the KP-tree followed by the time to run the algorithm $$\tilde{\mathcal{O}}(N_{\text{new}}d) + \mathcal{O}(\emph{poly} \log(Nd)).$$

Under the small batch updates criterion, we see that our algorithm for tree retraining scales poly-logarithmically in $N$, thus ensuring that the tree retraining process is swift and efficient.

\begin{algorithm}[H]
 \caption{Tree retraining with Des-q} \label{alg:desq-retrain}
%\SetAlgoLined
\begin{algorithmic}
\Require $\textsf{Data} = \{X_{\text{tot}}, Y_{\text{tot}}\}$, where $X_{\text{tot}} = X:X_{\text{new}}= [x_1,\cdots,x_N,\cdots, x_{N_{\text{tot}}}]^T \in \mathbb{R}^{N_{\text{tot}} \times d}$ and $Y_{\text{tot}} = Y:Y_{\text{new}} =[y_1,\cdots, y_N,\cdots y_{N_{\text{tot}}}]^{T}$. The rest of the requirements are the same as Algorithm~\ref{alg:desq-cons}.   
\Ensure the structure of the decision tree, the centroids $\{c_{\text{node}} \in \mathbb{R}^d\}$ corresponding to each internal node in the tree, and the leaf label prediction.
\State \textbf{Step 1:} load $X_{\text{new}}$ and $Y_{\text{new}}$ using quantum-accessible-data structure \cite{kerenidis2016quantum}.
\State \textbf{Step 2:} Repeat steps 2:5 of Algorithm~\ref{alg:desq-cons} but with querying superposition states corresponding to $X_{\text{tot}}$ and $Y_{\text{tot}}$.
\end{algorithmic}
\end{algorithm}

\section{Tree construction with Des-q} \label{sec:qds}

Now that we have presented the high-level description of Des-q, we present the technical details of each component of the decision tree construction for the tasks of regression and binary classification. We also utilize other quantum algorithmic ingredients in Des-q, for which we refer the reader to Appendix~\ref{sec:ingre}. The five key components of Des-q for tree construction are
\begin{enumerate}
    \item \textit{Data loading:} Load the classical data into quantum amplitude states. 
    \item \textit{Feature weight estimation:} Estimate the Pearson correlation between each feature vector and the target label. 
    \item \textit{Supervised clustering:} Perform supervised clustering at root node to generate nodes at first depth.
    \item \textit{Tree growth:} For each node, generate the superposition of the data examples and perform supervised clustering on the superposition state to further grow the tree until reaching maximum set depth $D$. 
    \item \textit{Leaf label assignment:} When the leaf node is reached, compute the mean of the label values of the examples in the cluster. For regression, the leaf label is the mean, while for binary classification, the leaf label is 0 if the mean is less than 0.5, and 1 otherwise. 
\end{enumerate}

The subsections below correspond to each of the above-mentioned components. 

\subsection{Data loading} \label{sec:init_data_loading}

The first step to build a quantum algorithm on classical dataset $\textsf{Data} = \{X, Y\}$ (Sec~\ref{sec:term}) is to be able to load it into a quantum accessible memory such that one can efficiently query the data as quantum encoded states. Among the most space-efficient quantum encoding proposals is the amplitude encoding scheme which provides a natural link between quantum computing and linear algebra to be able to construct useful quantum algorithms \cite{montanaro2016quantum}. 

\subsubsection{Quantum amplitude encoding}

The quantum amplitude encoding requires $\ceil{\log(d)}$ qubits to encode a vector $x = (x_1,\cdots,x_d) \in \mathbb{R}^d$ using the following form
\begin{equation}
    \ket{x} = \frac{1}{\|x\|} \sum_{j=1}^d x_j \ket{j}.
    \label{eq:amp_enc}
\end{equation}

Further one can also encode $N$ vectors simultaneously, which is equivalent to encoding a matrix $X \in \mathbb{R}^{N \times d}$ as
\begin{equation}
    \ket{X} = \frac{1}{\|X\|_F} \sum_{i=1}^N \|x_i\| \ket{x_i} \ket{i},
    \label{eq:amp_enc_matrix}
\end{equation}

where $\ket{x_i} = \frac{1}{\|x_i\|} \sum_{j=1}^d x_{ij} \ket{j}$ and $\|X\|_F = \sqrt{\sum_{i=1}^N \|x_i\|^2}$ is the matrix Frobenius norm.

There are multiple proposals for preparing the states given in Eq~\ref{eq:amp_enc} and Eq~\ref{eq:amp_enc_matrix} that use an efficient data loading structure. Here we showcase a time-efficient (poly-logarithmic in the size of the input) method of such state preparation using the quantum-accessible-data structure called the KP-tree \cite{kerenidis2016quantum}. In Appendix~\ref{sec:QRAM_AMP} we also show an alternate method of preparing such states when one is given access to the oracular quantum-random-access memory structure, whose time complexity is, in general, proportional to the square root of the size of the input unless in special cases. 

\subsubsection{Amplitude encoding with KP-tree} \label{sec:KP-tree}

As highlighted by the original authors \cite{kerenidis2016quantum} and further in \cite{chakraborty2018power}, the KP-tree data structure can be seen as a quantum-accessible-data structure because it is a classical tree-like data structure, which is stored in a quantum read-only-memory and accessed via superposition. It facilitates the creation of amplitude encoding states as stated in the lemmas below. 

\begin{lemma} \label{Thm:KP-tree}
    Let $X \in \mathbb{R}^{N \times d}$ be a given dataset. Then there exists a classical data structure to store the rows of $X$ with the memory and time requirement to create the data structure being $T_{kp} = \mathcal{O}(Nd \log^2(Nd))$ such that, there is a quantum algorithm with access to the data structure which can perform the following unitaries (and also in superposition)
    \begin{align}
        \ket{i}\ket{0} &\rightarrow \ket{i} \frac{1}{\|x_i\|}\sum_{j=1}^d x_{ij}\ket{j} \\
        \ket{0} &\rightarrow \frac{1}{\|X\|_F} \sum_{i=1}^N \|x_i\| \ket{i}
    \end{align}
    in time $T = \mathcal{O}(\text{poly}\log (Nd))$.
\end{lemma}

\begin{proof}
We provide proof of this lemma in Appendix~\ref{App:proof_data_load}.
\end{proof}
Similarly, using the results of Lemma~\ref{Thm:KP-tree}, we also have the following two Lemmas to load the elements of the columns of the matrix $X$ which correspond to the feature vectors, and also to load the elements of the label vector $Y$.
\begin{lemma}[Superposition over example columns] \label{lemma:supcol}
        Let $X \in \mathbb{R}^{N \times d}$ be a given dataset. Then there exists a classical data structure to store the columns of $X$ given by feature vectors $x^{(j)}:= (x_{1j}, \cdots, x_{Nj}), j \in [d]$ with the memory and time requirement to create the data structure being $\mathcal{O}(Nd \log^2(Nd))$ such that, there is a quantum algorithm with access to the data structure which can perform the following unitary (and also in superposition) 
    \begin{align}
       \ket{0} \ket{j} \rightarrow  \left(\frac{1}{\|x^{(j)}\|}\sum_{i=1}^N x_{ij}\ket{i} \right)\ket{j}  = \ket{x^{(j)}} \ket{j}
    \end{align}
    in time $T = \mathcal{O}(\text{poly}\log (Nd))$.
\end{lemma}

\begin{lemma}[Superposition over label data]  \label{lemma:suplabel}
        Let $Y = [y_1,\cdots, y_N]^T \in \mathbb{R}^{N \times 1}$ be a given label vector. Then there exists a classical data structure to store the elements of $Y$ with the memory and time requirement to create the data structure being $\mathcal{O}(N \log^2(N))$ such that, there is a quantum algorithm with access to the data structure which can perform the following unitaries (and also in superposition) in time $T = \mathcal{O}(\text{poly}\log (N))$ which can perform the following operation
    \begin{align}
       \ket{0}  \rightarrow  \frac{1}{\|Y\|}\sum_{i=1}^N y_i\ket{i},
    \end{align}
    where $\|Y\| = \sqrt{\sum_{i=1}^N y_i^2}$.
\end{lemma}

\subsection{Feature weight estimation} \label{sec:bivariatecorr}

Statistical analysis provides various quantitative techniques to assess the relationship between two variables \cite{akoglu2018user}. When dealing with two numerical variables, a common method to quantify their bivariate linear correlation, indicating the predictability of one variable based on the other, is through the Pearson correlation coefficient \cite{pearson1895vii}.
In our context, we are interested in examining the relationship between each feature vector and the target label vector using the Pearson correlation coefficient. This approach applies to regression tasks where both the feature vector and label vector consist of numerical variables. 

For binary classification tasks, point-biserial correlation measures the connection between numerical feature vectors and a categorical label vector with two classes (typically denoted as $\{0,1\}$) \cite{linacre2008expected}. It is worth noting that point-biserial correlation essentially represents a specialized application of the Pearson correlation for scenarios involving binary label vectors. Therefore, we choose to employ the Pearson correlation to quantify feature weights for both regression and binary classification tasks.

\subsubsection{Pearson correlation coefficient}

The Pearson correlation, also known as the Pearson product-moment correlation coefficient, provides a normalized measure of the covariance between two data sets. In essence, it quantifies the relationship between two variables by comparing the ratio of their covariance to their individual standard deviations. While the Pearson correlation effectively captures linear relationships between variables, it may not accurately represent non-linear associations. Pearson correlation assumes normal distribution for both variables and can be applied to assess relationships between either nominal or continuous variables.

Our objective is to compute the Pearson correlation between each feature vector $x^{(j)}$ and the label vector $Y$ to ascertain the importance of each feature in the prediction. Consider we want to compute the correlation between the feature $j$ with entries $ x^{(j)} = \{x_{1j},\cdots x_{Nj}\}$ and the output label $Y$ with entries $y_1,\cdots, y_N$. The Pearson correlation coefficient between the two $N$ dimension vectors is defined as
\begin{equation}
\small
    w_{j} =  \frac{\sum_{i=1}^N(x_{ij} - \mu_j)(y_i - \mu_y)}{\sqrt{\sum_{i=1}^N(x_{ij} - \mu_j)^2}\sqrt{\sum_{i=1}^N(y_{i} - \mu_y)^2}} \forall j \in [d],
    \label{Eq:pearson}
\end{equation}
where $\mu_j = \frac{1}{N} \sum_{i=1}^N x_{ij}$ and $\mu_y = \frac{1}{N} \sum_{i=1}^N y_{i}$.

With the aforementioned amplitude-encoded states, we propose two efficient (in $N$) methods of quantumly computing the Pearson correlation coefficient compared to the classical method \cite{cohen2009pearsonharrison1978hedonic} that takes time $\mathcal{O}(N)$, whereas the proposed quantum algorithm takes time poly-logarithmic in $N$ when the data is encoded in amplitude-encoded states. Whereas the first method has a complexity that scales on the error as $\mathcal{O}(1/\epsilon^2)$, the second method improves the error dependence quadratically to $\mathcal{O}(1/\epsilon)$. Here the error dependence is defined by $|\Bar{w_{j}} - w_j| \leq \epsilon$, where $\Bar{w_{j}}$ is the estimated correlation.
We note that while we consider amplitude encoding states in the following methods, the feature and label vectors can also be encoded in QRAM-like states (e.g., $\sum_i^N\ket{i}\ket{y_i}$) and we can use quantum counting based algorithms~\cite{PhysRevLett.130.150602} to efficiently calculate the correlation coefficient (see Appendix~\ref{appendix:QBC_for_correlation} for details).

\subsubsection{Method 1: Inverse error-squared dependence on run time} \label{sec:method1}

\begin{theorem} \label{Thm:inverse_sqared_error}
    Given access to the amplitude-encoded states for feature vectors $|x^{(j)}\rangle, j \in [d]$ and the label vector $\ket{Y}$ along with their norms $\|x^{(j)}\|$, $\|Y\|$ which are prepared in time $T= \mathcal{O}(\text{poly}\log (Nd))$, there exists a quantum algorithm to estimate the Pearson correlation coefficients $\Bar{w_j}, \forall j \in [d]$ in time $\mathcal{O}(\frac{Td\eta}{\epsilon^2})$, where $|\Bar{w_{j}} - w_j| \leq \epsilon$, and \\
    $ \eta = \frac{7 \cdot \text{max}\left(\|x^{(j)}\| \|Y\|, \|x^{(j)}\|^2, \|Y\|^2\right)}{N \cdot \text{min}\left(\sigma_{x^{(j)}}\sigma_Y, \sigma_{x^{(j)}}^2, \sigma_Y^2\right)}$,\\
    where $\sigma_{x^{(j)}}$ and $\sigma_{Y}$ denote the standard deviation for $x^{(j)}$ and $Y$.
\end{theorem}

\begin{proof}
We provide the proof of Theorem~\ref{Thm:inverse_sqared_error} in Appendix~\ref{app:inverse_sqared_error}.
\end{proof}

\subsubsection{Method 2: Inverse error dependence on run time} \label{sec:PEinverseerror}

\begin{theorem}\label{thm:pe}
    Given access to the amplitude-encoded states for feature vectors $|x^{(j)}\rangle, j \in [d]$ and the label vector $\ket{Y}$ along with their norms $\|x^{(j)}\|$, $\|Y\|$ which are prepared in time $T = \mathcal{O}(\text{poly}\log (Nd))$, there exists a quantum algorithm to estimate each Pearson correlation coefficient $\Bar{w_j}, j \in [d]$ in time $T_{w} = \mathcal{O}(\frac{Td\eta}{\epsilon})$, where $|\Bar{w_{j}} - w_j| \leq \epsilon$, and \\
    $ \eta = \frac{7 \cdot \text{max}\left(\|x^{(j)}\| \|Y\|, \|x^{(j)}\|^2, \|Y\|^2\right)}{N \cdot \text{min}\left(\sigma_{x^{(j)}}\sigma_Y, \sigma_{x^{(j)}}^2, \sigma_Y^2\right)}$, \\
    where $\sigma_{x^{(j)}}$ and $\sigma_{Y}$ denote the standard deviation for $x^{(j)}$ and $Y$. 
\end{theorem}

\begin{proof}
We provide the proof of Theorem~\ref{thm:pe} in Appendix~\ref{app:pe}. The key idea is to augment the proof technique in Theorem~\ref{Thm:inverse_sqared_error} with the amplitude estimation step Appendix~\ref{sec:AE} to provide a quadratic improvement in the error dependence. 
\end{proof}

\subsubsection{Normalizing the feature weights}

Once we obtain the feature weights $w = [w_1,\cdots,w_d]^T$, we take the absolute values of the weights since the absolute value of the Pearson correlation coefficient is the true indicator of the strength of the relationship between the two variables. Subsequently, in order to obtain the relative impact of each feature, we normalize the feature weights resulting in the normalized feature weights as
\begin{equation}
    w_{\text{norm}} = \frac{1}{\|w\|}[|w_1|,\cdots,|w_d|]^T,
\end{equation}
where $\|w\| = \sqrt{\sum_{j=1}^d w_l^2}$. This makes the norm of $\|w_{\text{norm}}\| = 1$. 
For simplicity, we refer to $ w_{\text{norm}}$ as $w$ henceforth. 

\subsection{Clustering: supervised $q$-means to expand the decision tree} \label{sec:clustering}

After having computed the feature weights $w = [w_1,\cdots,w_d]^T$, we proceed with clustering on the parent node to generate $k$ children nodes. This process is repeated sequentially until the maximum tree depth $D$ is reached. To perform this clustering, we leverage the unsupervised $q$-means algorithm developed by Kerenidis \etal \cite{kerenidis2019q} to include the feature weights in the distance estimation in such a way that the training examples are assigned to the centroid labels according to the minimum weighted distances defined in Eq~\ref{ec_distance_algo}.
$q$-means is analogous to the classical robust $k$-means clustering, also known as $\delta$-$k$-means, where the idea is to start with $k$ initial centroids which are either chosen randomly or using heuristics like $k$-means++ \cite{arthur2007k}. The algorithm then alternates between two steps: (a) each data point is assigned to the closest centroid thus forming total $k$ disjoint clusters, and (b) the centroid vectors are updated to the average of data points assigned to the corresponding cluster. The two steps are repeated for a total of $K$ times to find the approximate \emph{optimal} centroid points.
We make the same assumption of working with datasets that allow for
good clustering as made by the Kerenidis \etal \cite{kerenidis2019q}. A dataset being\textit{ well-clusterable} means that the $k$ clusters arise from picking $k$ well-separated centroids, and then each point in the cluster is sampled from a Gaussian distribution with small variance centered on the centroid of the cluster. This implies that the size of each cluster $\mathcal{C}_j$ can be bounded by $|\mathcal{C}_j| = \Omega(\frac{N}{k})$. For Des-q the assumption is that the data is well-clusterable at root (i.e., considering the dataset entirely) and at each internal node of the tree, which contains a subset of this dataset. For such well-clusterable datasets we can obtain a polylogarithmic time complexity in terms of number of samples $N$.
Moreover, in order to bound the norms of $X$ and $Y$, we make the assumption that they have bounded constant entries allowing to take $\|x^{(l)}\| = \mathcal{O} (\|x^{(l)}\|_{\infty}\sqrt{N}) = \mathcal{O}(\sqrt{N})$ and $\|Y\| = \mathcal{O} (\|Y\|_{\infty}\sqrt{N}) = \mathcal{O}(\sqrt{N})$, which we will use to bound some complexities in Theorem 5.6 and 5.7.

Our main result for performing supervised clustering for any given parent node is the following
\begin{theorem} 
    Given quantum access to the dataset $X$ in time $T = \mathcal{O}(\text{poly} \log(Nd))$, the supervised $q$-means algorithms takes $K$ iterations steps to output with a high probability the centroids $\overline{c_1}, \cdots, \overline{c_k}$ that are arbitrarily close to the output centroids of the $\delta$-$k$-means algorithm $c_1, \cdots, c_k$ in time complexity
    \begin{equation}
    \small
    \mathcal{O}\left(\emph{poly} \log(Nd) \frac{K k^\frac{5}{2} d \log (k)\log(1/\Delta)\eta_1}{\epsilon_1\epsilon_2}\right),
    \end{equation}
    where $\eta_1 = \text{max}_i(\|x_i\|^2)$ and $\epsilon_1$ is the error in the estimation of the weighted inner product between each sample and centroid such that $|\overline{I_w(x_i, c_j)} - I_w(x_i, c_j)| \leq \epsilon_1$ with a probability at least $1 - 2 \Delta$ (refer to Theorem \ref{thm:cd}) and $\epsilon_2$ is the error in the centroid estimation as $\|\overline{c_j} - c_j\|_\infty \leq \epsilon_2$ $\forall j \in [k]$ (refer to Theorem \ref{thm:cent_update}). 
\end{theorem}
%$\epsilon_1, \epsilon_2, \Delta > 0$
Here, we present an overview of the key steps in our proposed quantum-supervised algorithm that utilizes the $q$-means algorithm \cite{kerenidis2019q}. Initially, quantum subroutines are employed to assign training examples to their nearest centroids using weighted distance estimation. Subsequently, a quantum centroid update subroutine calculates the new centroid values. The primary subroutines, which we will elaborate on in the following subsections, are as follows.
    \begin{enumerate}
        \item  \emph{Weighted centroid distance estimation.}
        \item \emph{Cluster assignment.}
        \item \emph{Centroid states creation.}
        \item \emph{Centroid update.}
    \end{enumerate}

\subsubsection{Weighted centroid distance estimation} \label{sec:cd}
%romina here 
The algorithm starts by selecting $k$ $d$-dimensional vectors as initial centroids $c_1^0,\cdots, c_k^0$ (for ease of notation, we refer to them as $c_1,\cdots, c_k$) which can be chosen either randomly or by an efficient procedure such as $k$-means++ \cite{arthur2007k}. Next, one performs the Hadamard multiplication (also called element-wise multiplication) between each centroid vector and the weight vector $w$ resulting in the \emph{weighted}-centroids as
\begin{equation}
c_{j, w} = w \circ c_j = [w_1c_{j1}, \cdots w_dc_{jd}]^T \hspace{2mm} j \in [k].
\end{equation}
This operation takes a total time $\mathcal{O}(kd)$. Next, the $k$ weighted centroids are loaded in KP-tree (Sec~\ref{sec:KP-tree}) in time $T_{kp}^c = \mathcal{O}(kd \log^2(kd))$  and retrieved as amplitude-encoded states in time $\mathcal{O}(\emph{poly} \log(kd))$ by doing
\begin{equation}
    \ket{j}\ket{0} \rightarrow \ket{j}|c_{j,w}\rangle = \ket{j}\frac{1}{\|c_{j,w}\|}\sum_{l=1}^d w_lc_{jl}\ket{l},
\end{equation}
where $\|c_{j,w}\| = \sqrt{\sum_{l=1}^{d}w_l^2c_{jl}^2}$ is the norm of the weighted centroid.

Further, using Lemma~\ref{Thm:KP-tree}, we can also query amplitude-encoded training example states $\ket{x_i}$ in time $T = \mathcal{O}(\emph{poly} \log(Nd))$ as
\begin{equation}
\begin{split}
    \ket{i}\ket{0} &\rightarrow \ket{i}\ket{x_i}, \\
    %\frac{1}{\sqrt{N}}\ket{i}\ket{0} &\rightarrow \frac{1}{\sqrt{N}}\ket{i}\ket{x_i},
\end{split}
\end{equation}
where $\ket{x_i} = \frac{1}{\|x_i\|}\sum_{j=1}^d x_{ij}\ket{j}$.

The idea of assigning all the $N$ training examples to their closest centroid begins by estimating the weighted distance between the examples (stored in quantum superposition) and each of the $k$ centroid vectors (also stored in quantum amplitude encodings) and then selecting the minimum distance between the example and the $k$ available centroids. Here we assign the examples based on them having the maximum \emph{weighted} inner product, or overlap, with respect to the $k$ centroid vectors. The weighted inner product between an example $x_i$ and the centroid $c_j$ is defined as
\begin{equation}\label{eq:weighted-inner-product}
    I_w(x_i, c_j) = \sum_{l=1}^d w_l x_{il}c_{jl} = x_i \cdot c_{j, w} = I(x_i, c_{j,w}).
\end{equation}

Let us first look at the procedure of estimating the inner product between all the examples and the centroids. 

\begin{theorem} \label{thm:cd}
        Given quantum access to the dataset $X$ in time $T = \mathcal{O}(\text{poly} \log(Nd))$ and quantum access to the weighted centroids $c_{1,w},\cdots,c_{k,w}$ in time $\mathcal{O}(\text{poly} \log(kd))$, then for any $\Delta > 0$ and $\epsilon_1 > 0$, there exists a quantum algorithm such that
        \begin{equation}
            \begin{split}
            &\frac{1}{\sqrt{N}} \sum_{i=1}^N \ket{i} \otimes_{j \in [k]}(\ket{j}\ket{0}) \\ &\rightarrow \frac{1}{\sqrt{N}} \sum_{i=1}^N \ket{i} \otimes_{j \in [k]}(\ket{j}\ket{\overline{I_w(x_i, c_j)}}),
            \label{eq:ip-sup}
            \end{split}
        \end{equation}
        where $I_w(.)$ is the weighted inner product between the two vectors and $|\overline{I_w(x_i, c_j)} - I_w(x_i, c_j)| \leq \epsilon_1$ with probability at least $1 - 2\Delta$, in time $T_{cd} = \mathcal{O}(T k \frac{\eta_1}{\epsilon_1} \log(1/\Delta))$, where $\eta_1 = \text{max}_{i}(\|x_i\|^2)$.
\end{theorem}

\begin{proof}
We prove the above Theorem~\ref{thm:cd} in Appendix~\ref{app:cd}.
\end{proof}

\subsubsection{Cluster assignment} \label{sec:ca}

Once we have the superposition state of Theorem~\ref{thm:cd}, we can use the following theorem to assign each training example $x_i$ to the closest centroid label referred by $\emph{centroid-label}_{i} \in [k]$\footnote{Note that the centroid label is different from the original dataset label. Since we are doing $k$-way clustering, thus each cluster centroid $c_j$ is assigned a label $j$, for $j \in [k]$}. In our case, this would correspond to the largest weighted inner product between each example and the $k$ centroid points. 

\begin{theorem}[Finding maximum weighted inner product] \label{thm:ca}
    Given $k$ different $\log p$ size registers $\otimes_{j \in [k]} |\overline{I_w(x_i, c_j)}\rangle$ whose creation time is $T_{cd} = \mathcal{O}(T k \frac{\eta_1}{\epsilon_1} \log(1/\Delta))$, where $T$, $\epsilon_1$ and $\Delta$ are defined in Theorem \ref{thm:cd}, there is a quantum map which performs the following operations
    \begin{equation}
    \begin{split}
        &\otimes_{j \in [k]} |\overline{I_w(x_i, c_j)}\rangle\ket{1}\\
        &\rightarrow  \otimes_{j \in [k]}|\overline{I_w(x_i, c_j)}\rangle |\text{arg-max}(\overline{I_w(x_i, c_j)})\rangle
    \end{split}
    \end{equation}
    incurring a total time $T_l = \mathcal{O}(k\log p) + T_{cd} = \mathcal{O}(k\log p + T k \frac{\eta_1}{\epsilon_1} \log(1/\Delta)) = \mathcal{O}(T k \frac{\eta_1}{\epsilon_1} \log(1/\Delta)) = T_{cd}$. Note that we removed the first term as $\log p$ is a constant as it is the size of the registers to encode each inner product $|\overline{I_w(x_i, c_j)}\rangle$ (Theorem  \ref{thm:cd}).
\end{theorem}
\begin{proof}
    This process can be done by having an additional register initialized in $\ket{1}$. Then for $2 \leq j \leq k$, a total of $k$ repeated comparisons are performed of the two inner product registers, i.e., if the value of the inner product register corresponding to centroid $c_j$ is greater than the value corresponding to $c_{j+1}$, then the last register takes the centroid label value $j+1$ and so on. Thus the total time complexity of this procedure is $\mathcal{O}(k \log p)$. 
    This allows us to produce the state
    \begin{equation}
        \frac{1}{\sqrt{N}}\sum_{i=1}^N\ket{i}\ket{\emph{centroid-label}_{i}},
        \label{eq:label state}
    \end{equation}
    where $\emph{centroid-label}_{i}$ is the centroid label corresponding to the maximum weighted inner product with respect to the $k$ centroids for the example $x_i$. Thus the total time complexity (including the previous step of weighted inner product estimation) is $T_l$.
\end{proof}

\subsubsection{Updated centroid vectors creation} \label{sec:centroid_val_create}

Once the clusters are assigned, the next step is to perform the averaging of the values in the cluster (i.e., an average of weighted examples with the same centroid label) to obtain the new cluster centroids. One way to do this is to note that the state in Eq~\ref{eq:label state} can be re-written as
\begin{equation}
\small
    \sum_{j=1}^k\sqrt{\frac{|\mathcal{C}_j|}{N}}\left(\frac{1}{\sqrt{|\mathcal{C}_j|}}\sum_{i \in \mathcal{C}_j}\ket{i}\right)\ket{j} = \sum_{j=1}^k\sqrt{\frac{|\mathcal{C}_j|}{N}}|\xi_j\rangle\ket{j},
    \label{eq:characteristic}
\end{equation}

where $|\xi_j\rangle$ corresponds to the uniform superposition over the indices in the cluster $\mathcal{C}_j$ also referred to as the characteristic vector state corresponding to the centroid label value $j$. As mentioned, we assume well-clusterable datasets, implying that the clusters are all non-vanishing and they have the size $|\mathcal{C}_j| = \Omega(N/k)$, we see that with $\mathcal{O}(k \log k)$ measurements of the final register (using coupon collector arguments \cite{boneh1997coupon}), we obtain all the $k$ index states $|\xi_j\rangle, j \in [k]$ with $\Omega(1)$ probability. Thus we can obtain each characteristic vector state $|\xi_j\rangle$, $j \in [k]$ in time $T_\xi = \mathcal{O}(T_l\log k)$ such that the total time to obtain all the characteristic vectors is $kT_\xi$. 

Now, we use the following theorem to obtain all the $k$ updated classical centroid vectors denoted by $c^1_1, \cdots, c^1_k$ where $c^1_j \in \mathbb{R}^d$ $\forall j \in [k]$ and the superscript denotes the updated centroid vectors at the cluster iteration 1. For ease of notation, we omit the superscript and denote them as  $c_1, \cdots, c_k$.
\begin{theorem}
    Given access to the characteristic vector states $|\xi_j\rangle$ $\forall j \in [k]$ where each state is prepared in time $T_\xi = \mathcal{O}(T_l\log k)$ and the amplitude-encoded states for feature vectors $|x^{(l)}\rangle$ $\forall l \in [d]$ which can be prepared in time $T = \mathcal{O}(\text{poly} \log(Nd))$, there exists a quantum algorithm to obtain the updated centroid vectors $\overline{c_1}, \cdots, \overline{c_k}$ such that $\|\overline{c_j} - c_j\|_\infty \leq \epsilon_2$ $\forall j \in [k]$ in time 
    $$T_{\text{sup-cluster}} =\mathcal{O}\left(\frac{T_\xi k^\frac{3}{2}d}{\epsilon_2}\right),$$
    where $c_j = X^T\xi_j$ is the true mean of the weighted examples in the cluster and thus the true updated centroid vector of $c_j$. 
    \label{thm:cent_update}
\end{theorem}

\begin{proof}
We prove the Theorem~\ref{thm:cent_update} in Appendix~\ref{app:cent_update}.
\end{proof}

We can now easily compute the time taken to generate the first iteration of centroid vectors\footnote{For the time being we will exclude the time taken to load the initial data $X, X^T, Y$ in the KP-tree and the time to compute the Pearson correlation weight vector. This is because these are one-time processes and they will be included in the total decision tree construction time in the end.}. The total time taken is the combination of time to load the initial centroid vectors in KP-tree and to perform supervised clustering (a combination of data update and clustering). This can be written as

\begin{widetext}
   \begin{equation}
    \begin{split}
            T_{\text{iter-1}} &= T_{kp}^c + T_{\text{sup-cluster}}  \\
    &=  \mathcal{O}(\log^2(kd)\cdot kd) + \mathcal{O}\left(T_\xi\frac{ k^\frac{3}{2}}d{\epsilon_2}\right) \\
    &=   \mathcal{O}(\log^2(kd)\cdot kd) + \mathcal{O}\left(T_l\frac{ k^\frac{3}{2}} d \log (k) {\epsilon_2}\right)\\
    &=   \mathcal{O}(\log^2(kd)\cdot kd) +  \mathcal{O}\left(T_{cd}\frac{ k^\frac{3}{2}} d \log (k){\epsilon_2}\right) \\
    &=   \mathcal{O}(\log^2(kd)\cdot kd) + \mathcal{O}\left(\emph{poly} \log(Nd) \log(1/\Delta)\eta_1 \frac{   k^\frac{5}{2}} d \log (k){\epsilon_1\epsilon_2}\right) \\
    &= \mathcal{O}\left(\frac{ \emph{poly} \log(Nd) \log(1/\Delta)\eta_1 k^\frac{5}{2}} d \log (k){\epsilon_1\epsilon_2}\right),
    \end{split}
\end{equation} 
\end{widetext}
where $\eta_1$ and $\epsilon_1$ are defined in Theorem \ref{thm:cd}) and $\epsilon_2$ in Theorem \ref{thm:cent_update}). Further, $T_{kp}^c$ is the time to load the weighted centroid vector in the KP-tree as mentioned in Sec~\ref{sec:cd}.

\subsubsection{Repeating clustering iteration step}

Once we have obtained the $k$ classical centroid vectors for one iteration, we repeat the above 3-step process $K$ times to obtain the final centroid vectors at the $K$-th iteration. This creates the first depth of the decision tree with a total of $k$ children nodes corresponding to the $k$ clusters generation. The total time complexity of this procedure is (excluding the time to load the initial data $X, X^T, Y$ into the KP-tree and the time to compute the Pearson correlation weight vector)
\begin{equation}
    \begin{split}
            T_{\text{depth}:[0\rightarrow 1]} &=  KT_{\text{iter-1}}.
    \end{split}
\end{equation}

\subsection{Tree growth} \label{sec:Tree_growth}

The previous steps in Sec~\ref{sec:bivariatecorr}-\ref{sec:clustering} show how to grow the tree from the root node to the $k$ children nodes at depth 1. As a result we obtain the $k$ centroid vectors at depth 1 $\{c_{1}^{1}, \cdots, c_{k}^{1}\}$ where the superscript denotes the centroids at tree depth 1. Subsequently, in order to expand the tree further, we need to perform clustering again on each of the $k$ children nodes. This is repeated till we reach the final tree depth $D$. The challenge here is that we started with all the samples in superposition at root ($\frac{1}{\sqrt{N}}\sum_{i=1}^N \ket{i}\ket{x_i}$). After having performed clustering, now we require the superposition of the samples at each children node to continue expanding the tree. This state is the index superposition state. For the $j$-th cluster at depth $l$ ($\mathcal{C}_{j,l}$), it is defined as $|\xi_j\rangle = \frac{1}{\sqrt{|\mathcal{C}_{j, l}|}} \sum_{i \in \mathcal{C}_{j,l}}\ket{i}$. Creating this state is non-trivial, as we will show in Section \ref{sec:majority label}.

Let us see how to further expand the tree by expanding at the nodes from depth $l$ to $l+1$. This can be done in the following two sequential steps,
\begin{enumerate}
    \item Let us denote the node $j$ in depth $l$ with the cluster $\mathcal{C}_{j,l} \forall j \in [k^l]$ and $l \in [1, D-1]$. The first step is to create the superposition over the examples in the cluster $\mathcal{C}_{j,l}$.
    \item Perform the clustering to create $k$ children nodes for each parent cluster node $\mathcal{C}_{j,l}$ at depth-$l$.
\end{enumerate}

%\subsubsection{Creating superposition over cluster examples in node $j$ at depth $l$} \label{sec:sup_cluster}

For 1, we create the superposition over the examples in node $j \in [k^l]$. Let us denote the nodes by their centroids vectors $c_j^{l}, j \in [k^l]$ where the superscript denotes the depth of the tree and the subscript denotes the node. Next, we perform the Hadamard multiplication (also called element-wise multiplication) between each centroid vector and the weight vector $w$ resulting in the \emph{weighted}-centroids
\begin{equation}
c_{j, w} = w \circ c_j = [w_1c_{j1}, \cdots w_dc_{jd}]^T, \hspace{2mm} j \in [k^l].    
\end{equation}

This operation takes a total time $\mathcal{O}(k^ld)$. The $k^l$ weighted centroids are then loaded in KP-tree (Sec~\ref{sec:KP-tree}) in time $T_{kp}^c = \mathcal{O}(k^ld \log^2(kd))$  and retrieved as amplitude-encoded states in time $\mathcal{O}(\emph{poly} \log(k^ld))$.

For 2, we perform clustering for each node at depth-$l$. For this, similar to Sec~\ref{sec:cd}, the idea is to create the superposition state of the indices of $N$ examples in the dataset and their weighted inner product distance with the $k^l$ available centroid vectors. For simplicity, again let us refer to $\{c_{j}^{l}\}$ as $\{c_{j}\}$. We utilize Theorem \ref{thm:cd} but in this case for the $k^l$ nodes. Therefore, the complexity time is $T_{cd}^{k^l} = \mathcal{O}(T k^l \cdot \frac{\eta_1}{\epsilon_1} \log(1/\Delta))$, where we utilize the super-index $k^l$ to refer to the dependency on the $k^l$ nodes.

Next, we can coherently find the maximum weighted inner product for each example state and the $k^l$ centroid vectors. This allows one to cluster the original data into the disjoint clusters at depth $l$. This is done using Theorem \ref{thm:ca} for $k^l$ nodes. In this case, the complexity time is $T_l^{k^l} = \mathcal{O}(T_{cd}^{k^l} + k^l\log p) = T_{cd}^{k^l}$ due to the same argument as before.

The resulting quantum state is
%The above superposition state in Eq~\ref{eq:label state_depthl} can also be written as
    \begin{equation}
    \small
    \sum_{j=1}^{k^l}\sqrt{\frac{|\mathcal{C}_{j,l}|}{N}}\left(\frac{1}{\sqrt{|\mathcal{C}_{j,l}|}}\sum_{i \in \mathcal{C}_{j,l}}\ket{i}\right)\ket{j} = \sum_{j=1}^{k^l}\sqrt{\frac{|\mathcal{C}_{j,l}|}{N}}|\xi_j\rangle\ket{j}.
    \label{xi:superposition}
\end{equation}

Upon measuring the last register of the above superposition state $\mathcal{O}(k^l \log k^l)$ times, we obtain the cluster index states $|\xi_j\rangle, \forall j \in [k^l]$ with an $\Omega(1)$ probability. Thus each state $|\xi_j\rangle$ can be prepared in time
\begin{equation}
   T_\xi^{k^l}= \mathcal{O}(T_l^{k^l} \cdot \log k^l) = \mathcal{O}(T_l^{k^l} \cdot l\log k).
   \label{eq:txi_depthl}
\end{equation}

Upon querying the $KP$-tree in Theorem~\ref{Thm:KP-tree} separately with each index state $|\xi_j\rangle$, we obtain the state
\begin{equation}
    |\xi_j\rangle\ket{0} \rightarrow \frac{1}{\sqrt{|\mathcal{C}_{j,l}|}}\sum_{i \in \mathcal{C}_{j,l}}\ket{i}\ket{x_i} \hspace{2mm} \forall j \in [k^l],
\end{equation}
where $\ket{x_i} = \frac{1}{\|x_i\|}\sum_{j=1}^d x_{ij}\ket{j}$ in time $T_{\text{cs}}^{k^l} = \mathcal{O}(T_\xi^{k^l}\emph{poly}\log(Nd))$. This then allows us to create the superposition of the example states in the node $j$ at depth $l$.

%\subsubsection{Clustering to expand to depth $l+1$}

%Next, we perform the clustering steps for each parent node $j \in [k^l]$ at depth-$l$ as highlighted in Sec~\ref{sec:clustering} to get the next set of $k^{l+1}$ clusters at depth $l+1$ of the decision tree where each node in depth $l$ creates the set of $k$ clusters at depth $l+1$. It can now be verified that 

Considering steps 1 and 2, the total time to expand all nodes at depth-$l$ to get the $k^{l+1}$ children nodes is the time $T_{\text{cs}}^{k^l}$ to generate the superposition states for each cluster multiplied by the time to do the clustering of each node, which is the same as clustering the root node (i.e., clustering from the root to depth 1 nodes $T_{\text{depth}:[0\rightarrow 1]}$). Therefore, the total complexity time is
\begin{widetext}
\begin{equation}
    \begin{split}
         T_{\text{depth}:[l\rightarrow l+1]} &= \mathcal{O}(k^{l}d + k^{l} \cdot T_{\text{cs}}^{k^l} \cdot  T_{\text{depth}:[0\rightarrow 1]})\\
         &= \mathcal{O}(k^{l}d + k^{l} \cdot T_{\text{cd}}^{k^l} \cdot l \cdot \log k \cdot  \emph{poly}\log(Nd) \cdot KT_{\text{iter-1}}) \\
         &= \mathcal{O}\left(\emph{poly} \log(Nd) \frac{K l k^{2(l+\frac{5}{4})} d \log^2 (k)  \log^2(1/\Delta)\eta^2_1}{\epsilon^2_1 \epsilon_2}\right).
    \end{split}
\end{equation}    
\end{widetext}

The tree keeps growing using the previous strategy until it reaches the maximum allowed depth $D$ or until a stopping criterion is achieved. The total time taken for the algorithm to reach from the root node to the last node at depth $D$ (excluding the time to load the initial data $X, X^T, Y$ into the KP-tree and the time to compute the Pearson correlation weight vector at the root node) is
\begin{widetext}
    \begin{equation}
    \begin{split}
       T_{\text{depth}:[0\rightarrow D]} &=  T_{\text{depth}:[0\rightarrow 1]} +  \sum_{l=1}^{D-1}  T_{\text{depth}:[l\rightarrow l+1]} \\
        &=  \mathcal{O}\left(\emph{poly} \log(Nd) \frac{K D k^{2D + \frac{1}{2}}d \log^2 (k)\log^2(1/\Delta)\eta_1^2}{\epsilon_1^2\epsilon_2}\right).
    \end{split}
    \label{eq:depth0T}
\end{equation}
\end{widetext}

\subsection{Leaf label assignment} \label{sec:majority label}

Once the tree reaches the final depth $D$, it consists of $k^D$ leaf nodes. Now our objective is to compute the label values for each of the leaf nodes, which are clusters $\mathcal{C}_{j, D}, j \in [k^D]$ containing some training samples. For the task of regression, the label value is simply the mean of the label values of these samples i.e.,

\begin{equation}
    \emph{label}_j = \frac{1}{|\mathcal{C}_{j, D}|} \sum_{i \in \mathcal{C}_{j,D}} y_i.
    \label{eq:mean-label}
\end{equation}

For the task of binary regression with label set $\mathcal{M} \in \{0,1\}$, the leaf label value can similarly be computed by computing the mean of the label values of the examples in the cluster. If the mean is less than 0.5, it is assigned the value $0$, and $1$ otherwise.  

The mean value for any leaf node can be calculated by first creating a superposition of the indices in the cluster corresponding to the leaf node, i.e., by creating the index superposition state $|\xi_j\rangle = \frac{1}{\sqrt{|\mathcal{C}_{j, D}|}} \sum_{i \in \mathcal{C}_{j,D}}\ket{i}$. Subsequently, the mean value is simply obtained from the inner product between the index superposition state and label superposition state $|Y\rangle = \frac{1}{\|Y\|}\sum_{i=1}^N y_i \ket{i}$ as
\begin{equation}
\begin{split}
    I(|\xi_j\rangle, |Y\rangle) &= \langle \xi_j | Y \rangle \\
    &= \frac{1}{\sqrt{|\mathcal{C}_{j, D}|}\|Y\|} \sum_{i \in \mathcal{C}_{j,D}} y_i\\
    &= \frac{\sqrt{|\mathcal{C}_{j, D}|}}{\|Y\|} \emph{label}_j.
    \label{eq:label-inner}
\end{split}
\end{equation}

We highlight the steps in greater detail below.

\subsubsection{Creating index superposition states}

The first step of obtaining the leaf node label is to create the superposition over the indices of the cluster $\mathcal{C}_{j,D}$ corresponding to the leaf node $j$ i.e., our objective is to first create the state $ \ket{\xi_j}$ by performing the mapping in Eq.~\ref{xi:superposition} in time $T_{\xi}$ as shown in Eq~\ref{eq:txi_depthl} for $l = D$ as
%\begin{equation}
%    \ket{\xi_j} = \frac{1}{\sqrt{|\mathcal{C}_{j,D}|}}\sum_{i \in \mathcal{C}_{j,D}}\ket{i}.
%\end{equation}

%The above index superposition state can be created 

%\begin{widetext}
\begin{equation}
\begin{split}
 T_{\xi} &= \mathcal{O}(T_{l}^{k^D}\cdot D\log k) \\
 &= \mathcal{O}\left(\emph{poly} \log(Nd) \frac{ D k^{D} \log (k)\log(1/\Delta)\eta_1}{\epsilon_1}\right).
\end{split}
\end{equation} 
%\end{widetext}

\subsubsection{Obtaining mean label value}

Our next step is to query the label superposition state $|Y\rangle$ using Lemma~\ref{lemma:suplabel} by doing
\begin{equation}
      \ket{0}  \rightarrow  \frac{1}{\|Y\|}\sum_{i=1}^N y_i\ket{i} = |Y\rangle.
\end{equation}
This takes time $\mathcal{O}(\emph{poly} \log (N))$. Next, using the following theorem, we can obtain the mean label value for each leaf node $j \in [k^D]$.

\begin{theorem}
    Given access to the characteristic vector states $|\xi_j\rangle$ $\forall j \in [k^D]$ where each state is prepared in time $T_\xi = \mathcal{O}(T_{l}^{k^D}D\log k)$ and the amplitude-encoded states label superposition state $|Y\rangle$ which is be prepared in time $\mathcal{O}(\text{poly} \log(N))$, there exists a quantum algorithm to obtain the mean label values $\{\overline{\text{label}_1}, \cdots, \overline{\text{label}_{k^D}}\}$ such that $|\overline{\text{label}_j} - \text{label}_j| \leq \epsilon_3$, $\forall j \in [k^D]$ in time 
    $$T_{\text{leaf-label}} =\mathcal{O}\left(\frac{T_\xi k^{\frac{3D}{2}}} {\epsilon_3}\right),$$ 
    where $\text{label}_j$ is the true label mean of the examples in the cluster as given in Eq~\ref{eq:mean-label}.
    \label{thm:cent_update2}
\end{theorem}
\begin{proof}
We prove the  Theorem~\ref{thm:cent_update2} in Appendix~\ref{app:cent_update2}.
\end{proof}

\subsection{Time complexity for decision tree construction \label{subsection:main_complexity_algo}}

After showcasing the method to grow the tree and compute the leaf labels, we are now in a position to calculate the total time it takes to build the decision tree. We note that our algorithm has a dependency poly-logarithmic in $N$. However, given the classical dataset, the total time taken to build the decision tree is the time to load the data $X$ into the KP-tree and $Y$ in the classical data structure, plus the time to compute the Pearson correlation coefficient weights plus the rest of the algorithmic steps. 

The time taken to load the classical data in the quantum-accessible-data structure (either KP tree for $X$ or list type data structure for $Y$) is
\begin{equation}
\begin{split}
    T_{\text{load}} &= T_{kp} + T_Y \\
    &= \mathcal{O}(Nd\log^2(Nd) + N\log^2(N))\\
    &= \mathcal{O}(Nd\log^2(Nd)),
\end{split}
\end{equation}
where $T_{kp}, T_Y$ are defined in Lemmas~\ref{Thm:KP-tree},\ref{lemma:suplabel} respectively. 

Next, the time taken to compute the Pearson correlation coefficients is
\begin{equation}
    T_{\text{corr}} = T_{w} = \mathcal{O}\left(\emph{poly} \log(Nd) \frac{d\eta}{\epsilon}\right),
    \label{eq:corr}
\end{equation}
where $T_{w}$ is defined in Eq~\ref{eq:PEtime}.

Therefore, performing supervised clustering consecutively until tree depth $D$ and the leaf label assignment takes time
\begin{widetext}
   \begin{equation}
\small
\begin{split}
        T_{\text{algo}} &=  T_{\text{depth}:[0\rightarrow D]} + T_{\text{leaf-label}} \\
        &= \mathcal{O}\left(\emph{poly} \log(Nd) \frac{K d D  k^{2D+\frac{1}{2}} \log^2(k) \log^2(1/\Delta)\eta_1^2}{\epsilon_1^2\epsilon_2}\right), \\
\end{split}
\label{eq:algo}
\end{equation} 
\end{widetext}
where $ T_{\text{depth}:[0\rightarrow D]}$ and $T_{\text{leaf-label}}$ are defined in Eq~\ref{eq:depth0T} and \ref{eq:leaf_maj} respectively. 

Now, we can combine the three steps to compute the total time taken for the decision tree as
\begin{equation}
      T_{\text{des-tree}} =   T_{\text{load}} + T_{\text{corr}} +  T_{\text{algo}}. 
\end{equation}

We see that $T_{\text{corr}}$ and $T_{\text{algo}}$ have the runtime depending only poly-logarithmically in the number of examples $N$. However, since the initial data loading process in the suitable quantum-accessible-data structures takes time $ T_{\text{load}}$ which depends polynomially in $N$, this makes the initial decision tree construction algorithm to depend linearly in $N$. 

We will see in Section \ref{sec:periodic_update}that the decision tree retraining is extremely fast in the regime of large $N$ due to $T_{\text{corr}} +  T_{\text{algo}}$ depending only poly-logarithmically in $N$. 

\subsection{Test inference with Des-q} \label{sec:tree_inference}

Once the tree is built, we need a mechanism to classify new unseen examples i.e., test unlabelled examples. Here we propose to do this classically based on the following theorem. 

\begin{theorem}
    Given a test example, $x \in \mathbb{R}^d$, there exists a classical algorithm to predict the target label of the example from our proposed decision tree in time complexity
    \begin{equation}
        T_{\text{inf}} = \mathcal{O}(kDd).
    \end{equation}
\end{theorem}

\begin{proof}    Our decision tree relies on the inner product-based clustering method (as described in Sec~\ref{sec:clustering}). For this, we break down the process of assigning a label to a test example as
    \begin{enumerate}
        \item \emph{Initial Depth}: We start by calculating the inner product between the test example $x \in \mathbb{R}^d$ and the $k$ centroids at the first depth of the tree. We select the centroid with the highest inner product value, indicating it is the closest centroid to $x$. The time complexity for this step is $\mathcal{O}(kd)$, as computing the inner product between two $d$-dimensional vectors classically takes $\mathcal{O}(d)$ time.
        \item \emph{Depth Expansion}: We proceed to perform inner product estimations between $x$ and the centroids corresponding to the $k$ children nodes of the parent node (which corresponds to the centroid with the highest inner product from the previous depth layer). Again, we select the centroid with the highest inner product. This process repeats until we reach the maximum tree depth $D$, which is typically a leaf node.
        \item \emph{Leaf Node}: Finally, the target label of the example $x$ is determined by the target label associated with the leaf node that the example is assigned to.
    \end{enumerate}
In summary, the total time complexity for assigning a label to a test example in this tree structure is $\mathcal{O}(kDd)$.
\end{proof}

\section{Retraining with Des-q} \label{sec:periodic_update}

After the tree has been constructed and put online to classify new examples, one would want to retrain the model with new labelled $d$ dimensional examples $N_{\text{new}} \ll N$. This helps prevent their performance degradation which is crucial to maintain the consistent model performance. The key steps to retraining the tree include
\begin{enumerate}
    \item  Load the new $N_{\text{new}}$ examples in the same KP-tree data structure introduced in Sec~\ref{sec:KP-tree}.
    \item Recompute the feature weights for $N + N_{\text{new}}$ training dataset. 
    \item Perform quantum-supervised clustering to grow the tree.
    \item Compute the leaf labels. 
\end{enumerate}

Using the same techniques used to build the tree at the initial time, we showcase in the following theorem that our algorithm for tree retraining scales only poly-logarithmically with the total number of training examples. 

\begin{theorem}
    Given quantum access to previous training examples $X \in \mathbb{R}^{N \times d}$ and new training examples  $X_{\text{new}} \in \mathbb{R}^{N_{\text{new}} \times d}$ new examples such that $N_{\text{new}} \ll N$, there is a quantum decision tree algorithm to retrain the tree with the new examples in time
    \begin{equation}
        T_{\text{retrain}} \approx T_{\text{load-new}} +  T_{\text{corr}} +  T_{\text{algo}},
    \end{equation}
   where  $T_{\text{load-new}} \approx \mathcal{O}(N_{\text{new}}d\log^2(N_{\text{new}}d))$ is the time taken to load the new examples and the corresponding labels in the KP-tree.  $T_{\text{corr}} , T_{\text{algo}}$ are computed in Eq~\ref{eq:corr}, \ref{eq:algo} and both depend poly-logarithmically in the total number of examples $N_{\text{tot}} = N + N_{\text{new}} \approx N$. 
\end{theorem}
\begin{proof}
We structure the proof of the above theorem in the following steps.
\begin{enumerate}
    \item \emph{Loading new examples}: Our Des-q algorithm for retraining would require quantum access to not just the examples $x_i, i \in [N_{\text{tot}}]$, but also the ability to to create the feature vectors $x^{(j)}, j \in [d]$ and the label vector $Y$ over the combined dataset of size $N_{\text{tot}}$. To load the examples $x_i$ in the KP-tree data structure, we create $N_{\text{new}}$ binary trees data structure $B_i^{new}, i \in [N_{\text{new}}]$ as highlighted in Sec~\ref{sec:KP-tree}. Creation of these trees take time $\mathcal{O}(N_{\text{new}}\log^2(N_{\text{new}}d))$. Once these trees are constructed, one can query the examples of the combined dataset $x_i, i \in [N_{\text{tot}}]$ in time  $\mathcal{O}(\emph{poly}\log(N_{\text{tot}}d))$. \\

    It follows the loading of the feature values $x^{(j)}$ of the new dataset $X_{\text{new}}$ into the existing KP-tree data structure consisting of $d$ binary trees $B_j, j \in [d]$. These trees were initially constructed to store the $d$ feature values of $N$ examples in the initial tree construction as mentioned in Lemma~\ref{lemma:supcol}. One can do this by \emph{horizontally} adding the feature values of the new examples into the existing $B_j$ i.e., by expanding each of the $d$ trees to add more leaf nodes which are connected all the way up to the root node. Each binary tree $B_j$ initially had $N$ leaf nodes, which are now appended with $N_{\text{new}}$ new leaves such that the root node of $B_j$ contains the norm value of the $j$-th feature vector of the combined dataset consisting of $N_{\text{tot}}$ examples i.e., $\|x^{(j)}\| = \sqrt{\sum_{i=1}^{N_{\text{tot}}}x_{ij}^2}$. Modifying all the $d$ binary trees takes time $\mathcal{O}(N_{\text{new}}d\log^2(N_{\text{new}}d))$. Once these trees are modified, one can query the feature states $x^{(j)}$ of the combined dataset in time $\mathcal{O}(\emph{poly}\log(N_{\text{tot}}d))$. \\

    Similarly, we need to modify the KP-tree structure $B$ which stores the label values as introduced in Lemma~\ref{lemma:suplabel}. The tree is modified by adding $N_{\text{new}}$ new leaf nodes such that the root of the tree stores norm value $\|Y\| = \sum_{i=1}^{N_{\text{tot}}}y_{i}^2$. Modifying the tree takes time  $\mathcal{O}(N_{\text{new}}\log^2(N_{\text{new}}))$ and the label state $\ket{Y}$ for the combined dataset can be retrieved in time $\mathcal{O}(\emph{poly}\log(N_{\text{tot}}))$. \\

    Thus the total time to load the examples, features and the label for the dataset $N_{\text{new}}$ into the KP-tree data structure is
    \begin{equation}
    \begin{split}
        T_{\text{load-new}} &= \mathcal{O}(N_{\text{new}}d\log^2(N_{\text{new}}d)\\ &+ \mathcal{O}(N_{\text{new}}\log^2(N_{\text{new}} ))\\
        &= \mathcal{O}(N_{\text{new}}d\log^2(N_{\text{new}}d)). \\
    \end{split}
    \end{equation}

    \item \emph{Feature Weights + Supervised Clustering + Leaf Label}: Once the new data has been loaded into the KP-tree data structure, it can be retrieved in time poly-logarithmic in the total number of examples $N_{\text{tot}}$. Subsequently, one can apply the techniques introduced in Sec~\ref{sec:bivariatecorr}, Sec~\ref{sec:clustering}, Sec~\ref{sec:Tree_growth}, and Sec~\ref{sec:majority label} respectively to build the decision tree and compute the leaf labels. The time to generate the feature weights for the combined dataset is
    \begin{equation}
    \begin{split}
     T_{\text{corr-new}} &= \mathcal{O}\left(\text{poly} \log(N_{\text{tot}}d) \frac{d\eta}{\epsilon}\right)\\
     &=  \mathcal{O}\left(\text{poly} \log(Nd) \frac{d\eta}{\epsilon}\right)\\ &= T_{\text{corr}}.
    \end{split}
    \end{equation}
    And the time to do the supervised clustering and leaf label assignment is\\

\begin{widetext}
   \begin{equation}
\small
\begin{split}
        T_{\text{algo-new}} 
        &= \mathcal{O}\left(\emph{poly} \log(N_{tot}d) \frac{ D k^{2D}  \log (k) \log(1/\Delta)\eta_1}{\epsilon_1}\left(\frac{K d k^{\frac{1}{2}}\log (k) \log(1/\Delta)\eta_1}{\epsilon_1\epsilon_2}\right)\right)\\
        &=  T_{\text{algo}}. \\
\end{split}
\end{equation} 
\end{widetext}

    %\begin{widetext}
    %\begin{equation}
    %\small
    %\begin{split}
        %T_{\text{algo-new}}
        %&= \mathcal{O}\left(\emph{poly} \log(N_{\text{tot}} d) \frac{ Dk^{3D}d \log (k)\log(p)\log(1/\Delta)\eta_1}{\epsilon_1}\left(\frac{K\log (k)\log(p\%log(1/\Delta)\eta_1\eta_2}{\epsilon_1\epsilon_2}+ \frac{\eta_3}{\epsilon_3}\right)\right) \\
%        &=  T_{\text{algo}} \\
%    \end{split}
%    \end{equation}    
%    \end{widetext}  
\end{enumerate}

Thus the total time taken to retrain the decision tree is
\begin{equation}
\begin{split}
     T_{\text{retrain}} &= T_{\text{load-new}} +  T_{\text{corr-new}} +  T_{\text{algo-new}} \\
     &= T_{\text{load-new}} +  T_{\text{corr}} +  T_{\text{algo}}.
\end{split}
\end{equation}

From this, we see that our decision tree for retraining has a linear dependence on $N_{\text{new}}$ due to the $T_{\text{load-new}} $ factor, and poly-logarithmic dependence on $N_{\text{tot}} \approx N$ which comes from the term $T_{\text{corr-up}} +  T_{\text{algo-up}}$. Under the small batch updates criterion $N_{\text{new}} \ll N$, we see that our algorithm for tree retraining has scales poly-logarithmically in $N$, thus ensuring that the tree update is extremely fast. 
\end{proof}

%Discuss more practical considerations related to implementing Des-q on current quantum hardware, such as noise, qubit requirements, and how the algorithm might scale with larger datasets.

\section{Numerical results} \label{sec:numerics}

This section aims to demonstrate that the proposed quantum decision tree construction method, Des-q, achieves competitive performance compared to state-of-the-art axis-parallel split-based methods commonly employing impurity measures or information gain for classification tasks and variance reduction for regression tasks. In our numerical simulations, we utilize the classical version of Des-q denoted as Des-c, which leverages the $k$-means clustering algorithm. The performance in terms of the accuracy of predictions of $k$-means closely aligns with the robust $\delta$-k-means algorithm, as demonstrated by Kerenidis \etal\cite{kerenidis2019q}, under appropriate $\delta$ selection. Further, $\delta$-k-means is a good approximation of the performance of $q$-means algorithm. Therefore, the performance in accuracy of predictions of Des-c closely aligns with the performance of Des-q executed on quantum hardware.

As discussed in Sec~\ref{sec:bivariatecorr}, the core component of the Des-q algorithm involves quantum estimation of the point-biserial correlation for binary classification and the Pearson correlation for regression, with numerical labels. For our numerical simulations, we calculate these values classically. For more technical details about the implementation, we refer to Appendix \ref{appendix:details_numerics}. 

Note that we focus on benchmarking the performance in the quality of predictions and we do not include comparisons of runtime because Des-c has no speedup over classical standard methods for decision tree construction and retraining as it utilizes the classical $k$-means algorithm. Therefore, in contrast to the quantum algorithm, Des-c scales polynomially in $N$ and $d$. The speedup in Des-q comes from leveraging $q$-means for construct and retrain the tree.

\textbf{Limitations}. The implementation of Des-q requires the use of quantum-accesible-data structures: KP tree for $X$ and list type data structure for $Y$. We refer the reader to Jaques \etal \cite{jaques2023qram} and Allcock \etal \cite{allcock2023constant} for a discussion on implementations. The speedup of Des-q is in the asymptotic regime. Therefore, it is needed that these structures handle large datasets, which can be challenging. However, some erasing techniques may be used or compression methods (e.g., with Principal Component Analysis) over old data to allow for space for more recent data. In terms of circuit requirement, since we perform quantum clustering in a sequential manner to build the tree, it is anticipated that the resulting circuits will be quite deep, necessitating either extended coherence times or the use of error correction techniques. In terms of number of qubits, this will depend on the data size $d$. Although this number may be big, the proposed algorithm targets datasets where $N \gg d$.

We hope that the proposed quantum method may trigger interesting future works that provide a quantitative resource analysis of hardware requirements and potential hardware demonstrations, as the hardware keeps on developing.

\subsection{Main results}

We present our results using four datasets comprising numerical features: PIMA \cite{pima-indians-diabetes}, Spambase \cite{misc_spambase_94}, Blood \cite{misc_blood_transfusion_service_center_176}, and Boston housing \cite{harrison1978hedonic}. Among these, the first three datasets are binary classification tasks, while Boston housing has numerical label values making it suitable for the task of regression. Their characteristics are shown in Table \ref{table:datasets} where \emph{Instances} denotes the number of labeled examples, and \emph{Features} denotes the number of features ($d$) for each example. For all the considered datasets, we split the data into train and test using a ratio of $0.3$ and created ten train-test of equally sized subsets of data, called folds, for binary classification dataset and five for the regression dataset. These folds maintain the proportion of labels. Subsequently, each fold is used to train the decision tree model and the performance is evaluated in both train and test sets. We perform standard data normalization techniques on each fold, ensuring that each feature has a mean of 0 and a standard deviation of 1.

\begin{table}[h!]
  \centering
    \begin{tabular}{|p{0.8in}|p{0.6in}|p{0.6in}|p{0.6in}|} \hline 
      \cline{1-4} 
    Task & Dataset & Instances & Features \\ \cline{1-4} 
    \multirow{3}{*}{}
     & PIMA & 768 & 8  \\ \cline{2-4} 
    Classification & Spambase & 4601  & 57  \\ \cline{2-4} 
    %BUPA & 245  & 7 & $\mu=0.17, \sigma=0.55\mu$  \\ \cline{1-4} 
     & Blood & 748  & 4   \\ \cline{1-4} 
    Regression & Boston housing & 506  & 8   \\ \cline{1-4} 
    \hline    
    \end{tabular}
\caption{Characteristics of the datasets used for the performance benchmark.  \label{table:datasets}}
\end{table}

We compare the performance of Des-c for multi-way decision tree construction against the state-of-the-art decision tree construction method. In particular, we employ the DecisionTree class for classification and the DecisionTreeRegressor class for regression, with both classes being implemented in scikit-learn \cite{scikit-learn}. We set the criterion, the impurity measure, to be the entropy for classification and mean square error for regression. The only hyperparameter of this method that we modify for the experiments is tree depth. For Des-c we evaluate its performance for different number of clusters $k$, ranging from two to seven, depending on the dataset. We include in the data availability section below the link to a repository with the post-processed data utilized for the experiments, the results obtained and the Jupyter notebooks to reproduce the results in the article.

The main performance results of Des-c against the baseline decision tree method are highlighted in Table \ref{table:main_results_numerics}. The metric used to benchmark the performance, shown in column \emph{Performance}, is the accuracy for the classification task and the root mean square error (RMSE) for regression, which is defined as: 
\begin{equation}
    RMSE=\sqrt{\frac{1}{N}\sum_{i=1}^N (Y_i - \hat{Y_i})^2},
\end{equation}
where $N$ is the number of samples, $Y_i$ is the label of the $i$-th sample and $\hat{Y_i}$ is the predictor obtained for that same sample.

\begin{table*}[t]
  \centering
    \begin{tabular}{|p{0.9in}|p{1in}|p{0.8in}|p{0.8in}|p{0.4in}|p{0.4in}|p{0.5in}|} \hline
    Task & Dataset & Model & Performance & Tree size &Tree depth &Clusters  \\ \cline{1-7} 
    
    \multirow{3}*{}  & \multirow{2}*{PIMA} &Baseline &  $74.64\%\pm 2.81$  & 7 & 2 & 2 \\ \cline{3-7} 
     &  & Des-c &  $70.34\%\pm 4.53$    & 52.2 & 2 & 7\\ \clineB{2-7}{2.5} %k=7

     & \multirow{2}*{Spambase} & Baseline &    $81.89\%\pm 2.32$  & 7 & 2 & 2\\ \cline{3-7} 
   Classification & & Des-c & $80.71\%\pm 9.38$ & 27  &  2 & 5 \\ \clineB{2-7}{2.5}
        
    %&Baseline & \npboldmath $81.89\pm 2.32$  & 7 & 2 & 2 \\ \cline{2-7} 
    %& & Des-c &  $80.71\pm 9.38$    & 27 & 2 & 5\\ \clineB{2-7}{2.5} %k=5
    
    & \multirow{2}{*}{Blood} & Baseline &  $77.07\%\pm 2.10$  & 7 & 2 & 2 \\ \cline{3-7} 
    & & Des-c & $77.26\%\pm 2.05$  & 12.8  & 2  & 3\\ \clineB{2-7}{2.5} %k=3
    \hline

    \multirow{2}*{Regression} & \multirow{2}*{Boston Housing} & Baseline & $0.053 \pm0.007$ & 7 & 2 & 2  \\ \cline{3-7} 
    & & Des-c & $0.053 \pm0.007$   & 20 & 2&4 \\ \cline{2-7}   %k=4
    \hline 
    
    \end{tabular}
\caption{Main results of the performance in the test for two tasks: binary classification and regression and four datasets. The performance metric utilized for classification is the accuracy (in $\%$) and the one for regression is the RMSE. The values correspond to the average and the standard deviation across the folds considered. We compare the classical implementation of the quantum proposed algorithm, referred as Des-c, to the baseline, a binary decision tree built based on the reduction of entropy for classification and mean-square error for regression. For classification, Des-c incorporates the point-biserial correlation as a feature weight whereas for regression, it utilizes the Pearson correlation. Note that beyond depth $2$ models start to overfit. The results reported correspond to the minimal number of clusters for which Des-c performs on par with the baseline. The tree size is defined as the average number of tree nodes across the folds.\label{table:main_results_numerics}}
\end{table*}

The results show that the performance of the proposed method, Des-c, matches the baseline across all analyzed datasets. This comparison is made at the same tree depth, affirming that Des-c competes effectively with the baseline. We also conducted evaluations with deeper trees, extending up to $5$ for classification tasks and $10$ for regression tasks. It becomes evident that for depths beyond $2$, all models exhibit a decline in test accuracy, signifying the onset of overfitting. Furthermore, it is noteworthy that while the baseline relies on binary splits, Des-q employs multi-way splits, and the optimal results were achieved with a number of clusters (denoted as \emph{Clusters} in the table) greater than two. Importantly, an increase in the number of clusters ($k$) does not introduce significant complexity to the algorithm (as indicated by Eq~\ref{eq:algo}), as $k \ll N$, where $N$ corresponds to the number of training samples. In summary, the proposed method demonstrates competitive performance when compared to state-of-the-art approaches.

In the following subsections, we provide a more detailed examination of the results for both the classification task (Sec~\ref{subsec:numerics_classification}) and regression (Sec~ \ref{subsec:numerics_regression}). Not only do we compare Des-c to the baseline but also to the same method but without incorporating feature weight in order to highlight the benefits that the incorporation of these weights in the distance metric (Eq \ref{ec_distance_algo}), part of the supervised clustering, brings about. In each subsection, we assess the performance of tree inference with test data. Furthermore, we discuss the tree construction performance, specifically how well it reduces the impurity measure (entropy for classification and variance for regression) and its accuracy with the training data. This analysis is crucial to determine whether the tree construction aligns with expectations, which involve reducing the impurity measure, despite the fact that the proposed method does not directly optimize this measure.

\subsection{Classification} \label{subsec:numerics_classification}

For this task, Des-c utilizes the point-biserial correlation as the feature weight. We compare this approach against using the same method without any feature weights, which corresponds to setting $w_l = 1$ for all $l \in [d]$ in the distance calculations within Des-c (as shown in Eq \ref{eq:weighted-inner-product}). We evaluate the performance of the proposed algorithm against a state-of-the-art method, which optimizes the splits to reduce entropy. Although Des-c does not directly optimize splits based on entropy, we expect to observe a reduction in entropy as the tree grows. We define the entropy of a tree trained up to depth $D$ as the weighted sum of the entropy of its individual leaves, following this formula: $E_D = \sum_{i=1}^n f_i e_i$, where $f_i$ represents the fraction of samples in the $i$-th node, $e_i$ denotes the entropy at that node, and $n$ is the total number of leaves in the tree.

\subsubsection{Tree construction: performance with training data} 
The conclusions we will discuss regarding tree construction have been observed for the three datasets considered. To facilitate visualization and discussion, we will focus on the results for the PIMA dataset, but it is important to note that these conclusions apply across all datasets. We begin our analysis with a depth one ($D=1$), which corresponds to the root node and the first level of children nodes. In Fig. \ref{fig:clustering_classification}, we compare the entropy ($E_{D=1}$) obtained with Des-c, incorporating feature weights (point-biserial correlation), to the entropy obtained without weights (\emph{no weight}) as a function of the number of clusters. Additionally, we compare these results to the entropy obtained by the binary decision tree (baseline method) at depth one.

We observe an overall reduction in entropy as the number of clusters increases. As expected, the inclusion of label information for supervised clustering results in a more pronounced reduction in entropy as the tree depth increases, as illustrated in Figure \ref{fig:training_classification}(a). In practice, while not incorporating feature weights results in an entropy that is not significantly different from the baseline, incorporating feature weights reduces the entropy even more than the baseline. This improvement in entropy improves the performance in training accuracy where Des-c achieves competitive and even better performance compared to the baseline for some higher values of the number of clusters ($k \in [4,5,6,7]$) and for higher tree depths ($d \in [3,4,5]$), as shown in Figure \ref{fig:training_classification}(b).

\begin{figure}[h!]%
    \centering
    %\subfloat[\centering]{{\includegraphics[width=7.5cm]{pictures_final/pima/pearson_distribution.png} }}%
    %\qquad
    \includegraphics[width=8 cm]{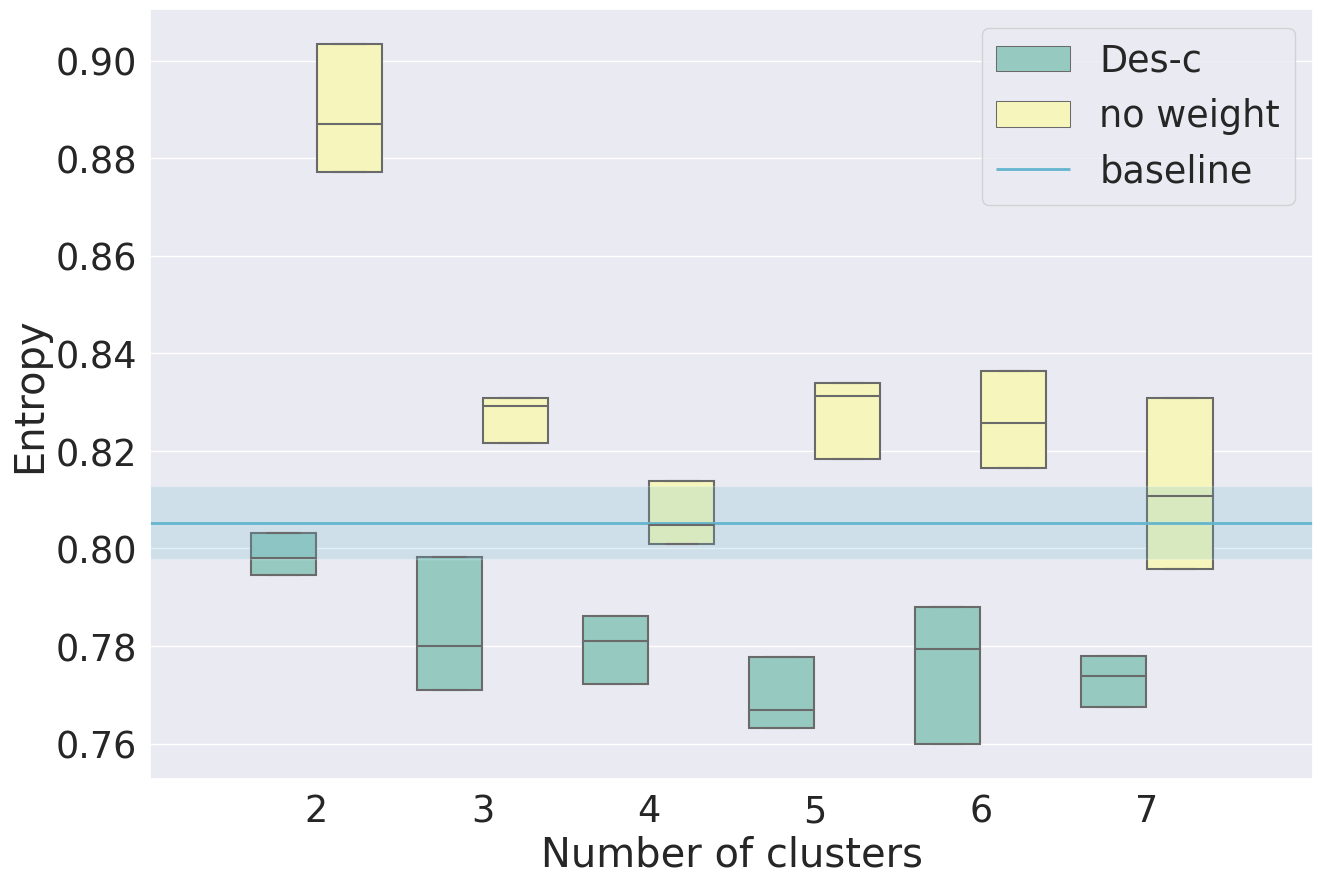} 
    \caption{Entropy as a function of the number of clusters. Des-c is compared with the same method but without weight (\emph{no weight}). The boxes represent the statistics over the ten folds considered. We include the baseline, which is the entropy calculated by the baseline method. The shaded area corresponds to the standard deviation of the median entropy obtained by the baseline method (the median value is shown with the solid line).} 
    \label{fig:clustering_classification}%
\end{figure}

We assess the performance of trees constructed using Des-c across various depths and with different cluster sizes ($k$), comparing the results with the method lacking feature weighting (\emph{no weight}) using the same cluster count, as well as the baseline binary-split method. Figures \ref{fig:training_classification} (a) and (b) showcase the entropy and accuracy, respectively, for visualization purposes, plotting the mean values while omitting the standard deviation as error bars. In Table \ref{table:main_results_numerics}, which presents the primary outcomes reflecting the best performance, we include the standard deviation.

\begin{figure*}[]%
    \centering
    \subfloat[\centering]{{\includegraphics[width= 7.8 cm]{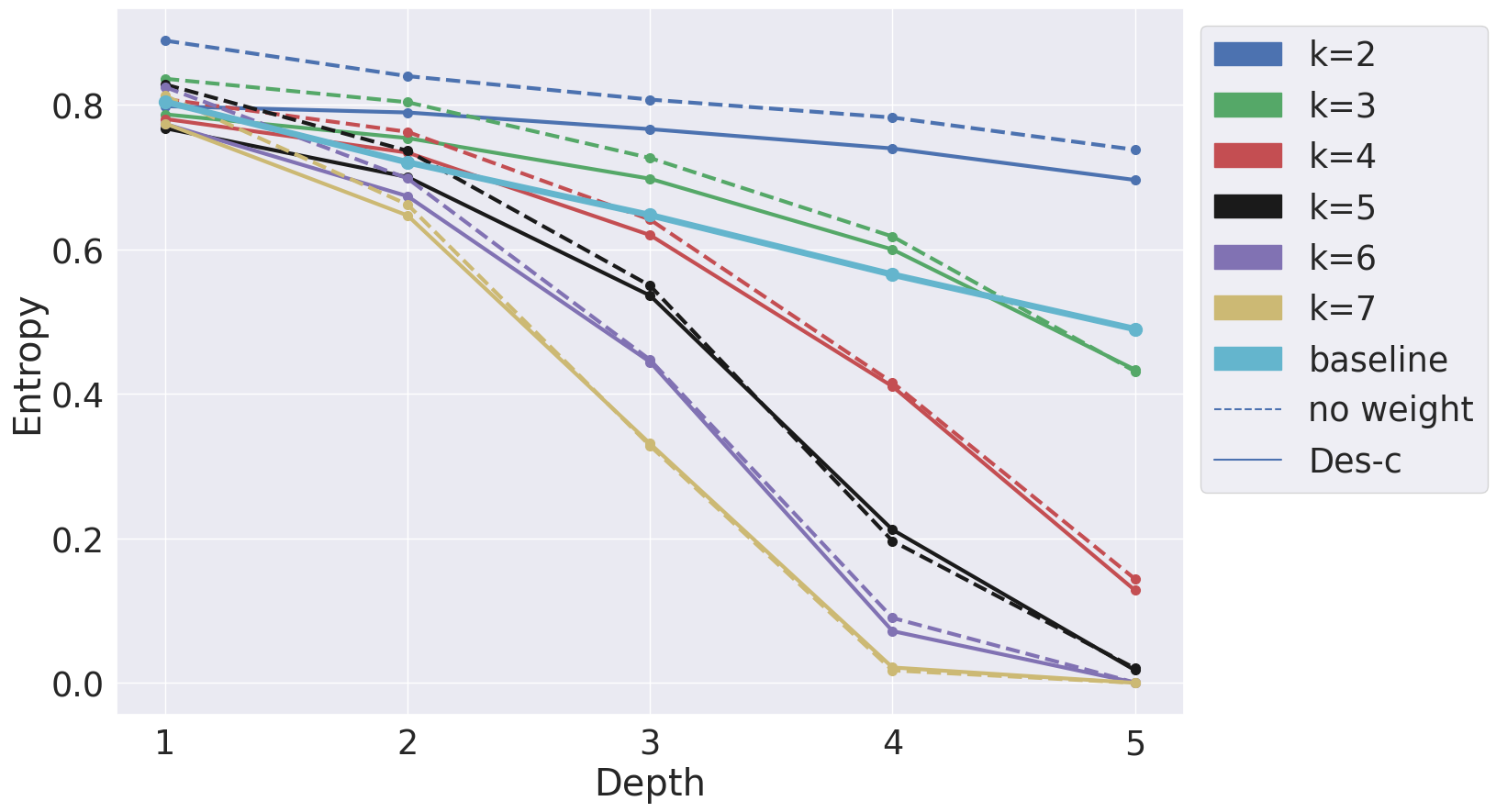} }}%
    \qquad
    \subfloat[\centering]{{\includegraphics[width=7.8 cm]{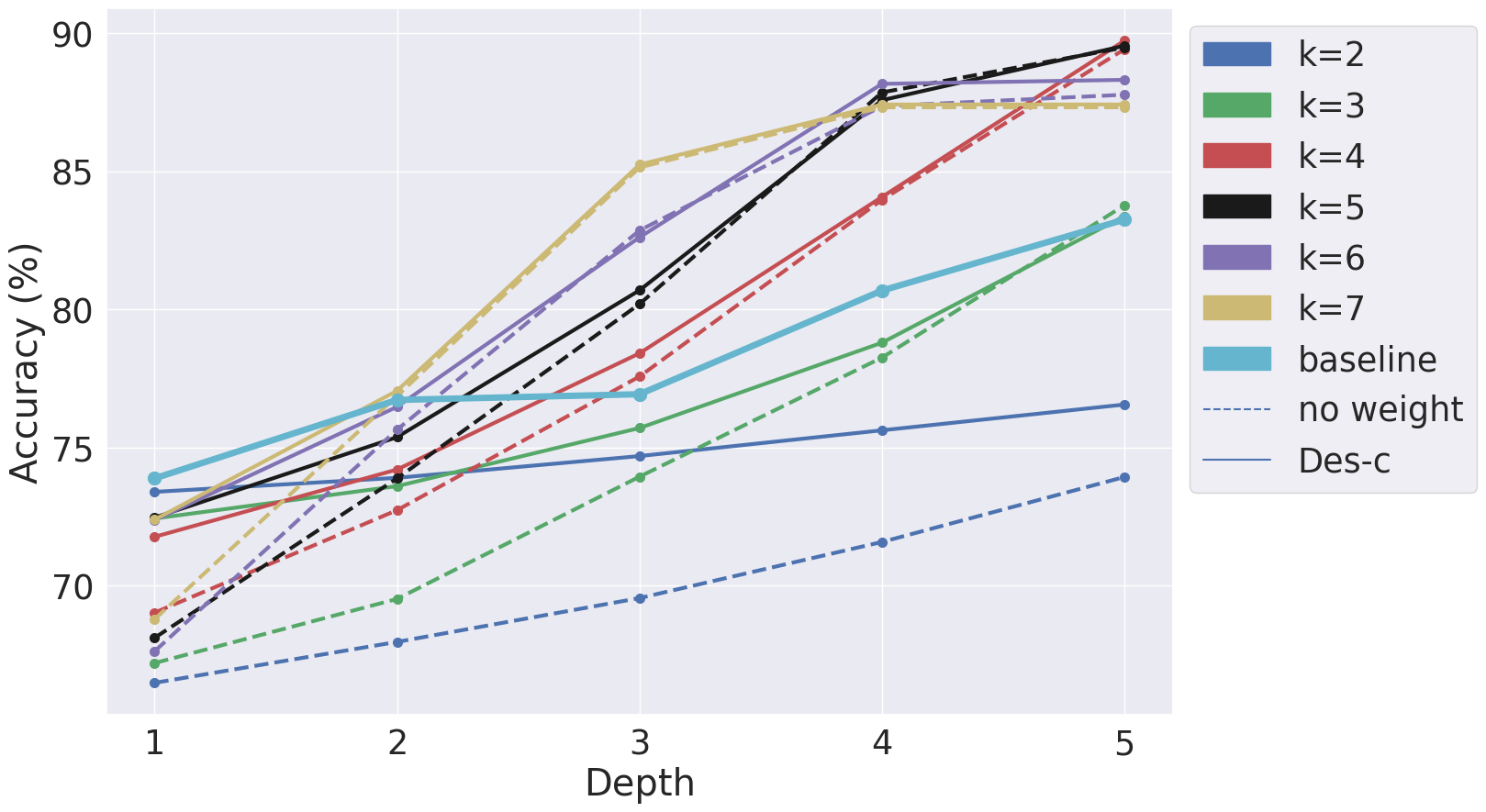} }}%
    \caption{Performance in the training of the decision trees for the PIMA dataset: (a) shows the entropy at each depth as a function of the tree depth and (b) shows the accuracy (in $\%$) as a function of the tree depth. It is compared Des-c (solid line) to the same method but with no weight  (\emph{no weight}) (dotted line) for different numbers of clusters ($k$), shown in colors. We compared against the baseline (cyan solid line). The values shown correspond to the mean across the folds. We avoid plotting the error bars (standard deviation) for visualization purposes and error values are in line with the values in Table \ref{table:main_results_classification}. } %
    \label{fig:training_classification}%
\end{figure*}

As depicted in Fig.\ref{fig:training_classification} (a), entropy consistently improves (i.e., decreases) as the depth increases across all methods. A higher cluster count ($k$) also contributes to improved entropy in our methods, and for high values of $k$, we even observe better entropy at a fixed depth compared to the baseline, suggesting the potential advantages of Des-c. Furthermore, we observe comparable or better performance when incorporating feature weights. This demonstrates that our proposed method effectively reduces entropy as nodes are split and the tree grows, aligning with our expectations. In terms of accuracy (Figure \ref{fig:training_classification} (b)), Des-c consistently achieves higher accuracy than the method without feature weights for a fixed number of clusters ($k$). All three methods exhibit improved accuracy as the tree depth increases. Notably, for $D \geq 2$, within the overfitting regime, Des-c surpasses the baseline in terms of accuracy. Overall, our findings indicate that, in terms of training performance, the proposed Des-c method competes effectively with the baseline.

\subsubsection{Tree inference: performance with test data} 
In Table \ref{table:main_results_classification}, we present the results for the three datasets at depths below the over-fitting regime ($D \leq 2$), considering different numbers of clusters when evaluating Des-c. The table displays results corresponding to the minimum $k$ value for which Des-c achieves performance comparable to the baseline. Additionally, we provide a comparison with the baseline, which employs binary splits for all depths. For completeness, we report the accuracy obtained with these specific trees on the training data.
Our findings indicate that, across the three datasets, the test accuracy achieved with Des-c closely aligns with the accuracy achieved with the baseline for the depths considered. To approach or match the baseline accuracy, it is necessary to increase the number of clusters beyond binary splits when using Des-c. This explains the higher numbers for tree size, i.e., total number of nodes, utilized by the proposed method.

\begin{table*}[t]
  \centering
    \begin{tabular}{|p{0.6in}|p{0.7in}|p{1.2in}|p{1.35in}|p{0.4in}|p{0.35in}|p{0.5in}|} \hline 
      \cline{3-5} 
    Dataset & Model & Test Acurracy (\%)  & Training Accuracy (\%) &Tree size &Tree depth &Clusters per depth  \\ \cline{1-7} 
    
    \multirow{3}{*}{PIMA} 
    %% at depth 1 
    & Baseline &   $73.51\pm 4.36$  &  $73.88\pm 0.51$ & 3 & 1 & 2\\ \cline{2-7} 
    %&no weight & $68.2\pm 6.57$ & $66.48\pm 1.93$  & 3 &  & 2\\ \cline{2-7} % is it k=2
    &  Des-c  & $69.9\pm 6.41$ & $73.39\pm 2.45$   & 3  & 1 & 2\\  \clineB{2-7}{2.5} %k=2
    %&no weight & $65.31\pm 4.06$ & $68.78\pm 1.31$   &8  &  & 7\\  \cline{2-7} %k=7
    %&  Des-c  & $69.22\pm 3.68$ & $72.39\pm 1.03$ & 8  &  & 7\\ \clineB{2-7}{2.5} %k=7
    
    % here it is depth 2
    &Baseline &  $74.64\pm 2.81$ & $76.72\pm 1.1$ & 7 & 2 & 2 \\ \cline{2-7} 
    %&no weight & $71.56\pm 8.51$ & $73.88\pm 2.43$  & 31  &  & 5 \\ \cline{2-7} %k=5
    %&  Des-c  & $68.76\pm 4.05$ & $75.37\pm 1.03$  & 31  & 2 & 5\\ \cline{2-7} %k=5
    %&no weight & $70.51\pm 5.91$ & $76.87\pm 1.13$  & 54.6 &  & 7\\ \cline{2-7} %k=7
    &  Des-c  &  $70.34\pm 4.53$ & $77.05\pm 0.81$   & 52.2 & 2 & 7\\ \clineB{2-7}{2.5} %k=7
    \hline 
    
    \multirow{3}{*}{Spambase} 
    %% at depth 1 
    & Baseline &   $74.97\pm 12.1$  &  $79.44\pm 0.13$ & 3 & 1 & 2\\ \cline{2-7} 
    %&no weight & $66.34\pm 10.31$ & $65.82\pm 9.34$  & 5 &  & 4\\ \cline{2-7} % is it k=4
    %&  Des-c  & $73.35 \pm 12.07$ & $73.3\pm 12.06$   & 5  & 1 & 4\\ \cline{2-7} %k=4
    %&no weight & $65.56\pm 8.45$ & $66.09\pm 9.59$   &6  &  & 5\\  \cline{2-7} %k=5
    &  Des-c  & $75.47\pm 11.78$ & $75.01\pm 11.3$ & 6  & 1 & 5\\ \clineB{2-7}{2.5} %k=5
    
    % here it is depth 2
    &Baseline &  $81.89\pm 2.32$ &  $82.59\pm 0.13$ & 7 & 2 & 2 \\ \cline{2-7} 
    %&no weight & $72.8\pm 11.38$ & $71.45\pm 10.23$  & 18.3  &  & 4 \\ \cline{2-7} %k=4
    %&  Des-c  & $80.25\pm 9.52$ & $80.15\pm 9.25$  & 17.7  & 2 & 4\\ \cline{2-7} %k=4
    %&no weight & $80.5\pm 6.43$ & $79.71\pm 6.38$  & 24.6 &  & 5\\ \cline{2-7} %k=5
    &  Des-c  &  $80.71\pm 9.38$ & $79.61\pm 9.02$   & 27 & 2 & 5\\ \clineB{2-7}{2.5} %k=5
    \hline

    \multirow{3}{*}{Blood} 
    %% at depth 1 
    & Baseline  & $77.63\pm 0.82$  &   $77.63\pm 0.09$   & 3 & 1 & 2\\ \cline{2-7} 
    %&no weight & $77.63\pm 1.17$ & $77.71\pm 0.15$ &  4 & 1 & 3\\ \cline{2-7} % is it k=3
    &  Des-c  & $76.87\pm 1.5$ &  $77.74\pm 0.18$   & 4  & 1 & 3\\ \clineB{2-7}{2.5} %k=3
    
    % here it is depth 2
    &Baseline &  $77.07\pm 2.10$ &  $78.05\pm 0.7$ & 7 & 2 & 2 \\ \cline{2-7} 
    %&no weight & $77.25\pm 1.48$ & $77.88\pm 0.16$   & 12.8  & 2  & 3 \\ \cline{2-7} %k=3
    &  Des-c  & $77.26\pm 2.05$ & $77.88\pm 0.19$   & 12.8  &  2 & 3\\ \cline{2-7} %k=3
    \hline
    \end{tabular}
\caption{Benchmark of binary classification for the PIMA, Spambase and Blood datasets. We compare the proposed method Des-c to the baseline for two tree depths: $1$ and $2$, below the over-fitting regime. It is shown the mean values and their standard deviations for the test and training accuracy ($\%$). The tree sizes correspond to the average of the sizes among the trees constructed in each fold. \label{table:main_results_classification}}
\end{table*}

%\st{One of the reasons for this could be that after having constructed the tree, we need to assign a value to each leaf node. We take the majority vote. In order for this assignment to be representative, there is need a big enough number of training samples assigned to the leaf nodes. This limits up to which point the tree can continue growing and still obtain good results.} \PM{If the baseline is overfitting, I'm fine having our model to overfit too}

\subsection{Regression} \label{subsec:numerics_regression}

For the regression task, our goal is to minimize variance as the tree nodes are split. The variance of a tree trained at a specific depth $D$ is calculated as follows: $V_D = \sum_{i=1}^n f_i Var(\hat{Y_i})$, where, $f_i$ is the fraction of samples in the $i$-th node, $\hat{Y_i}$ is the predictor of the sample, $Var$ refers to the variance and $n$ is the total number of leaves in the tree. 

\subsubsection{Tree construction: performance with training data} 
Similar to our approach for classification, we start by analyzing the performance at a depth of one. In Figure \ref{fig:regression_clustering_pearson}, we compare the variance achieved with Des-c and the same method without weight (labeled as ``no weight") as a function of the number of clusters ($k$), and we compare these results against the baseline. We observe that as we increase the number of clusters, the variance decreases and approaches the baseline. Specifically, for $k=5$ and $k=6$, the variance obtained with Des-c overlaps with the baseline. This observation indicates that Des-c becomes increasingly competitive with the baseline as the number of clusters grows.

\begin{figure}[h!]
    \centering
    %\subfloat[\centering]{{\includegraphics[width=7.5cm]%{pictures_final/regression/probability_density_pearson.png} }}%
    %\qquad
    \includegraphics[width=8 cm]{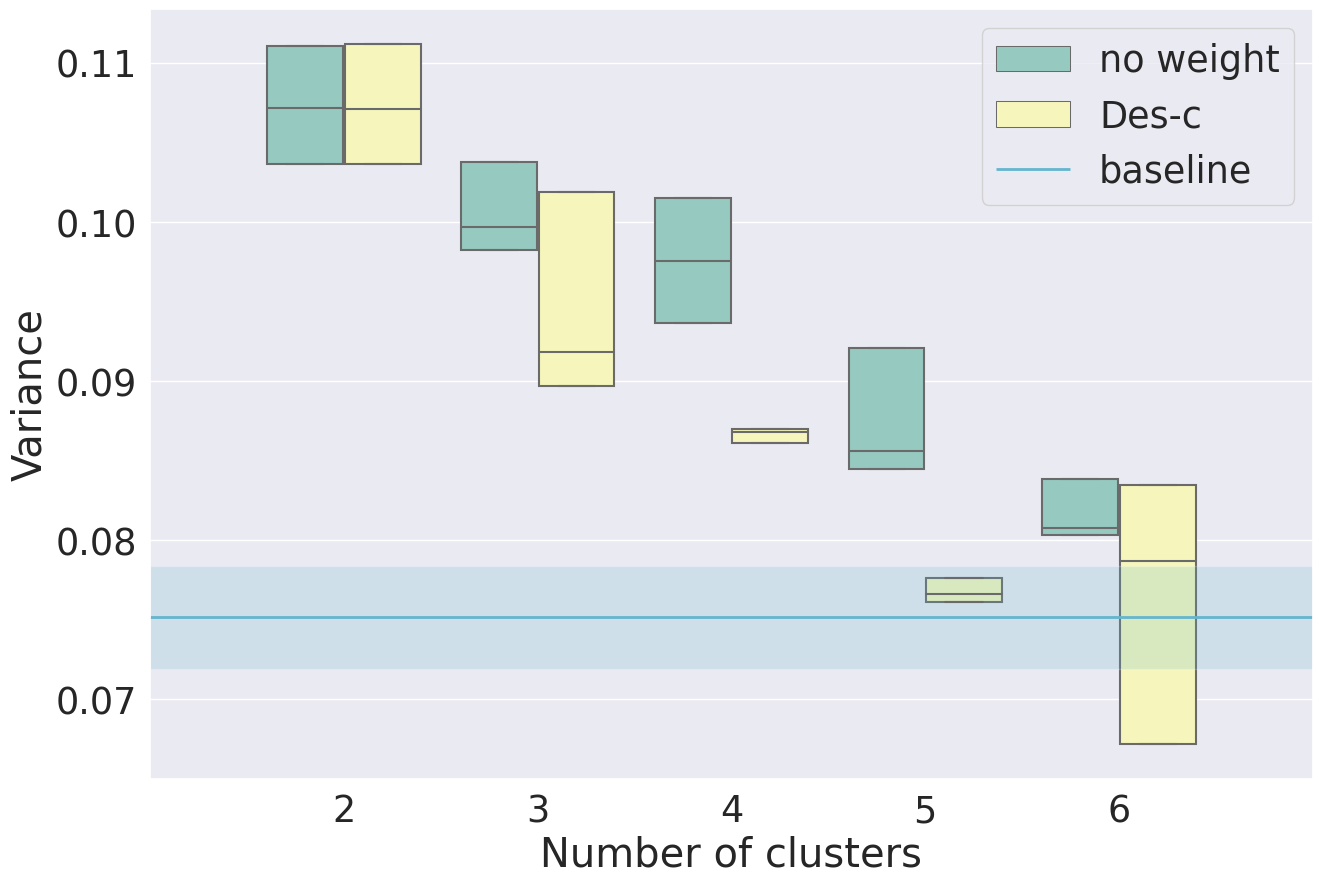}%
    \caption{Variance as a function of the number of clusters. It is compared Des-c and the same method but without weight ( \emph{no weight}). The boxes correspond to the statistics over the ten folds considered. We include the baseline, which is the variance calculated by the baseline method. The shaded area corresponds to the standard deviation of the variance corresponding to the baseline method (the median value is shown with the solid line). }%
    \label{fig:regression_clustering_pearson}%
\end{figure}

In Fig. \ref{fig:regression_training}, we compare the performance of the decision trees as a function of the tree depth in the training dataset. Similar to the regression scenario, the performance of the models benefits from an increase in depth or the number of clusters. As $k$ increases, we see in Fig. \ref{fig:regression_training} (a) that the variance drops even more rapidly with Des-c as the depth increases than the method with no weight, showing the improvement of tree construction by incorporating the feature weights. Regarding the RMSE, as shown in Figure \ref{fig:regression_training} (b), all three methods experience a reduction in RMSE as the tree depth increases. The RMSE values obtained with Des-c are competitive with those of the baseline. Furthermore, it is evident that the proposed method, Des-c, achieves a greater reduction in RMSE compared to the case without weight. Consequently, in terms of training performance, we can conclude that the proposed method, Des-c, competes effectively with the baseline.

\begin{figure*}%
    \centering
    \subfloat[\centering]{{\includegraphics[width=7.8cm]{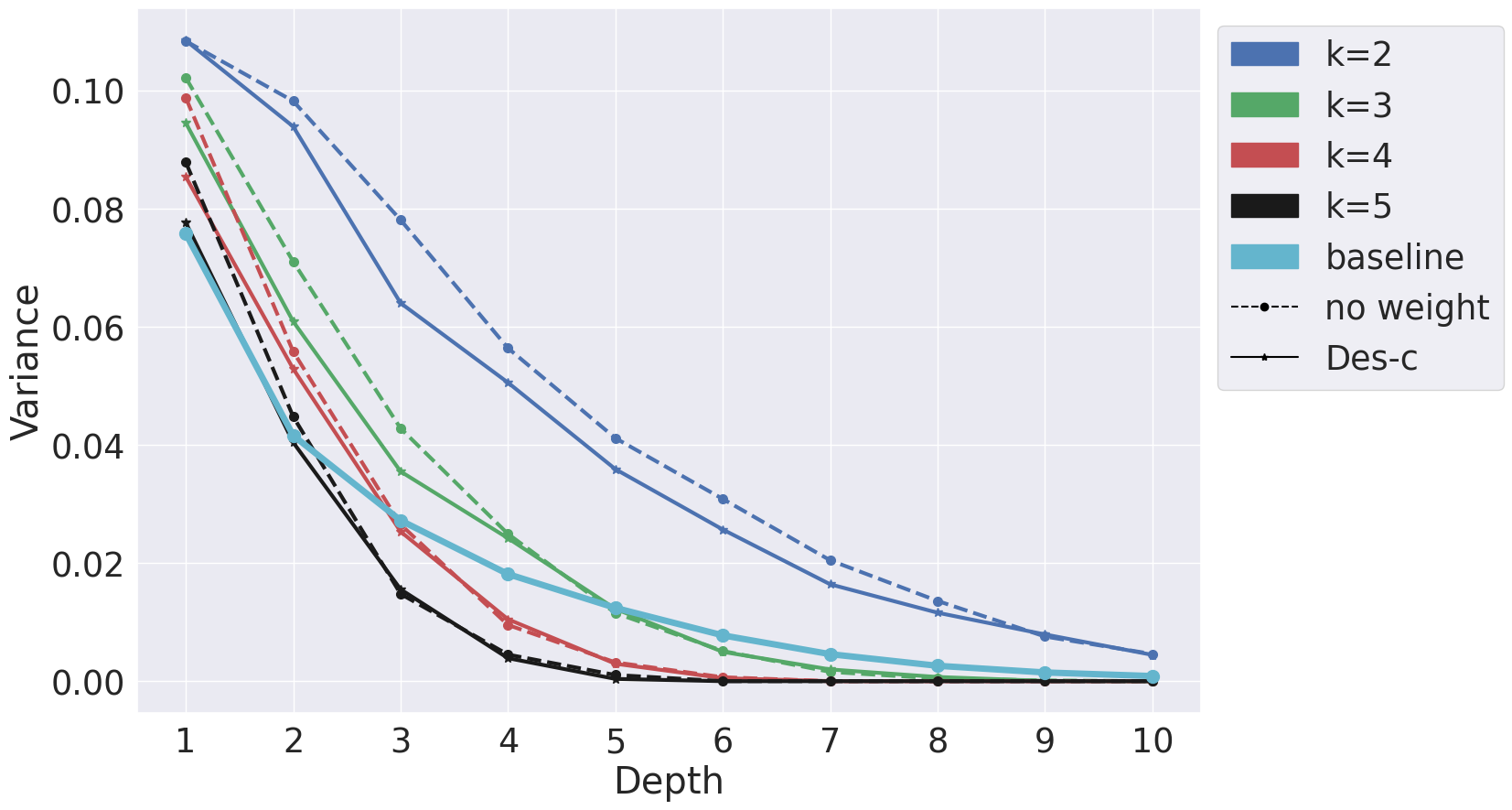} }}%
    \qquad
    \subfloat[\centering]{{\includegraphics[width=7.8cm]{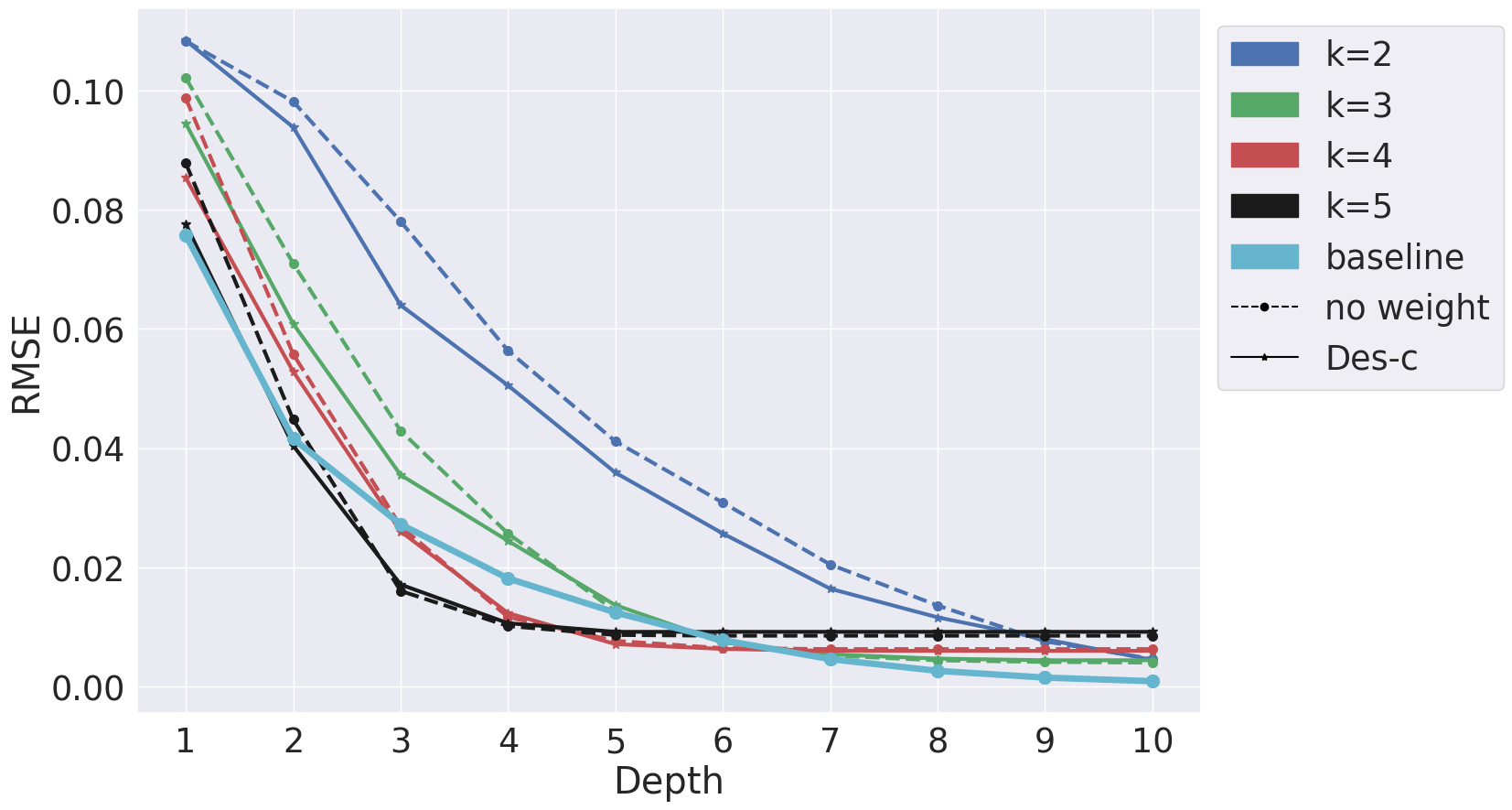} }}%
    \caption{Performance in the training of the decision trees for the regression of the Boston housing dataset: (a) shows the variance at given depth as a function of the tree depth and (b) shows the RMSE of the predictors as a function of the tree depth. The performance of Des-c (solid line) is compared to the same method but with no weight  (\emph{no weight} ) (dotted line) for different numbers of clusters ($k$), shown in colors. We compared against the baseline (cylon solid line). The values shown correspond to the mean across the folds. The standard deviation is not plotted as error bars to help visualize the trend. In the tables, we report the values with their standard deviation. 
    }%
    \label{fig:regression_training}%
\end{figure*}

\subsubsection{Tree inference: performance with test data} 
To address a regression task, once we have constructed the tree using the training data, it becomes essential to assign a predictor to each leaf node. This predictor corresponds to the prediction for all the samples allocated to that specific leaf node. The approach involves calculating the average value of the labels associated with the training samples assigned to the leaf node, where ${Y_i}, i \in n_i$, and $n_i$ represents the training samples assigned to the $i$-th leaf. Table \ref{table:main_results_regression} presents the results for two different tree depths (depths $1$ and $2$). We assess the trees built with Des-c using varying numbers of clusters and compare them to the baseline, which employs binary splits at all depths. Our report includes the results obtained with "Des-c" using the minimum number of clusters ($k$) required to achieve RMSE comparable to the baseline. We also provide the variance observed in the test results, which represents the weighted sum of variances among the test samples assigned to the leaf nodes of the tree.

Our findings indicate that, for the two depths considered, Des-c performs comparably with the baseline, as there is an overlap between RMSE and variance for both methods. However, as previously noted for classification, the proposed method demonstrates improved performance and competitiveness with the baseline as the number of clusters is increased.

\begin{table*}[t!]
  \centering
    \begin{tabular}{|p{0.9in}|p{0.9in}|p{1.1in}|p{0.6in}|p{0.8in}|p{0.8in}|} \hline 
      \cline{1-6} 
     Method & RMSE & Variance & Tree depth   & Tree size & Clusters \\ \cline{1-6} 
    
    %\multirow{3}{*}{Boston} 
     Baseline &  $0.084 \pm 0.015$ & $0.081 \pm 0.017$ & 1 & 3 & 2 \\ \cline{1-6} 
    %no weight & $10.29\pm2.56$ & $ 9.30 \pm 2.21$  & 1 &  5 &4  \\ \cline{1-3} \cline{5-6}   %k=4
    Des-c   & $0.091 \pm0.026$ & $0.084 \pm0.025$  & 1 & 5 &4 \\ \cline{1-6}  %k=4
    %no weight & $9.58\pm2.58$ & $8.63 \pm 2.5$   &  & 6 &5\\ \cline{1-3} \cline{5-6}   %k=5
    %Des-c & $8.58\pm2.27$ &  $7.53 \pm 2.48$  & & 6 & 5\\ \clineB{1-6}{2.5} %k=5
    %% below is depth 2
     Baseline & $0.054\pm0.007$  & $0.049\pm0.011 $ &2 & 7 & 2\\ \cline{1-6} 
    %no weight & $5.98\pm1.40$ & $ 4.81 \pm1.46$  &  & 21 & 4 \\ \cline{1-3} \cline{5-6} %k=4
    %Des-c & $5.32 \pm0.73$ &  $3.54\pm 0.68$  & 2 & 20 &4 \\ \cline{1-3} \cline{5-6}   %k=4
    %no weight &  $5.23 \pm1.03$ &  $3.55 \pm 0.76$  &  & 31 &5 \\ \cline{1-3} \cline{5-6}  %k=5
    Des-c & $0.057 \pm0.051$ & $0.037 \pm 0.057$ & 2 & 31 & 5\\   %k=5
    \hline 
    
    \end{tabular}
\caption{Performance in test dataset for the regression of the Boston housing dataset. RMSE and the variance are reported. We compare the  for the same number of cluster for two tree depths: $1$ and $2$. They are also compared against the binary-split technique (the baseline) for the same tree sizes and (number of nodes). It is reported the average and the standard deviation across the five folds of both RMSE and variance and the tree size correspond to the average. \label{table:main_results_regression}}
\end{table*}

\section{Conclusion} \label{sec:conclusion}

State-of-the-art decision tree methods rely on axis-parallel splits, using impurity measures for classification and variance reduction for regression. These methods scale polynomially with both the number of training examples $N$ and features $d$. Existing quantum algorithms achieve quadratic speedups in $d$ but not in $N$, which is usually more critical for practical speedups. Additionally, they don't address the need for periodic updates to maintain model performance with new data.

We introduce Des-q, a novel quantum algorithm for constructing and retraining decision trees. Once deployed, it can retrain the tree in poly-logarithmic time relative to the number of training examples $N$. Des-q is designed for binary classification and regression tasks with numerical data.

Des-q leverages the quantum version of the k-means algorithm devised by Kerenidis \etal \cite{kerenidis2019q}, known as q-means, by developing a quantum-supervised clustering algorithm that is utilized to sequentially split the nodes of the tree, from root up to reaching a desired depth. This quantum-supervised clustering algorithm incorporates feature weights in the distance calculation between the training examples and the centroids. To implement feature weighting, we have introduced a quantum algorithm for estimating the Pearson correlation in regression tasks, where the labels are continuous, and the point-biserial correlation in classification tasks, which involve binary labels. This method scales as $\mathcal{O}(1/\epsilon)$ to estimate the feature weights within $\epsilon$ precision. 

The most significant advantage of our proposed algorithm, Des-q, lies in tree retraining, in particular when dealing with small batches of new data, where $N_{\text{new}} \ll N$. After initially loading the first $N$ examples into the KP-tree data structure with polynomial time complexity in $N$, adding a new small batch of data only requires time polynomial in $N_{\text{new}}$. This efficient process enables us to build Des-q for tree retraining with poly-logarithmic time complexity in $N + N_{\text{new}} \approx N$. 

In terms of the model performance, we have investigated the performance of the classical analog of Des-q on three datasets for binary classification and one for regression. The results demonstrate that our proposed method achieves test performance (accuracy for classification and root-mean-square error for regression) that is highly competitive with state-of-the-art axis-parallel split-based methods, which utilize entropy for constructing classification trees and variance reduction for regression trees.

We also remark on the possibility of improving the existing classical algorithms for decision tree construction and retraining by leveraging quantum-inspired classical algorithms, also referred to as \emph{de-quantized} algorithms \cite{tang2019quantum}. This could be relevant in light of the recent de-quantization result of the unsupervised clustering algorithm q-means by Doriguello \etal \cite{doriguello2023you}. We consider this as a promising direction for future research.

\section*{Data and code availability} The results presented in the paper can be reproduced utilizing the data and code available in \url{https://zenodo.org/records/8428525}.

\section*{Author contribution}

R. Yalovetzky, N. Kumar, and C. Li devised the project. N. Kumar and R. Yalovetzky wrote the algorithm for Des-q. N. Kumar, C. Li and R. Yalovetzky did the error analysis for the components of Des-q.  R. Yalovetzky devised the plan of carrying out the numerics. R. Yalovetzky, C. Li, and P. Minnsen carried out the numerics and benchmarked Des-q against state-of-art baseline decision trees. M. Pistoia led the overall project. All authors
contributed to technical discussions and the writing of
the manuscript.

\section*{Acknowledgements}
The authors thank Shouvanik Chakrabarti, Ruslan Shaydulin, Yue Sun, Dylan Herman, Arthur Rattew and the other colleagues at the Global Technology Applied Research Center of JPMorgan Chase for support and helpful discussions. In addition, the authors thank Kate Stern-Jones and Bryan C Gasche from CIB, JPMorgan Chase for valuable feedback. 

\bibliographystyle{plainurl}
\bibliography{bibliography}

\onecolumn\newpage
\appendix

\section{Ingredients} \label{sec:ingre}  

This section covers the necessary ingredients required to build Des-q.

\subsection{Amplitude amplification and estimation} \label{sec:AE}

\begin{theorem}[Amplitude amplification \cite{brassard2002quantum}]
    Given the ability to implement a quantum unitary map $U$ and $U^{\dagger}$ such that $U\ket{0} = \sin \theta \ket{x,1} + \cos \theta \ket{G, 0}$, where $\ket{G}$ is the garbage state, then we can create the state $\ket{x}$ in time $\mathcal{O}(\frac{T(U)}{\sin \theta})$, where $T(U)$ is the time required to implement $U$ and $U^{\dagger}$.
\end{theorem}

\begin{theorem}[Amplitude estimation \cite{brassard2002quantum}]
        Given the ability to implement a quantum unitary map $U$ and $U^{\dagger}$ such that $U\ket{0} = \sin \theta \ket{x,1} + \cos \theta \ket{G, 0}$, where $\ket{G}$ is the garbage state, then $\sin \theta$ can be estimated up to additive error $\epsilon > 0$ in time $\mathcal{O}(\frac{T(U)}{\epsilon})$, where $T(U)$ is the time required to implement $U$ and $U^{\dagger}$.
\end{theorem}

Suppose we have a unitary $U$ that can implement the mapping
\begin{equation}
U\ket{0}_n\ket{0} = \sqrt{p} \ket{x}\ket{1} + \sqrt{1 - p} \ket{G}\ket{0},  
\label{eq:amplitude_estimation}
\end{equation}
where $\sqrt{p} = \sin \theta$. Then for any $\epsilon > 0$, the amplitude estimation algorithm can output $\Tilde{\theta}$ such that $|\theta - \Tilde{\theta}| \leq \epsilon$ and consequently, the respective probability $\Tilde{p} = \sin \Tilde{\theta}$ is such that
\begin{equation}
    |\Tilde{p} - p| \leq 2\pi \sqrt{p(1-p)}\epsilon + (\pi\epsilon)^2 < 2\pi \epsilon + (\pi\epsilon)^2 \sim \mathcal{O}({\epsilon})
    \label{eq:ae}
\end{equation}
with probability at least $8/\pi^2$ in $P = \mathcal{O}(1/\epsilon)$ iterations of $U$ and $U^{\dagger}$. Under the hood, the algorithm relies on the use of quantum phase estimation routine \cite{nielsen2010quantum} where geometrically increasing powers of the operator $\mathcal{Q} = U\mathcal{S}_0 U^\dagger \mathcal{S}_{x}$ are applied with the operators being $S_0 = \mathbb{I} - 2\ket{0}_{n+1}\bra{0}_{n+1}$ and $\mathcal{S}_x = \mathbb{I} - 2\ket{x,1}\bra{x,1}$. This is then followed by the inverse quantum Fourier transform (QFT) to obtain the estimate $\Tilde{p}$. We note that subsequent proposals of QFT-free amplitude estimation procedures have also been developed \cite{aaronson2020quantum, grinko2021iterative, giurgica2022low}.                                                                                          
\subsection{Boosting amplitude estimation probability} \label{sec:ME}

Amplitude estimation allows us to estimate the value of $\Tilde{p}$ within a precision given in Eq~\ref{eq:ae} with probability $\alpha \geq 8/\pi^2$. We leverage the Lemma 8 from   Wiebe \etal \cite{wiebe2014quantum} to boost this success probability to any arbitrary value close to 1. 

\begin{theorem}[Majority Evaluation \cite{wiebe2014quantum}]
    Let $\mathcal{U}$ be a unitary operation that maps
    \begin{equation}
        \mathcal{U} : \ket{0}^{\otimes n} \rightarrow \sqrt{a} \ket{x,1} + \sqrt{1- a} \ket{G,0} 
    \end{equation}
    for some $1/2 < a \leq 1$ in time $T$. Then there exists a deterministic quantum algorithm such that for any $\Delta > 0$, the algorithm produces a quantum state $\ket{\Psi}$ which obeys $\|\ket{\Psi} - \ket{0}^{\otimes nk}\ket{x,1}\| \leq \sqrt{2\Delta}$ for some integer $k$ in time
    \begin{equation}
          \mathcal{O}(T\ln(1/\Delta)).
    \end{equation}
\label{theo:boosting}
\end{theorem}

We note that the amplitude estimation final state (pre-measurement) has the form (refer to Eq.\ref{eq:amplitude_estimation})
\begin{equation}
    %\mathcal{U} : \ket{0}_n\ket{0} \rightarrow 
    \sqrt{\alpha} \ket{\Tilde{p}}\ket{G,1} + \sqrt{1- \alpha} \ket{G',0},
\end{equation}
where $\ket{G}$ and $\ket{G'}$ are garbage states. The aim is to amplify the state $\ket{\Tilde{p}}$. We utilize Theorem \ref{theo:boosting} by taking $\ket{x,1} = \ket{\Tilde{p}}\ket{G,1}$ and $\ket{G,0} = \ket{G',0}$. This is utilized to calculate the distances in Theorem \ref{thm:cd}, as it is also done in the q-means algorithms \cite{kerenidis2019q}.

\begin{proof}

\begin{equation}
   \mathcal{M} : \ket{y_1}\cdots \ket{y_k}\ket{0} \rightarrow \ket{y_1}\cdots \ket{y_k}\ket{\Tilde{y}},
\end{equation}

\begin{equation}
    ( \sqrt{a} \ket{y} + \sqrt{1- a} \ket{y^{\bot}})^{\otimes k} := A \ket{\Psi} + \sqrt{1 - |A|^2}\ket{\Phi},
\end{equation}

\begin{equation}
    \mathcal{M} \left(( \sqrt{a} \ket{y} + \sqrt{1- a} \ket{y^{\bot}})^{\otimes k}\ket{0}^{\otimes n} \right) = A \ket{\Psi}\ket{y} + \sqrt{1 - |A|^2}\ket{\Phi}\ket{y^{\bot}}.
\end{equation}

\begin{equation}
\begin{split}
          \mathcal{U}^{\dagger \otimes k}\left(A \ket{\Psi}\ket{y} + \sqrt{1 - |A|^2}\ket{\Phi}\ket{y^{\bot}}\right)  &=  \mathcal{U}^{\dagger \otimes k}\left(A \ket{\Psi}\ket{y} + \sqrt{1 - |A|^2}\ket{\Phi}\ket{y}\right)\\ 
      &+ \mathcal{U}^{\dagger \otimes k}\left(\sqrt{1 - |A|^2}(\ket{\Phi}\ket{y^{\bot}} - \ket{\Phi}\ket{y})\right) \\ 
     &= \ket{0^{\otimes nk}}\ket{y} +  \mathcal{U}^{\dagger \otimes k}\left(\sqrt{1 - |A|^2}(\ket{\Phi}\ket{y^{\bot}} - \ket{\Phi}\ket{y})\right).
\end{split}
\end{equation}

\begin{equation}
    \|\mathcal{U}^{\dagger \otimes k}\left(A \ket{\Psi}\ket{y} + \sqrt{1 - |A|^2}\ket{\Phi}\ket{y^{\bot}}\right) - \ket{0^{\otimes nk}}\ket{y}\| \leq \sqrt{2(1 - |A|^2)}.
\end{equation}

\begin{equation}
    \mathbb{P}(y^{\bot}) = \mathbb{P}(N_y < k/2) \leq \exp\left(-2k(|a| - \frac{1}{2})^2\right) = \Delta.
\end{equation}
\begin{equation}
    k \geq \frac{\ln(\frac{1}{\Delta})}{2(|a| - \frac{1}{2})^2}.
\end{equation}

\begin{equation}
    \|\mathcal{U}^{\dagger \otimes k}\left(A \ket{\Psi}\ket{y} + \sqrt{1 - |A|^2}\ket{\Phi}\ket{y^{\bot}}\right) - \ket{0}^{\otimes nk} \ket{y}\| \leq \sqrt{2\Delta}.
\end{equation}

\begin{equation}
    \mathcal{O} (knT) = \mathcal{O}\left(2Tn\ceil*{\frac{\ln(1/\Delta)}{2(|a| - \frac{1}{2})^2}}\right) = \mathcal{O}(T\ln(1/\Delta)),
\end{equation}

\end{proof}

\subsection{Useful quantum subroutines} \label{sec:subroutines}

Similar to classical boolean subroutines, quantum circuits offer a reversible version of classical subroutines in time linear in the number of qubits. We highlight a few subroutines required in the construction of Des-q.

\noindent \textbf{1. Equality and inequality}:
For two integers $i$ and $j$, there is a unitary to check the equality with the mapping $$\ket{i}\ket{j}\ket{0} \rightarrow \ket{i}\ket{j}\ket{[i=j]}$$. If the integers are represented with $n$ bits i.e., $i = i_1\cdots i_n$ (similarly for $j$), then one can perform the SWAP-test operation with the ancilla qubit being 0 if the two integers are equal, and 1 otherwise \cite{buhrman2001quantum}. 
One can also use the SWAP-test to perform the inequality test for two real numbers $x_1$ and $x_2$ by doing the following map $$\ket{x_1}\ket{x_2}\ket{0} \rightarrow \ket{x_1}\ket{x_2}\ket{[x_1 \leq x_2]}. $$

\noindent \textbf{2. Quantum adder:} Another operation of interest is performing the modulo additions of the quantum states, namely
\begin{equation}
    \ket{x_1}\ket{x_2}\cdots \ket{x_N}\ket{0} \rightarrow \ket{x_1}\ket{x_2}\cdots \ket{x_N}\ket{x_1 + x_2 +\cdots x_N \hspace{1mm} (\text{mod} \hspace{1mm} d)},
\end{equation}
where $d$ is the number of qubits used to represent the states $\ket{x_i}$, $i \in [N]$. As highlighted by the authors in \cite{ruiz2017quantum}, this can be done by applying the following operation on the input state
\begin{equation}
    IQFT_N \cdot CZ_{1, N}\cdots CZ_{N-2,N}\cdot CZ_{N-1,N} \cdot QFT_N,
\end{equation}
where $QFT_N, IQFT_N$ implies doing the quantum Fourier transform and inverse quantum Fourier transform respectively on the state $\ket{x_N}$. $CZ_{i,j}$ is the control-phase operation on the states $\ket{x_i}$ and $\ket{x_j}$. This operation takes time $\mathcal{O}(Nd)$ time to perform. 

\noindent \textbf{3. Quantum multiplier:} One can similarly perform the quantum operation to multiply the contents of $N$ quantum states namely
\begin{equation}
    \ket{x_1}\ket{x_2}\cdots \ket{x_N}\ket{0} \rightarrow \ket{x_1}\ket{x_2}\cdots \ket{x_N}\ket{x_1 \cdot x_2 \cdot \cdots x_N \hspace{1mm} (\text{mod} \hspace{1mm} d)},
\end{equation}
where $d$ is the number of qubits used to represent the states $\ket{x_i}$, $i \in [N]$. As highlighted by the authors in \cite{ruiz2017quantum},this operation can again be done with the application of quantum Fourier transform and controlled weighted sum in time $\mathcal{O}(Nd)$.

\noindent \textbf{4. Non-linear operations} One can apply a non-linear function $x \rightarrow \phi(x)$ in the quantum domain using the following mapping $$\ket{x}\ket{0} \rightarrow \ket{x}\ket{\phi(x)}.$$ The typical non-linear functions, including $\arcsin{x}, x^2, \sqrt{x}$, can be applied using Taylor decomposition or other techniques in time linear in the number of qubits of $\ket{x}$.

\section{Amplitude encoding with oracle QRAM} \label{sec:QRAM_AMP}

There are multiple proposals of preparing the state given in Eq~\ref{eq:amp_enc} or Eq~\ref{eq:amp_enc_matrix} which uses an efficient data loading structure. One such proposal is using direct manipulation of the state generated via quantum random access memory (QRAM) \cite{giovannetti2008quantum}. They are memory models which act as a link between the classical data and quantum states and are capable of answering queries in a quantum superposition. A general QRAM allows having access to the oracle that implements the following unitary
\begin{equation}
    \sum_{i} \alpha_{i} \ket{i}\ket{0} \rightarrow \sum_{i} \alpha_{i} \ket{i} \ket{x_i},
\end{equation}

whereas the standard classical random access memory provides for the mapping $i \rightarrow x_i$. 

This oracle-based QRAM model has been extensively studied in quantum complexity literature, for example in algorithms for Grover's search, amplitude amplification, HHL, quantum principal component analysis, and quantum recommendation systems among many others \cite{harrow2009quantum, kerenidis2016quantum, lloyd2013quantum}. Physical realizations and architectures of oracle QRAM have been proposed in the literature \cite{giovannetti2008quantum}. For a recent overview of QRAM techniques and implementations, we refer the reader to Jaques \etal \cite{jaques2023qram} and Allcock \etal \cite{allcock2023constant}.

\begin{theorem}
    Given a unitary U which takes the state $\ket{0}$ and creates in time $\mathcal{O}(\log d)$ the oracle QRAM encoding quantum state of the vector $x = (x_1,\cdots,x_d)$, then one can perform the following map
    \begin{equation}
         \frac{1}{\sqrt{d}} \sum_{j=1}^d \ket{j}\ket{0}  \rightarrow \frac{1}{\sqrt{d}} \sum_{j=1}^d \ket{j}\ket{x_j} \rightarrow \frac{1}{\|x\|} \sum_{j=1}^d x_j \ket{j}\ket{0}
    \end{equation}
    in time $\mathcal{O}(\sqrt{d}\log(d) \frac{\eta^2}{\|x\|})$, where $\eta \geq max(x_j)$.
\end{theorem}

\begin{proof}
    Let us apply the conditional rotation on the QRAM state
    \begin{equation}
        \frac{1}{\sqrt{d}} \sum_{j=1}^d \ket{j}\ket{x_j}\ket{0} \rightarrow \frac{1}{\sqrt{d}} \sum_{j=1}^d \ket{j}\ket{x_j} \left(\frac{x_j}{\eta}\ket{0} + \sqrt{1 - \frac{x_j^2}{\eta^2}}\ket{1}\right).
    \end{equation}
    So if the vector $x$ is normalized then $\eta = 1$. From this, one can simply measure the ancilla bit in the state $\ket{0}$ to end up with te desired state with probability $P(0)$. So this probability is
    \begin{equation}
        P(0) = \frac{1}{d} \sum_{j=1}^d \frac{x_j^2}{\eta^2} = \frac{\|x\|^2}{d\eta^2}.
    \end{equation}
    Note that this gives us that one needs $\mathcal{O}(1/P(0))$ samples to have a constant probability of creating the amplitude encoded state. This complexity can be quadratically improved if instead of directly measuring the last qubit, we perform amplitude amplification on it to obtain the amplitude-encoded state. This reduces the complexity of preparing the state from the QRAM state to $\mathcal{O}(1 / \sqrt{P(0)})$. Unless in the special cases, the complexity of this method then scales as $\mathcal{O}(\frac{1}{\sqrt{d}})$.
    
    Thus the final state obtained by either directly measuring the last qubit state or by performing amplitude amplification is
    \begin{equation}
        \frac{1}{\sqrt{d}} \sum_{j=1}^d cx_j\ket{j}\ket{x_j}\ket{0},
    \end{equation}
    where $c$ is the proportionality factor. Given that the state must respect the normalization rule, we have that $c = d/\|x\|^2$. This gives us the state
    \begin{equation}
        \frac{1}{\|x\|} \sum_{j=1}^d x_j \ket{j} \ket{x_j}.
    \end{equation}
    Now in order to get the final amplitude encoding state, one can apply the inverse QRAM unitary state to obtain the desired amplitude encoding state. 
\end{proof}

Note that a direct implication of this theorem is the following lemmas
\begin{lemma}[Loading the matrix]
    Given a unitary U which takes the state $\ket{0}$ and creates in time $\mathcal{O}(\log (Nd))$ the oracle QRAM encoding quantum state of the data matrix $X \in \mathbb{R}^{N \times d}$, then one can perform the following map
    \begin{equation}
         \frac{1}{\sqrt{Nd}} \sum_{i=1}^N\sum_{j=1}^d \ket{i}\ket{j}\ket{0}  \rightarrow \frac{1}{\sqrt{Nd}} \sum_{i,j} \ket{i}\ket{j}\ket{x_{ij}} \rightarrow \frac{1}{\|X\|_F} \sum_{i=1}^N \|x_i\| \ket{x_i} \ket{i}
    \end{equation}
    in time $\mathcal{O}(\sqrt{Nd}\log(Nd) \frac{\eta^2}{\|X\|_F})$, where $\eta \geq max(x_{ij})$. Where $\ket{x_i} = \frac{1}{\|x_i\|} \sum_{j=1}^d x_{ij} \ket{j}$ and $\|X\|_F = \sqrt{\sum_{i=1}^N \|x_i\|^2}$
\end{lemma}

\begin{lemma}[Loading the features]
    Given a unitary U which takes the state $\ket{0}$ and creates in time $\mathcal{O}(\log (N))$ the oracle QRAM encoding quantum state of the feature vector $x^{(j)} := (x_{1j}, \cdots, x_{Nj}), j \in [d]$ for the data matrix $X \in \mathbb{R}^{N \times d}$, then one can perform the following map
    \begin{equation}
         \frac{1}{\sqrt{N}} \sum_{i=1}^N \ket{i}\ket{0}  \rightarrow \frac{1}{\sqrt{N}} \sum_{i=1}^N \ket{i}\ket{x_{ij}} \rightarrow \frac{1}{\|x^{(i)}\|} \sum_{i=1}^N x_{ij} \ket{i}
    \end{equation}
    in time $\mathcal{O}(\sqrt{N}\log(N) \frac{\eta^2}{\|x^{(i)}\|})$, where $\eta \geq max(x_{ij})$.
\end{lemma}

An immediate implication of producing amplitude-encoded states via Oracle QRAM routine is that unless in special cases, the time complexity is proportional to square root of the size of the input. Thus, in the aim of having provable exponential speedups in general cases, this method cannot be leveraged. Next, we showcase an alternative memory model with which one can produce the amplitude-encoded states with time complexity logarithm in the size of the input data.

\subsection{Estimating Pearson correlation by directly using QRAM-like state}\label{appendix:QBC_for_correlation}

The goal here is to calculate the correlation coefficient defined in Eq~\ref{Eq:pearson2} using the following QRAM-like states
\begin{equation}
   \frac{1}{\sqrt{N}} \sum_i^N \ket{i}_n \ket{x_i^{(j)}}, \quad \frac{1}{\sqrt{N}}\sum_i^N \ket{i}_n \ket{y_i},
\end{equation}
where the data $x_i^{(j)}$ and $y_i$ are encoded in the states and $|i\rangle_n$ is an $n\equiv\lceil\log_2(N)\rceil$-qubit (called index qubit hereafter) state $|i_1i_2\cdots i_n\rangle$, representing the index of the queried component. This state can be prepared with QRAM, or with a well-designed oracle circuit. 

For now, we consider the case where $x_i^{(j)}$ and $y_i$ are binary and can be extended to general floating numbers. To estimate the Pearson correlation coefficient, we are interested in estimating $\sum x_i^{(j)} \cdot y_i$, $\sum x_i^{(j)}$ and $\sum y_i$. Note that $I(x^{(j)}, \mathds{1}) = x^{(j)} \cdot \mathds{1} = \sum_i x_i^{(j)}$ and  $I(Y, \mathds{1}) = Y \cdot \mathds{1} = \sum_i y_i$, where $\mathds{1}$ is the identity vector. Therefore, the problem is reduced to estimating these inner products. There are several methods for this binary correlation problem. One example is to use the recently-developed quantum bipartite correlator (QBC) algorithm~\cite{PhysRevLett.130.150602} that is based on quantum counting~\cite{Brassard1998}. 

While the detailed information can be found in~\cite{PhysRevLett.130.150602}, here we briefly illustrate the algorithm flow. The idea is to convert the correlation to be estimated into phase information. 

As an example, we consider the estimation of $\sum_i^N x_i^{(j)} \cdot y_i$ for evaluation of Pearson correlation coefficient $w_j$. Firstly, one encodes the information with two qubits $q_1$ and $q_2$, that is
\begin{equation}
     \frac{1}{\sqrt{N}} \sum_i^N \ket{i}_n \ket{0}_{q_1} \ket{0}_{q_2} \rightarrow \frac{1}{\sqrt{N}} \sum_i^N \ket{i}_n \ket{x_i^{(j)}}_{q_1} \ket{y_i}_{q_2}.
\end{equation}
Then, a CZ gate between qubits $q_1$ and $q_2$ will yield state
\begin{equation}
   \frac{1}{\sqrt{N}} \sum_i^N (-1)^{x_i^{(j)} y_i} \ket{i}_n \ket{x_i^{(j)}}_{q_1} \ket{y_i}_{q_2}.
\end{equation}
Finally, the qubits $q_1$ and $q_2$ are disentangled with index qubits: 
\begin{equation}
    \frac{1}{\sqrt{N}} \sum_i^N (-1)^{x_i^{(j)} y_i} \ket{i}_n \ket{x_i^{(j)}}_{q_1} \ket{y_i}_{q_2} \rightarrow \frac{1}{\sqrt{N}} \sum_i^N (-1)^{x_i^{(j)} y_i} \ket{i}_n \ket{0}_{q_1} \ket{0}_{q_2}.
\end{equation}
The above operations, followed by operations $\hat{H}^{\otimes n}(2|0\rangle_n \langle 0|_n-\hat{I})\hat{H}^{\otimes n}$ on index qubits, serve as the Grover operator for quantum counting algorithm. 
By applying the quantum counting algorithm with the help of register qubits,  the correlation $\frac{1}{N}\sum_i^N x_i^{(j)} y_i$ would be extracted.

We remark that the computational complexity for this process would be $O(\frac{logN}{\epsilon})$ with $\epsilon$ being the estimation error bound. %That is, one can have an exponential speedup compared with classical algorithms.  

\section{Proof of Lemma \ref{Thm:KP-tree}} \label{App:proof_data_load}
\begin{lemma} 
    Let $X \in \mathbb{R}^{N \times d}$ be a given dataset. Then there exists a classical data structure to store the rows of $X$ with the memory and time requirement to create the data structure being $T_{kp} = \mathcal{O}(Nd \log^2(Nd))$ such that, there is a quantum algorithm with access to the data structure which can perform the following unitaries (and also in superposition) in time $T = \mathcal{O}(\text{poly}\log (Nd))$
    \begin{align}
        \ket{i}\ket{0} &\rightarrow \ket{i} \frac{1}{\|x_i\|}\sum_{j=1}^d x_{ij}\ket{j} \\
        \ket{0} &\rightarrow \frac{1}{\|X\|_F} \sum_{i=1}^N \|x_i\| \ket{i}
    \end{align}
\end{lemma}

\begin{proof}
Consider $X \in \mathbb{R}^{N \times d}$ where each row is a $d$ dimensional vector $x_i \in \mathbb{R}^d$, such that $x_i = (x_{i1}\cdots x_{id})$. Then in order to store it in the data structure, we construct $N$ binary-tree data structures $B_i, i \in [N]$ with $d$ leaves. The tree is initially empty and it is updated in an online sequential manner. When a new entry $(i, j, x_{ij})$ arrives, the leaf node $j$ in the tree $B_i$ is created if it is not present already. The leaf $j$ then stores the value of $x_{ij}^2$ as well as the sign of $x_{ij}$. Since there are $d$ leaves in the tree $B_i$, the depth of the tree is $\ceil{\log d} $. An internal node  $l$ of the tree at depth $t$ stores the sum of the values of all the leaves in the subtree rooted at $l$ i.e., the sum of the square amplitudes of the leaves in the subtree. It follows that the root node (depth = 0) of the tree $B_i$ contains the sum of the amplitudes $\sum_{j=1}^d x_{ij}^2$. Let us denote the value of the internal node $l$ at depth $t$ of the $i$-th tree as $B_{i,(t,l)}$, where $l \in \{0,1\}^t$. Then this value can be written as
\begin{equation}
    B_{i,(t,l)} = \sum_{j_1\cdots j_t = l; j_{t+1}, \cdots, j_{\ceil{\log d}} \in \{0,1\}} x_{ij}^2
\end{equation}

The above equation says that the first $t$ bits of $j$, written in the binary notation $j \equiv j_1\cdots j_{\ceil{\log d}}$, are fixed to $l$ indicating the depth $t$. The rest of the bits ($j_{t+1}, \cdots, j_{\ceil{\log d}}$) are values in $\{0, 1\}$.

A new entry into the tree updates all the nodes of the path from the leaf to the root node, thus updating $\ceil{\log d}$ nodes. A caveat with this tree-like structure is that the tree levels are stored as ordered lists and thus the nodes are retrieved in $\mathcal{O}(\log (Nd))$ time in order to be updated. Thus, if one has to update the tree with the entry $(i, j, x_{ij})$, then the total time required for this is $\mathcal{O}(\log^2 (Nd))$ since $\ceil{\log d}$ nodes are updated from leaf to the node. Thus the total time to update all the elements of the matrix in the empty tree is $T_{kp} = \mathcal{O}(Nd \log^2 (Nd))$. Similarly, the total memory requirement of this tree is $\mathcal{O}(Nd \log^2 (Nd))$.

Now we will see how to create the amplitude encoding state given that we have quantum superposition access to this classical data structure. In order to create the amplitude encoding state corresponding to the row $i \in [N]$, we start with the initial state $\ket{0}^{\otimes \ceil{\log d}}$, we query the tree $B_i$ and then we perform  $\ceil{\log d}$ conditional rotations as explained below.
\begin{align*}
    \text{Initial state} & = \underbrace{\ket{0}\cdots \ket{0}}_{\ceil{\log d}} \\
    \text{Conditional rotation (qubit 1)} &= \frac{1}{\sqrt{B_{i,(0,0)}}}(\sqrt{B_{i,(1,0)}} \ket{0} + \sqrt{B_{i,(1,1)}}\ket{1}) \underbrace{\ket{0\cdots 0}}_{\ceil{\log d} -1} \\
    \text{Conditional rotation (qubit 2)} &= \frac{1}{\sqrt{B_{i,(0,0)}}}\Bigg(\sqrt{B_{i,(1,0)}} \ket{0} \frac{1}{\sqrt{B_{i,(1,0)}}}(\sqrt{B_{i,(2,00)}} \ket{0} + \sqrt{B_{i,(2,01)}}\ket{1}) \\
    &+ \sqrt{B_{i,(1,1)}} \ket{0} \frac{1}{\sqrt{B_{i,(1,1)}}}(\sqrt{B_{i,(2,10)}} \ket{0} +  \sqrt{B_{i,(2,11)}}\ket{1}) \Bigg)\underbrace{\ket{0\cdots 0}}_{\ceil{\log d} -2} \\
    \vdots \\
        \text{Conditional rotation (qubit \hspace{1mm}} \ceil{\log d}) &= \frac{1}{\sqrt{B_{i,(0,0)}}}\sum_{l=0}^{2^{\ceil{\log d}}-1} \text{sgn}(x_{il}) \sqrt{B_{i,(\ceil{\log d}, l)}} \ket{l} \\
        &= \frac{1}{\|x_i\|^2} \sum_{j=1}^{d} x_{ij} \ket{j},
\end{align*}

where $B_{i,(0,0)} = \sum_{i=1}^d x_{ij}^2$ is the value stored in the data structure at the root node. The rotation on the $(t+1)$-th qubit is conditioned on the first $t$ qubits. Also after performing the conditional rotation on the $\ceil{\log d}$-th qubit, one also appends the signs of the values stored in the leaf node to prepare the amplitude encoding state. With this method, conditioned on the classical data structure being prepared, the runtime complexity of preparing the amplitude encoded state (for any $i \in [N]$ or also in superposition) is $\mathcal{O}(\emph{poly}\log (Nd))$.

We also note that we can also create a quantum superposition of the norms of the $N$ examples i.e., implement a unitary $\ket{i}\ket{0} \rightarrow \ket{i}\ket{\|X\|}$. This can be done by noting that the roots of all the trees $B_i$ store the values $\|x_i\|^2$. Hence we can construct another binary tree with $N$ leaves such that the leaves store the amplitudes $\|x_i\|^2, i \in [N]$. Then the root will store the values of $\|X\|_F$. Upon applying the conditional rotations, we can build the final state,
\begin{equation}
   \ket{\|X\|} =  \frac{1}{\|X\|_F}\sum_{i=1}^N \|x_i\| \ket{i}
\end{equation}
\end{proof}

\section{Proof of Theorem \ref{Thm:inverse_sqared_error}} \label{app:inverse_sqared_error}

\begin{theorem} \label{thmapp:inverse_sqared_error}
    Given access to the amplitude-encoded states for feature vectors $|x^{(j)}\rangle, j \in [d]$ and the label vector $\ket{Y}$ along with their norms $\|x^{(j)}\|$, $\|Y\|$ which are prepared in time $T= \mathcal{O}(\text{poly}\log (Nd))$, there exists a quantum algorithm to estimate the Pearson correlation coefficients $\Bar{w_j}, \forall j \in [d]$ in time $\mathcal{O}(\frac{Td\eta}{\epsilon^2})$, where $|\Bar{w_{j}} - w_j| \leq \epsilon$, and \\
    
   $ \eta = \frac{7 \cdot \text{max}\left(\|x^{(j)}\| \|Y\|, \|x^{(j)}\|^2, \|Y\|^2\right)}{N \cdot \text{min}\left(\sigma_{x^{(j)}}\sigma_Y, \sigma_{x^{(j)}}^2, \sigma_Y^2\right)}$, \\
   
   where $\sigma_{x^{(j)}}$ and $\sigma_{Y}$ denote the standard deviation for $x^{(j)}$ and $Y$.
\end{theorem}

\begin{proof}

Using the Lemmas~\ref{lemma:supcol}, \ref{lemma:suplabel} respectively, we can query in time $T = \mathcal{O}(\emph{poly}\log (Nd))$ the following states\footnote{Since the label state $\ket{Y}$ can be queried in time $\mathcal{O}(\emph{poly} \log(N))$, for simplicity, we consider the time to query both the states $\ket{Y}$ and $\ket{x^{(j)}}$ to be $T = \mathcal{O}(\emph{poly} \log(Nd))$}
\begin{equation}
    \ket{j}\ket{0} \rightarrow \ket{j}|x^{(j)}\rangle, \hspace{2mm} \ket{0} \rightarrow \ket{Y}
\end{equation}
Further, the nature of KP-tree memory model also allows us access to the norms $\|x^{(j)}\|, j \in [d]$, and $\|Y\|$.

We note that the Pearson correlation coefficient $w_j$ can be computed by calculating the inner product between the vectors $x^{j}$ and $Y$ and also by obtaining the mean of the elements of $x^{(j)}$ and $Y$. Here we showcase that these operations can be done efficiently if the vectors are encoded in amplitude-encoded states. We start by showcasing the method for computing the inner produce of the two vectors. This can then be easily extended to computing the mean of the vectors.   

We start with the initial state
\begin{equation}
    \ket{\phi_j} = \ket{j}\frac{1}{\sqrt{2}}(\ket{0} + \ket{1})\ket{0}.
\end{equation}

Then the KP-tree is queried controlled on the second register which results in the mappings $\ket{j}\ket{0}\ket{0} \rightarrow \ket{j}\ket{0}|x^{(j)}\rangle$ and $\ket{j}\ket{1}\ket{0} \rightarrow \ket{j}\ket{1}\ket{Y}$. Thus the state after this controlled rotation operation is given by
\begin{equation}
    \frac{1}{\sqrt{2}}\ket{j}\left(\ket{0}|x^{(j)}\rangle + \ket{1}\ket{Y}\right).
\end{equation}

Applying Hadamard operation on the second register results in the state 
\begin{equation}
    \frac{1}{2}\ket{j}\left(\ket{0}(|x^{(j)}\rangle + \ket{Y}) + \ket{1}(|x^{(j)}\rangle - \ket{Y})\right).
    \label{eq:evolve_state}
\end{equation}

Now it can be seen that if we measure the second register, the probability of obtaining outcome 1 is
\begin{equation}
    p(1) = \frac{1}{2}\left(1 - I(x^{(j)}, Y)\right),
    \label{eq:proba}
\end{equation}
where $I(x^{(j)}, Y) = |\langle x^{(j)}|Y\rangle| =  1/(\|x^{(j)}\| \|Y\|)\sum_{i=1}^N x_{ij}y_i$. Here we use the fact that $|x^{(j)}\rangle$ and $\ket{Y}$ have only real amplitudes due to the entries of $X$ and $Y$ being real values. 

Let us denote the quantity we want to estimate to be $s_j = \sum_{i=1}^N x_{ij}y_i$. Using standard Chernoff bounds \cite{mitzenmacher2017probability}, one can estimate the quantity $\overline{p(1)}$, such that $|\overline{p(1)} - p(1)| \leq \epsilon$ with $\mathcal{O}(1/\epsilon^2)$ copies of the quantum states. This also provides an $\epsilon$-close estimate on the inner product, $|\overline{I(x^{(j)}, Y)} - I(x^{(j)}, Y)| \leq \epsilon$. Further, we can also quantify the error in estimating the quantity $s_j$ with $\mathcal{O}(1/\epsilon^2)$ copies
\begin{equation}
    |\overline{s_j} - s_j| \leq \|x^{(j)}\| \|Y\| \epsilon.
    \label{Eq:cov}
\end{equation}

The next step is to estimate the mean values of the vectors $x^{(j)}$ and $Y$ which are denoted by $\mu_j$ and $\mu_y$ respectively. Instead of estimating the mean, we estimate the quantities $$a_j = \sqrt{N}\mu_j = \frac{1}{\sqrt{N}} \sum_{i=1}^N x_{ij}, \hspace{2mm} b = \sqrt{N}\mu_y = \frac{1}{\sqrt{N}} \sum_{i=1}^N y_{i}$$ The quantity $a_j$ can be estimated by computing the inner product between the states $|x^{(j)}\rangle$ and $\ket{u} = \frac{1}{\sqrt{N}} \sum_{i=1}^N \ket{i}$
\begin{equation}
    I(x^{(j)}, u) = \langle x^{(j)}|u\rangle = \frac{1}{\|x^{(j)}\|}\frac{1}{\sqrt{N}}\sum_{i=1}^N x_{ij} = \frac{1}{\|x^{(j)}\|}a_j.
\end{equation}

Similarly, computing the inner product between $\ket{Y}$ and $\ket{u}$ results in estimating the quantity $I(Y, u) = (1/\|Y\|)b$. With $\mathcal{O}(1/\epsilon^2)$ copies of the quantum states, the error in estimating the quantities $I(x^{(j)}, u)$ and $I(Y, u)$ is $\epsilon$. Thus, the error in estimating $a_j$ and  $b$ is
\begin{equation}
    |\overline{a_j} - a_j| \leq \|x^{(j)}\|  \epsilon, \hspace{2mm} \|\overline{b} - b\| \leq \|Y\| \epsilon
    \label{Eq:mus}
\end{equation}

Once we have the estimate of $s_j, a_j$, and $b$, we can estimate the Pearson correlation coefficient, since it can be expressed as
\begin{equation}
    \begin{split}
      w_j &= \frac{\sum_{i=1}^N(x_{ij} - \mu_j)(y_i - \mu_y)}{\sqrt{\sum_{i=1}^N(x_{ij} - \mu_j)^2}\sqrt{\sum_{i=1}^N(y_{i} - \mu_y)^2}} \\
&= \frac{s_j - a_jb\left(2 - \frac{1}{N}\right)}{\sqrt{\|x^{(j)}\|^2 - a_j^2}\sqrt{\|Y\|^2 -b^2}} \equiv \frac{\text{Num}}{\text{Den}}
\label{Eq:pearson2}  
    \end{split}
\end{equation}

Using the error estimates in equations Eq~\ref{Eq:cov} and Eq~\ref{Eq:mus}, we can quantify the error in the estimation of the Pearson correlation coefficient. Prior to that, let us estimate the error in the numerator
\begin{equation}
    \begin{split}
    \delta \text{Num} &= \Bigg|\frac{\partial \text{Num}}{\partial s_j}\Bigg||\overline{s_j} - s_j| + \Bigg|\frac{\partial \text{Num}}{\partial a_j}\Bigg||\overline{a_j} - a_j| + \Bigg|\frac{\partial \text{Num}}{\partial b}\Bigg||\overline{b} - b| \\
    &\leq (\|x^{(j)}\| \|Y\| + \left(2 - \frac{1}{N}\right)b\|x^{(j)}\| + \left(2 - \frac{1}{N}\right)a_j\|Y\|)\epsilon\\
    &\leq 5 \|x^{(j)}\| \|Y\| \epsilon,
    \end{split}
\end{equation}
where to go from the second to final step, we use the Cauchy-Swartz inequality to bound $b \leq \|Y\|$ and $a_j \leq \|x^{(j)}\|$.
Similarly, the error in the denominator is given by
\begin{equation}\label{eq: error_in_wj}
    \begin{split}
        \delta \text{Den} & = \frac{\sqrt{N}\mu_j \sqrt{\|Y\|^2_2 -b^2}}{ \sqrt{\|x^{(j)}\|^2 - a_j^2}} |\overline{a_j} - a_j| + \frac{ \sqrt{N}\mu_y\sqrt{\|x^{(j)}\|^2_2 - a_j^2}}{\sqrt{\|Y\|^2 -b^2}} |\overline{b} - b|  \\
        & \leq \left( \frac{\sqrt{N}\mu_j \sqrt{\|Y\|^2 -b^2}}{ \sqrt{\|x^{(j)}\|^2 - a_j^2}} \|x^{(j)}\| + \frac{\sqrt{N}\mu_y \sqrt{\|x^{(j)}\|^2 - a_j^2}}{\sqrt{\|Y\|^2 -b^2}} \|Y\| \right)\epsilon.\\
    \end{split}
\end{equation}

Using the uncertainty in both the numerator and the denominator, we can upper-bound the error in the correlation coefficient estimation
\begin{equation}
    \begin{split}
          \delta w_j &= |\overline{w_j} - w_j| = \delta \text{Num} / \text{Den} + \delta \text{Den} \cdot \text{Num}/\text{Den}^2\\
    %&\leq \left(\frac{5 \|x^{(j)}\| \|Y\|}{\text{Den}}\right)\epsilon \\  
    &\leq \left(5 \|x^{(j)}\| \|Y\| / \text{Den} + \left( \frac{\sqrt{N}\mu_j\sqrt{\|Y\|^2 -\mu_y^2N}}{ \sqrt{\|x^{(j)}\|^2 - \mu_j^2N}} \|x^{(j)}\| + \frac{ \sqrt{N}\mu_y \sqrt{\|x^{(j)}\|^2 - \mu_j^2N}}{\sqrt{\|Y\|^2 -\mu_y^2N}} \|Y\| \right) \cdot w_j /\text{Den} \right)\epsilon \\
    & = \left( \frac{5 \|x^{(j)}\| \|Y\|}{N \sigma_{x^{(j)}}\sigma_Y}  + \left( \frac{\sqrt{N}\mu_j\sigma_Y \|x^{(j)}\|}{ \sigma_{x^{(j)}}}  + \frac{ \sqrt{N}\mu_y \sigma_{x^{(j)}} \|Y\| }{\sigma_Y}  \right) \cdot \frac{w_j}{N\sigma_{x^{(j)}}\sigma_Y} \right)\epsilon \\
    & = \left( \frac{5 \|x^{(j)}\| \|Y\|}{N \sigma_{x^{(j)}}\sigma_Y}  +  \frac{w_j \mu_j \|x^{(j)}\|}{\sqrt{N} \sigma_{x^{(j)}}^2}  + \frac{ w_j \mu_y  \|Y\| }{\sqrt{N} \sigma_Y^2}   \right)\epsilon \\
    &=  \left( \frac{5 \|x^{(j)}\| \|Y\|}{N \sigma_{x^{(j)}}\sigma_Y}  +  \frac{w_j a_j \|x^{(j)}\|}{N \sigma_{x^{(j)}}^2}  + \frac{ w_j b  \|Y\| }{N \sigma_Y^2}   \right)\epsilon \\
    &\leq \frac{7 \cdot \text{max}\left(\|x^{(j)}\| \|Y\|, \|x^{(j)}\|^2, \|Y\|^2\right)}{N \cdot \text{min}\left(\sigma_{x^{(j)}}\sigma_Y, \sigma_{x^{(j)}}^2, \sigma_Y^2\right)} \epsilon
    \end{split}
\end{equation}
where $\sigma_{x^{(j)}}$ and $\sigma_{Y}$ denote the standard deviation for $x^{(j)}$ and $Y$, respectively and in the last line use the we use the Cauchy-Swartz inequality to bound $b \leq \|Y\|$ and $a_j \leq \|x^{(j)}\|$, followed by the fact that $w_j \leq 1$.

Now, in order to bound $\delta w_j \leq \epsilon$, for any desired $\epsilon$, the total number of copies the quantum state required is $\mathcal{O}(\eta^2/\epsilon^2)$, where $\eta$ is
\begin{equation}
    \eta = \frac{7 \cdot \text{max}\left(\|x^{(j)}\| \|Y\|, \|x^{(j)}\|^2, \|Y\|^2\right)}{N \cdot \text{min}\left(\sigma_{x^{(j)}}\sigma_Y, \sigma_{x^{(j)}}^2, \sigma_Y^2\right)}.
\label{eq:eta}
\end{equation}

Given that the preparation of each quantum state takes time $T =\mathcal{O}(\emph{poly}\log (Nd))$, the total time complexity of estimating the correlation coefficient individually for every feature vector $j \in [d]$ is $\mathcal{O}\left(\frac{Td\eta}{\epsilon^2}\right)$.
% \begin{equation}
%     T_{cc} = \mathcal{O}\left(\frac{Td\eta}{\epsilon^2}\right)
% \end{equation}
This completes the proof. 
\end{proof}

\section{Proof of Theorem \ref{thm:pe}} \label{app:pe}

\begin{theorem} \label{thmapp:pe}
    Given access to the amplitude-encoded states for feature vectors $|x^{(j)}\rangle, j \in [d]$ and the label vector $\ket{Y}$ along with their norms $\|x^{(j)}\|$, $\|Y\|$ which are prepared in time $T = \mathcal{O}(\text{poly}\log (Nd))$, there exists a quantum algorithm to estimate each Pearson correlation coefficient $\Bar{w_j}, j \in [d]$ in time $T_{w} = \mathcal{O}(\frac{Td\eta}{\epsilon})$, where $\|\Bar{w_{j}} - w_j| \leq \epsilon$, and \\
    
    $ \eta = \frac{7 \cdot \text{max}\left(\|x^{(j)}\| \|Y\|, \|x^{(j)}\|^2, \|Y\|^2\right)}{N \cdot \text{min}\left(\sigma_{x^{(j)}}\sigma_Y, \sigma_{x^{(j)}}^2, \sigma_Y^2\right)}$, \\
    
    where $\sigma_{x^{(j)}}$ and $\sigma_{Y}$ denote the standard deviation for $x^{(j)}$ and $Y$. 
\end{theorem}

\begin{proof}
In Theorem \ref{thmapp:inverse_sqared_error}, one could estimate the inner product by measuring the third register in Eq~\ref{eq:evolve_state} which gives the outcome 1 with probability $p(1)$ as denoted in Eq~\ref{eq:proba}. Indeed as we show below, instead of measuring the third register, if one applies amplitude estimation on the state in Eq~\ref{eq:evolve_state} (as highlighted in Appendix~\ref{sec:AE}), we can reduce the time complexity of estimating the Pearson correlation coefficient quadratically in the $\epsilon$, where $|\overline{w_j} - w_j| \leq \epsilon, j \in [d]$. 

    The idea is that the state $\ket{1}(|x^{(j)}\rangle - \ket{Y})$ can be rewritten as $|z_{jY}, 1\rangle$ (by swapping the registers), and hence Eq~\ref{eq:evolve_state} has the following mapping
    \begin{equation}
        \ket{j}\ket{0}\ket{0} \rightarrow \ket{j}\left(\sqrt{p(1)}|z_{jY}, 1\rangle + \sqrt{1 - p(1)}\ket{G,0}\right),
        \label{eq:method2AE}
    \end{equation}
    where $G$ is some garbage state. 

    Now it is clear that this is the form of the input for amplitude estimation \cite{brassard2002quantum} algorithm where the task is to estimate the unknown coefficient $\sqrt{p(1)}$. Applying amplitude estimation results in the output state (prior to measurement)
    \begin{equation}
        U: \ket{j}\ket{0} \rightarrow \ket{j}\left(\sqrt{\alpha}|\overline{p(1)}, G', 1\rangle + \sqrt{1 - \alpha}|G^{''},0\rangle\right),
    \end{equation}
where $G', G^{''}$ are garbage registers. The above procedure requires $\mathcal{O}(1/\epsilon)$ iterations of the unitary $U$ (and its transpose) to produce the state such that $|\overline{p(1)} - p(1)| \leq \epsilon$. Measuring the above state results in the estimation of $\overline{p(1)}$ with a constant probability $\alpha \geq 8/\pi^2$. From this, it becomes clear that the quantity $s_j = \sum_{i=1}^N x_{ij}y_i$ can be estimated to an accuracy Eq~\ref{Eq:cov} in $\mathcal{O}(1/\epsilon)$ iterations of $U$.

Similar procedure applies to estimating the quantities $\overline{a_j}$ and $\overline{b}$ to accuracy as in Eq~\ref{Eq:mus} in $\mathcal{O}(1/\epsilon)$ iterations of $U$.

Thus following the error analysis procedure in Method-1, the Pearson correlation coefficient can be estimated to $\epsilon$ precision with time complexity
\begin{equation}
    T_{w} = \mathcal{O}\left(\frac{Td\eta}{\epsilon}\right),
    \label{eq:PEtime}
\end{equation}
where $\eta$ is defined in Eq~\ref{eq:eta}.
%$\eta =  5 \|x^{(j)}\| \|Y\|/\left({\sqrt{\|x^{(j)}\|^2_2 - \mu_j^2N}\sqrt{\|Y\|^2_2 -\mu_y^2N}}\right)$.
\end{proof}

\section{Proof of Theorem~\ref{thm:cd}} \label{app:cd}

\begin{theorem} 
        Given quantum access to the dataset $X$ in time $T = \mathcal{O}(\text{poly} \log(Nd))$ and quantum access to the weighted centroids $c_{1,w},\cdots,c_{k,w}$ in time $\mathcal{O}(\text{poly} \log(kd))$, then for any $\Delta > 0$ and $\epsilon_1 > 0$, there exists a quantum algorithm such that
        \begin{equation}
             \frac{1}{\sqrt{N}} \sum_{i=1}^N \ket{i} \otimes_{j \in [k]}(\ket{j}\ket{0}) \rightarrow \frac{1}{\sqrt{N}} \sum_{i=1}^N \ket{i} \otimes_{j \in [k]}(\ket{j}\ket{\overline{I_w(x_i, c_j)}}),
        \end{equation}
        where $I_w(.)$ is the weighted inner product between the two vectors and $|\overline{I_w(x_i, c_j)} - I_w(x_i, c_j)| \leq \epsilon_1$ with probability at least $1 - 2\Delta$, in time $T_{cd} = \mathcal{O}(T k \frac{\eta_1}{\epsilon_1} \log(1/\Delta))$, where $\eta_1 = \text{max}_{i}(\|x_i\|^2)$.
\end{theorem}

\begin{proof}
    We prove the theorem for the case of estimating the weighted inner product between the vector $x_i$ and $c_j$ within $\epsilon_1$ precision with a probability at least $1 - 2 \Delta$ i.e.,
    \begin{equation}
        \ket{i} \ket{j}\ket{0} \rightarrow \ket{i}\ket{j}\ket{\overline{I_w(x_i, c_j)}}.
    \end{equation}
    The general result in Eq~\ref{eq:ip-sup} then follows in a straightforward manner. 
    
    Similarly to the Pearson correlation estimation in Method-1 of Sec~\ref{sec:method1}, we start with the initial state
    \begin{equation}
        \ket{\phi_j} = \ket{i}\ket{j}\frac{1}{\sqrt{2}}(\ket{0} + \ket{1})\ket{0}.
    \end{equation}

    Then the KP-tree is queried controlled on the third register which results in the mappings $\ket{i}\ket{j}\ket{0}\ket{0} \rightarrow \ket{i}\ket{j}\ket{0}\ket{x_i}$ and $\ket{i}\ket{j}\ket{1}\ket{0} \rightarrow \ket{i}\ket{j}\ket{1}|c_{j,w}\rangle$. Thus the state after this controlled rotation operation is given by
    \begin{equation}
        \ket{i}\ket{j}\frac{1}{\sqrt{2}}\left(\ket{0}\ket{x_i} + \ket{1}|c_{j,w}\rangle\right).
    \end{equation}

    Applying Hadamard operation on the third register results in the state
    \begin{equation}
       \ket{i}\ket{j}\left( \frac{1}{2}\ket{0}(\ket{x_i} + |c_{j,w}\rangle) +  \frac{1}{2}\ket{1}(\ket{x_i} - |c_{j,w}\rangle)\right).
        \label{eq:centroid_sup}
    \end{equation}
    
    From the above equation, it can be seen that the probability of obtaining the third register's measurement outcome to be 1 is
    \begin{equation}
        p(1) = \frac{1}{2}(1 - \bra{x_i}c_{j,w}\rangle) = \frac{1}{2}(1 - I(\ket{x_i}, |c_{j,w}\rangle)),
    \end{equation}
    where $I(\ket{x_i}, |c_{j,w}\rangle)$ is the normalized inner product between the two states $\ket{x_i}$ and $|c_{j,w}\rangle$. Note that $I_w(x_i, c_j) = \|x_i\|\|c_{j,w}\|I(\ket{x_i}, |c_{j,w}\rangle)$.

    Now, as we had shown in Theorem~\ref{thmapp:pe}, instead of directly measuring the third register, we can apply amplitude estimation procedure as method in Appendix~\ref{sec:AE} followed by majority evaluation procedure in Appendix~\ref{sec:ME} to allow us to estimate the distance $|\overline{I(\ket{x_i},|c_{j,w}\rangle)} - I(\ket{x_i}, |c_{j,w}\rangle)| \leq \epsilon_1$ with probability at least $1 - 2\Delta$.

    The core idea is that after swapping the third and fourth registers of the state in Eq~\ref{eq:centroid_sup}, it can be rewritten as
    \begin{equation}
        U : \ket{i}\ket{j}\ket{0}\ket{0} \rightarrow \ket{i}\ket{j}\left(\sqrt{p(1)}\ket{z_{ij}, 1} + \sqrt{1 - p(1)}\ket{G,0}\right),
        \label{eq:CDAE}
    \end{equation}
    where $G$ is some garbage state.  We note that the above operation takes time $\mathcal{O}(\emph{poly} \log(kd))$ (for preparation of states $|c_{j,w}\rangle$) and $T = \mathcal{O}(\emph{poly} \log(Nd))$ for preparing the states $\ket{x_i}$. Thus the total time taken to perform the operation in Eq~\ref{eq:CDAE} is $$T_U = \mathcal{O}(\emph{poly} \log(kd)) + T \approx T$$ given that $N$ is much larger than $k$.

    Next, after applying the amplitude estimation (prior to measurement), we obtain the state
    \begin{equation}
        \ket{i}\ket{j}\left(\sqrt{\alpha}|\text{arcsin}\sqrt{\overline{p(1)}}\rangle\ket{G', 1} + \sqrt{1 - \alpha}\ket{G'',0}\right),
    \end{equation}
    where $G', G''$ are garbage registers and $\overline{p(1)}$ is the $\epsilon_1$ approximation to $p(1)$. 
    
    Subsequently, on the third register, we apply the quantum non-linear operation subroutines $\ket{a}\ket{0} \rightarrow \ket{a}\ket{\sin(a)}$ followed by $\ket{\sin(a)}\ket{0} \rightarrow \ket{\sin(a)}|\sin^2(a)\rangle$ as mentioned in Sec~\ref{sec:subroutines}. This is followed by applying the quantum adder subroutine, again mentioned in Sec~\ref{sec:subroutines}, to apply $|\sin^2(a)\rangle\ket{1} \rightarrow |\sin^2(a)\rangle|1- 2\sin^2(a)\rangle$. Here $a = \arcsin \overline{p(1)}$. This results in the state
    \begin{equation}
        \ket{i}\ket{j}\left(\sqrt{\alpha}|1 -2\overline{p(1)}\rangle|G', 1\rangle + \sqrt{1 - \alpha}|G'',0\rangle\right) := \ket{i}\ket{j}\left(\sqrt{\alpha}|\overline{I(\ket{x_i}, |c_{j,w}\rangle)}\rangle |G', 1\rangle + \sqrt{1 - \alpha}\ket{G'',0}\right).
    \end{equation}

    Subsequently, we multiply the contents of the third register by $\|x\|\|c_{j,w}\|$ to obtain the state (note that this can be done using the quantum multiplier subroutine mentioned in Sec~\ref{sec:subroutines})
    \begin{equation}
       \ket{i}\ket{j}\left(\sqrt{\alpha}|\overline{I_w(x_i, c_j)}\rangle |G', 1\rangle + \sqrt{1 - \alpha} |G'',0\rangle\right).
       \label{Eq:AEInner}
    \end{equation}
    The above procedure of amplitude estimation followed by the quantum subroutines takes time $\mathcal{O}(T\frac{\|x_w\|\|c_j\|}{\epsilon_1})$, where the extra factor $\|x_i\|\|c_{j,w}\|$ comes in because the amplitude estimation allows for the estimation of unnormalized values $|\overline{I(\ket{x_i}, |c_{j,w}\rangle)} - I(\ket{x_i}, |c_{j,w}\rangle)| \leq \epsilon_1$. In order to get the estimation of the normalized values and still keep the error bound to $\epsilon_1$, the time complexity gets multiplied by the factor $\|x_i\|\|c_{j,w}\|$.
    
    We can now boost the success probability of obtaining the inner product state from the value $\alpha$ to any desired value $1 - 2\Delta$ by applying the majority evaluation procedure in Sec~\ref{sec:ME} using $L = \mathcal{O}(\ln(1/\Delta))$ copies of state in Eq~\ref{Eq:AEInner}. This results in the state
    \begin{equation}
        \|\ket{\Psi}_{ij} - \ket{0}^{\otimes L}|\overline{I_w(x_i, c_j)}\rangle |G', 1\rangle\| \leq \sqrt{2\Delta}.
    \end{equation}
    The run time of this procedure is $\mathcal{O}(\frac{T\|x_i\|\|c_{j,w}\|}{\epsilon_1} \ln(1/\Delta))$.

    With this, the proof of the theorem follows straightforwardly by noting that the norm of the weighted centroids $\|c_{j,w}\| \leq \|w\|\|c_j\|$. Given that $\|w\| = 1$, this implies that the norm of the weighted centroids would be less than the maximum norm of the $k$ original centroids. Further, the norm of the centroid would always be less than the maximum norm of all the weighted examples in the dataset. This implies that $\|x_i\|\|c_{j,w}\| \leq \text{max}_i(\|x_i\|^2) = \eta_1$.
\end{proof}

\section{Proof of Theorem~\ref{thm:cent_update}} \label{app:cent_update}

\begin{theorem}
    Given access to the characteristic vector states $|\xi_j\rangle$ $\forall j \in [k]$ where each state is prepared in time $T_\xi = \mathcal{O}(T_l\log k)$ and the amplitude-encoded states for feature vectors $|x^{(l)}\rangle$ $\forall l \in [d]$ which can be prepared in time $T = \mathcal{O}(\text{poly} \log(Nd))$, there exists a quantum algorithm to obtain the updated centroid vectors $\overline{c_1}, \cdots, \overline{c_k}$ such that $\|\overline{c_j} - c_j\|_\infty \leq \epsilon_2$ $\forall j \in [k]$ in time 
    $$T_{\text{sup-cluster}} =\mathcal{O}\left(\frac{T_\xi k^\frac{3}{2}} d {\epsilon_2}\right),$$ 
    where $c_j = X^T\xi_j$ is the true mean of the weighted examples in the cluster and thus the true updated centroid vector of $c_j$ and $\eta_2 = \text{max}_{l \in [d]}\|x^{(l)}\|$. 
\end{theorem}

\begin{proof}
    We first highlight the method and time complexity for estimating the centroid vector $c_j$ for a given $j \in [k]$. The total time complexity for the rest of $k-1$ centroids is just repeating the above procedure $k$ times for different characteristic vector states.  

    Using Lemma~\ref{lemma:supcol}, we can query the following data feature state in time $T = \mathcal{O}(\emph{poly} \log(Nd))$
    \begin{equation}
        \ket{l}\ket{0} \rightarrow \ket{l}|x^{(l)}\rangle.
    \end{equation}
    Also in time $T_\xi = \mathcal{O}(T_l\log k)$, we obtain the state $|\xi_j\rangle$. Since the centroid vector $c_j$ is $d$-dimensional, our approach is to estimate each component $c_{jl}$, $l \in [d]$ separately within $\epsilon_2$ precision. This would allow us to estimate the centroid vector such that $\|\overline{c_j} - c_j\|_\infty \leq \epsilon_2$. 

    Before describing our method, it can be seen that the $l$-th component of the centroid is $c_{jl} = \frac{1}{|\mathcal{C}_j|}\sum_{i \in \mathcal{C}_i}x_{il}$. Now, this value can be computed from the inner product of the quantum states $|x^{(l)}\rangle$ and $|\xi_j\rangle$ as
    \begin{equation}
        I(|x^{(l)}\rangle, |\xi_j\rangle) = \langle x^{(l)}|\xi_j\rangle = \frac{1}{\sqrt{|\mathcal{C}_j|}\|x^{(l)}\|} \sum_{i \in \mathcal{C}_j} x_{il} = \frac{\sqrt{|\mathcal{C}_j|}}{\|x^{(l)}\|} c_{jl}.
    \end{equation}
    
    Thus, the estimation of the quantity $c_{jl}$ amounts to estimating the inner product of the two quantum states mentioned above. We now start with the initial state
    \begin{equation}
        \ket{\phi_l} = \ket{l}\ket{j}\frac{1}{\sqrt{2}}(\ket{0} + \ket{1})\ket{0}
    \end{equation}

    Then the KP-tree is queried controlled on the third register which results in the mappings $\ket{l}\ket{j}\ket{0}\ket{0} \rightarrow \ket{l}\ket{j}\ket{0}|x^{(l)}\rangle$ and $\ket{l}\ket{j}\ket{1}\ket{0} \rightarrow \ket{l}\ket{j}\ket{1}|\xi_j\rangle$. Thus the state after this controlled rotation operation is given by
    \begin{equation}
        \frac{1}{\sqrt{2}}\ket{l}\ket{j}\left(\ket{0}|x^{(l)}\rangle + \ket{1}|\xi_j\rangle\right)
    \end{equation}

    Applying Hadamard operation on the third qubit results in the state
    \begin{equation}
        \frac{1}{2}\ket{l}\ket{j}\left(\ket{0}(|x^{(l)}\rangle + |\xi_j\rangle) + \ket{1}(|x^{(l)}\rangle - |\xi_j\rangle)\right)
        \label{eq:evolve_state_cent}
    \end{equation}

    Now, the state $\ket{1}(|x^{(l)}\rangle - |\xi_j\rangle)$ can be rewritten as $|z_{lj}, 1\rangle$ (by swapping the registers), and hence Eq~\ref{eq:evolve_state_cent} has the following mapping
    \begin{equation}
        \ket{l}\ket{j}\ket{0}\ket{0} \rightarrow \ket{l}\ket{j}\left(\sqrt{p(1)}|z_{lj}, 1\rangle + \sqrt{1 - p(1)}\ket{G,0}\right),
    \end{equation}
    where $G$ is some garbage state and $p(1)$ is the probability of obtaining outcome 1 when the third register of the state in Eq~\ref{eq:evolve_state_cent} is measured i.e.,
    \begin{equation}
        p(1) = \frac{1}{2}\left(1 -  I(|x^{(l)}\rangle, |\xi_j\rangle)\right).
        \label{eq:prob}
    \end{equation} 

    Now it is clear that this is the form of the input for amplitude estimation \cite{brassard2002quantum} algorithm where the task is to estimate the unknown coefficient $\sqrt{p(1)}$. Applying amplitude estimation results in the output state (before measurement)
    \begin{equation}
        U: \ket{l}\ket{j}\ket{0} \rightarrow \ket{l}\ket{j}\left(\sqrt{\alpha}|\overline{p(1)}, G', 1\rangle + \sqrt{1 - \alpha}|G^{''},0\rangle\right),
    \end{equation}
    where $G', G^{''}$ are garbage registers. The above procedure requires $\mathcal{O}(1/\epsilon_2)$ iterations of the unitary $U$ (and its transpose) to produce the state such that $|\overline{p(1)} - p(1)| \leq \epsilon_2$. Measuring the above state results in the estimation of $\overline{p(1)}$ with a constant probability $\alpha \geq 8/\pi^2$. 

    From this, it becomes clear that we can also get an $\epsilon_2$ estimate on the inner product with $\mathcal{O}(1/\epsilon_2)$ iterations of $U$. This results in the estimation of $c_{jl}$ with accuracy
    \begin{equation}
        |\overline{c_{jl}} - c_{jl}| \leq \frac{\|x^{(l)}\|}{\sqrt{|\mathcal{C}_j|}} \epsilon_2 = \epsilon'_2.
    \end{equation} 
    Denoting $\epsilon'_2$ as $\epsilon_2$, the total time required to estimate the value $c_{jl}$ is the time to load the states $|x^{(l)}\rangle$ and $|\xi_j\rangle$ and the subsequent time to perform the inner product estimation between them
    \begin{equation}
    \begin{split}
        T_{c_{jl}} &= \mathcal{O}\left(\frac{(T_\xi + T)\|x^{(l)}\|}{\sqrt{|\mathcal{C}_j|}\epsilon_2}\right) \\
        &= \mathcal{O}\left(\frac{T_\xi\|x^{(l)}\|}{\sqrt{|\mathcal{C}_j|}\epsilon_2}\right)\\
       &= \mathcal{O}\left(\frac{T_\xi \sqrt{N}}{\sqrt{\frac{N}{k}} \epsilon_2}\right)\\
       &= \mathcal{O}\left(\frac{T_\xi \sqrt{k}}{\epsilon_2}\right),
    \end{split} 
    \end{equation}
    where we use the assumption that $X$ has bounded constants entries and therefore $\|x^{(l)}\| = \mathcal{O} (\|x^{(l)}\|_{\infty}\sqrt{N}) = \mathcal{O}(\sqrt{N})$ and the assumption of well-clusterable datasets introduced by Kerenidis \etal \cite{kerenidis2019q} that each cluster has at
least $\Omega(\frac{N}{k})$ samples, i.e., $|\mathcal{C}_j| = \Omega(\frac{N}{k})$. Note this assumption is crucial to remove the dependency with $\sqrt{N}$ that comes from the linear scaling with $\|x^{(l)}\|$.

    Now to calculate the other elements $c_{jl'}$ of the vector $c_j$, we repeat the inner product estimation procedure between $\ket{x^{(l')}}$ and $\ket{\xi_j}$ which requires the time complexity $\mathcal{O}\left(\frac{T_\xi \sqrt{k}}{\epsilon_2}\right)$. Thus the total time taken to estimate the vector $c_j$ can be bounded by $T_{c_j} = \mathcal{O}\left(\frac{T_\xi d \sqrt{k}}{\epsilon_2}\right)$.

    We repeat the above procedure $k$ times to get all the updated centroid vectors $c_1, \cdots, c_k$ within infinity norm $\epsilon_2$ precision with the total time complexity
    \begin{equation}
        T_{\text{sup-cluster}} = \sum_{j \in [k]} T_{c_j} = \mathcal{O}\left(\frac{T_\xi k^\frac{3}{2}}d {\epsilon_2}\right),
    \end{equation}
    where sup-cluster stands for supervised clustering. This completes the proof. 
\end{proof}

\section{Proof of Theorem~\ref{thm:cent_update2}} \label{app:cent_update2}

\begin{theorem}
    Given access to the characteristic vector states $|\xi_j\rangle$ $\forall j \in [k^D]$ where each state is prepared in time $T_\xi = \mathcal{O}(T_{l}^{k^D} D\log k)$ and the amplitude-encoded states label superposition state $|Y\rangle$ which is be prepared in time $\mathcal{O}(\text{poly} \log(N))$, there exists a quantum algorithm to obtain the mean label values $\{\overline{\text{label}_1}, \cdots, \overline{\text{label}_{k^D}}\}$ such that $|\overline{\text{label}_j} - \text{label}_j| \leq \epsilon_3$, $\forall j \in [k^D]$ in time 
    $$T_{\text{leaf-label}} =\mathcal{O}\left(\frac{T_\xi k^\frac{5D}{2}}{\epsilon_3}\right),$$ 
    where $\text{label}_j$ is the true label mean of the examples in the cluster as given in Eq~\ref{eq:mean-label}.  
\end{theorem}
\begin{proof}
 We start with the initial state
    \begin{equation}
        \ket{\phi_l} = \ket{j}\frac{1}{\sqrt{2}}(\ket{0} + \ket{1})\ket{0}.
    \end{equation}

    We query the states $|\xi_j\rangle$ and $|Y\rangle$ with the index $\ket{j}$ controlled on the second register which results in the mappings $\ket{j}\ket{0}\ket{0} \rightarrow \ket{j}\ket{0}|Y\rangle$ and $\ket{j}\ket{1}\ket{0} \rightarrow \ket{j}\ket{1}|\xi_j\rangle$. Thus the state after this controlled rotation operation is given by
    \begin{equation}
        \frac{1}{\sqrt{2}}\ket{j}\left(\ket{0}|Y\rangle + \ket{1}|\xi_j\rangle\right).
    \end{equation}

    Applying Hadamard operation on the third qubit results in the state
    \begin{equation}
        \frac{1}{2}\ket{j}\left(\ket{0}(|Y\rangle + |\xi_j\rangle) + \ket{1}(|Y\rangle - |\xi_j\rangle)\right).
        \label{eq:evolve_state_y}
    \end{equation}

    Now, the state $\ket{1}(|Y\rangle - |\xi_j\rangle)$ can be rewritten as $|z_{lj}, 1\rangle$ (by swapping the registers), and hence Eq~\ref{eq:evolve_state_y} has the following mapping
    \begin{equation}
        \ket{j}\ket{0}\ket{0} \rightarrow \ket{j}\left(\sqrt{p(1)}|z_{lj}, 1\rangle + \sqrt{1 - p(1)}\ket{G,0}\right),
    \end{equation}
    where $G$ is some garbage state and $p(1)$ is the probability of obtaining outcome 1 when the third register of the state in Eq~\ref{eq:evolve_state_y} is measured i.e.,
    \begin{equation}
        p(1) = \frac{1}{2}\left(1 -  I(|Y\rangle, |\xi_j\rangle)\right)
    \end{equation} 

    Now it is clear that this is the form of the input for amplitude estimation \cite{brassard2002quantum} algorithm where the task is to estimate the unknown coefficient $\sqrt{p(1)}$. Applying amplitude estimation results in the output state (before measurement)
    \begin{equation}
        U: \ket{j}\ket{0} \rightarrow \ket{j}\left(\sqrt{\alpha}|\overline{p(1)}, G', 1\rangle + \sqrt{1 - \alpha}|G^{''},0\rangle\right),
    \end{equation}
    where $G', G^{''}$ are garbage registers. The above procedure requires $\mathcal{O}(1/\epsilon_3)$ iterations of the unitary $U$ (and its transpose) to produce the state such that $|\overline{p(1)} - p(1)| \leq \epsilon_3$. Measuring the above state results in the estimation of $\overline{p(1)}$ with a constant probability $\alpha \geq 8/\pi^2$. 

    From this, it becomes clear that we can also get an $\epsilon_3$ estimate on the inner product with $\mathcal{O}(1/\epsilon_3)$ iterations of $U$. Using Eq~\ref{eq:label-inner}, this results in the estimation of $\emph{label}_j$ with accuracy
    \begin{equation}
        |\overline{\emph{label}_j} - \emph{label}_j| \leq \frac{\|Y\|}{\sqrt{|\mathcal{C}_{j,D}|}} \epsilon_3 = \epsilon'_3.
        \label{Eq:cov_cent}
    \end{equation} 
    Denoting $\epsilon'_3$ as $\epsilon_3$, the total time required to estimate the value $\emph{label}_j$ is the time to load the states $|Y\rangle$ and $|\xi_j\rangle$ and the subsequent time to perform the inner product estimation between them
    \begin{equation}
    \begin{split}
        T_{\text{label}_j} &= \mathcal{O}\left(\frac{(T_\xi + \mathcal{O}(\emph{poly} \log(N)))\|Y\|}{\sqrt{|\mathcal{C}_{j,D}|}\epsilon_3}\right) \\
        &= \mathcal{O}\left(\frac{T_\xi\|Y\|}{\sqrt{|\mathcal{C}_{j,D}|}\epsilon_3}\right) \\
        &= \mathcal{O}\left(\frac{T_\xi \sqrt{N}}{\sqrt{\frac{N}{k^D}}}\epsilon_3\right) \\ 
        &= \mathcal{O}\left(\frac{T_\xi k^\frac{D}{2}}{\epsilon_3}\right), 
    \end{split} 
    \end{equation}
    where we use the assumption that $Y$ has bounded constants entries and therefore $\|Y\| = \mathcal{O} (\|Y\|_{\infty}\sqrt{N}) = \mathcal{O}(\sqrt{N})$. We also use the assumption of well-clusterable data set which implies that $|\mathcal{C}_j| = \Omega(\frac{N}{k})$. Thus, $|\mathcal{C}_{j,D}|  = \Omega(\frac{N}{k^D})$. Note this assumption is crucial to remove the dependency with $\sqrt{N}$ that comes from the linear scaling with $\|Y\|$.

    Repeating the above procedure for the rest of the leaf nodes leads to the total time complexity
    \begin{equation}
    \begin{split}
     T_{\text{leaf-label}} &=\mathcal{O}\left(\frac{T_\xi k^\frac{3D}{2}}{\epsilon_3}\right) \\
        %&\mathcal{O}\left(\emph{poly} \log(Nd) \frac{ Dk^{3D}d \log (k)\log(p)\log(1/\Delta)\eta_1\eta_3}{\epsilon_1\epsilon_3}\right),
    &= \mathcal{O}\left(\emph{poly} \log(Nd) \frac{ D k^\frac{5D}{2}}d \log (k)\log(1/\Delta)\eta_1{\epsilon_1\epsilon_3}\right).
    \end{split}
        \label{eq:leaf_maj}
    \end{equation}
\end{proof}

\section{Numerical simulations} \label{appendix:details_numerics}
We use the $k$-means implementation in the library pyclustering \cite{Novikov2019}. This implementation contains a feature where the user can pass as an argument the distance metric to use. We initialized the centroids for k-means with $k$-means++ technique \cite{arthur2007k}. This is done at all the depths of the tree. For the Pearson correlation (or the point-biseral correlation), we use the r$\_$regression function defined sklearn.feature$\_$selection. We took the absolute values and then we normalized the vector. 
The maximum number of iterations taken for $k$-means algorithm for the cluster centroids convergence was set to $100$. However, it was seen that for the datasets considered, the convergence was achieved with much less iterations. Note that we used the Euclidean distance for the distance calculation, as defined in Eq \ref{ec_distance_algo}. For the regression task, we consider the dataset known as the Boston housing \cite{harrison1978hedonic}. We did some feature selection so as the features used were eight ('LSTAT', 'INDUS', 'NOX', 'PTRATIO', 'RM', 'TAX', 'DIS', 'AGE'). We also removed the skewness of the data through log transformation. 
Note that when we refer to the number of clusters ($k$) for a tree, we refer to the number of clusters that each node in the tree was split by the proposed supervised clustering method. In some cases, the samples could not be divided into the $k$ desired number of clusters because of the number of samples. That is why in the Table \ref{table:main_results_numerics}, Table \ref{table:main_results_classification} 
and Table \ref{table:main_results_regression} the values in the tree size column may not be integer values.
  
\section*{Disclaimer}

This paper was prepared for informational purposes by the Global Technology Applied Research center of JPMorgan Chase $\&$ Co. This paper is not a product of the Research Department of JPMorgan Chase $\&$ Co. or its affiliates. Neither JPMorgan Chase $\&$ Co. nor any of its affiliates makes any explicit or implied representation or warranty and none of them accept any liability in connection with this paper, including, without limitation, with respect to the completeness, accuracy, or reliability of the information contained herein and the potential legal, compliance, tax, or accounting effects thereof. This document is not intended as investment research or investment advice, or as a recommendation, offer, or solicitation for the purchase or sale of any security, financial instrument, financial product, or service, or to be used in any way for evaluating the merits of participating in any transaction.

\end{document}